\documentclass[a4paper,11pt, notitlepage]{article}
\usepackage{geometry}
\usepackage{graphicx}
\usepackage{setspace}
\usepackage{amsfonts}
\usepackage[bbgreekl]{mathbbol}
\DeclareSymbolFontAlphabet{\mathbb}{AMSb}
\usepackage{lscape}

\usepackage[english]{babel}
\usepackage{caption}
\usepackage{subcaption} 
\usepackage{mathtools}
\usepackage{mathbbol}
\usepackage{soul} %
\usepackage{enumitem}
 \usepackage[ruled]{algorithm2e}

\usepackage{amsmath}
\usepackage{graphicx}
\usepackage{graphics}
\usepackage{amssymb}
\usepackage{amsthm}
\usepackage{thmtools, thm-restate}
\usepackage{hyperref}

\newtheorem{theorem}{Theorem}[section]
\newtheorem{definition}[theorem]{Definition}
\newtheorem{lemma}[theorem]{Lemma}
\newtheorem{corollary}[theorem]{Corollary}
\newtheorem{algorithmm}[theorem]{Algorithm}

\usepackage{bbm}

\newcommand{\ceil}[1]{\left\lceil #1 \right\rceil}
\newcommand{\floor}[1]{\left\lfloor #1 \right\rfloor}
\newcommand{\union}{\cup}
\newcommand{\inter}{\cap}
\newcommand{\Union}{\bigcup}
\newcommand{\Inter}{\bigcap}
\newcommand{\card}[1]{\left|#1\right|}

\newcommand{\RT}{\mathcal{T}}
\newcommand{\Alg}{\mathcal{A}}
\newcommand{\N}{\mathbb{N}}
\newcommand{\R}{\mathbb{R}}
\newcommand{\Z}{\mathbb{Z}}
\newcommand{\Out}{\mathcal{N}}
\newcommand{\parts}{\mathcal{P}}
\newcommand{\Proba}{\mathbb{P}}
\newcommand{\expect}{\mathbb{E}}
\newcommand{\1}{\mathbbm{1}}

\renewcommand{\texorpdfstring}[1]{}
\newcommand{\F}{\mathcal F}

\author{Antoine El-Hayek\\
ISTA\\Austria
\and
Monika Henzinger\\
ISTA\\Austria
\and
Stefan Schmid\\
TU Berlin\\ Fraunhofer SIT \\Germany
}

\date{}

\begin{document}

\title{Time Complexity of Broadcast and Consensus for Randomized Oblivious Message Adversaries \thanks {This project has received funding from the European Research Council (ERC) under the European Union's Horizon 2020 research and innovation programme (MoDynStruct, No. 101019564)  \includegraphics[scale=0.4]{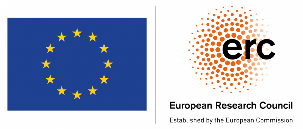} and the Austrian Science Fund (FWF) grant DOI 10.55776/Z422, grant DOI 10.55776/I5982, and grant DOI 10.55776/P33775 with additional funding from the netidee SCIENCE Stiftung, 2020–2024. This project has received further funding from the German Research Foundation (DFG), SPP 2378 (project ReNO), 2023-2027}}

\maketitle \sloppy 
\begin{abstract}
Broadcast and Consensus are most fundamental tasks in distributed computing. These tasks are particularly challenging in dynamic networks where communication across the network links may be unreliable, e.g., due to mobility or failures.  Over the last years, researchers have derived several impossibility results and high time complexity lower bounds  for these tasks.
Specifically for the setting where in each round  of communication the adversary is allowed to choose one rooted tree along which the information is disseminated,  there is a  lower as well as an upper bound that is linear in the number $n$  of nodes for Broadcast and for $n\ge 3$ the adversary can guarantee that Consensus never happens. This setting is called the \emph{oblivious message adversary for rooted trees}.
Also note that if the adversary is allowed to choose a graph that does not contain a rooted tree, then it can guarantee that Broadcast and Consensus will never happen.

However, such deterministic adversarial models may be overly pessimistic, 
 as many processes in real-world settings are stochastic in nature rather than worst-case.

This paper studies Broadcast on \emph{stochastic} dynamic networks and shows that the situation is very different to the deterministic case.
In particular, we show that if information dissemination occurs along random rooted trees and directed Erdős–Rényi graphs, Broadcast completes in $O(\log n)$ rounds of communication with high probability.
The fundamental insight in our analysis is that key variables are mutually independent.

We then study two adversarial models, (a) one with Byzantine nodes and (b) one where an adversary controls the edges. 
(a)~Our techniques without Byzantine nodes are general enough so that they can be extended to Byzantine nodes.
(b)~In the spirit of smoothed analysis, we introduce the notion of \emph{randomized oblivious message adversary}, where in each round, an adversary picks $k \le 2n/3$ edges to appear in the communication network, and then a graph (e.g. rooted tree or directed Erdős–Rényi graph) is chosen uniformly at random among the set of all such graphs that include these edges. We show that Broadcast completes in a finite number of rounds, which is, e.g., $O(k+\log n)$ rounds in rooted trees.

 We then extend these results to All-to-All Broadcast, and Consensus, and give lower bounds that show that most of our upper bounds are tight.

\end{abstract}
\section{Introduction}

Broadcast and Consensus are two of most fundamental operations in distributed computing which, in large-scale systems, typically have to be performed over a \emph{network}. These networks are likely to be dynamic and change over time due, e.g.,  to link failures, interference, or mobility.
Understanding how information disseminates in such dynamic networks is hence important for developing and analyzing efficient distributed systems. 

Over the last years, researchers have derived several important insights into information dissemination in dynamic networks. A natural and popular model assumes an  \emph{oblivious\footnote{Note that the term oblivious here refers to the property that nodes are oblivious to who their neighbors are. However, our adversary is actually adaptive.} message adversary} which controls the information flow between a set of $n$ nodes, by dropping an arbitrary set of messages sent by some nodes in each round~\cite{coulouma2015characterization}. Specifically, the adversary is defined by a set of directed communication graphs, one per round, whose edges determine which node can successfully send a message to which other node in a given round. 
Based on this set of graphs, the oblivious message adversary chooses a sequence of graphs over time, one per round with repetitions allowed, in such a way that the time complexity of the information dissemination task at hand is maximized. 
This model is appealing because it is conceptually simple and still provides a highly dynamic network model: The set of allowed graphs can be arbitrary, and the nodes that can communicate with one another can vary greatly from one round to the next. It is, thus, well-suited for settings where significant transient message loss occurs, such as in wireless networks.
As information dissemination is faster on dense networks, most literature studies oblivious message adversaries on sparse networks, in particular, on rooted trees~\cite{itcs23broadcast, schwarz2017linear, coulouma2015characterization,itcs23Consensus,galeana2023topological}. In fact, it is easy to see that rooted trees are a minimal necessary requirement for a successful Broadcast and Consensus: if an adversary may choose a graph that does not contain a rooted tree, then it may forever prevent the dissemination of a piece of information.

Unfortunately, information dissemination can be slow in trees: Broadcast can take time linear in the number of nodes under the oblivious message adversary~\cite{itcs23broadcast, schwarz2017linear}, even for constant-height trees (as we show in  Appendix~\ref{appendix:lowerbound}); 
and Consensus can even take super-polynomial time until termination, if it completes at all~\cite{coulouma2015characterization,itcs23Consensus}. 
Although this is bad news, one may argue that while the deterministic adversary model is useful in malicious environments, in real-word applications, the dynamics of communication networks is often more stochastic in nature. Accordingly, the worst-case model considered in existing literature may be overly conservative.

This motivates us, in this paper, to  study information dissemination, and in particular  Broadcast and Consensus tasks, in a scenario where the communication network is stochastic. Initially, we study a purely stochastic scenario where in each round, the communication network is chosen uniformly at random among all rooted trees. We then study several fundamental 
extensions of this model where the adversary has some limited control. In a first extension, we consider the case where some nodes (up to $\frac {2n} 3$) may be Byzantine, that is, they may deviate arbitrarily from the protocol (and stop forwarding messages, for example). 
In a second extension, in the spirit of smoothed analysis, we study a setting where an adversary has some limited control over the communication network; we call this adversary the {\em randomized oblivious message adversary}. 
More specifically, we study the setting where first a worst-case adversary chooses $k$ directed edges  in the dynamic $n$-node network for some fixed $k$ with $0 \le k < \frac {2n} 3 -1$\footnote{We can relax this condition to $k \le (1-\epsilon)n$ for a fixed parameter $\epsilon$, which results in a multiplicative factor of $\frac 1 \epsilon$ in the running time.}, and then a rooted tree is chosen uniformly at random among the set of all rooted trees that include these edges.

We show that Broadcast completes within time $O(\log n)$ with high probability. 
We then show that this result even holds with Byzantine nodes. Under our randomized oblivious message adversary, Broadcast completes in $O(k + \log n)$ time with high probability. 

It is useful to put our model into perspective with the  SI (Susceptible-Infectious) model in epidemics \cite{durrett2022susceptible}: while in the SI model interactions occur on a network that equals a clique, our model revolves around trees which are chosen by an adversary. This tree structure renders the analytical understanding of the information dissemination process harder, due to the lack of independence between the edges in the network in a particular round. 
A key insight from our paper is that we can  prove the independence of a key variable, namely the increase in the number of ``informed'' nodes, which is crucial for our analysis. 
Our proof further relies on stochastic dominance, which makes it robust to the specific adversarial objective, and applies to any adversary definition (e.g., whether it aims to maximize the minimum or the expected number of rounds until the process completes). %

We then extend our study to adversaries which are not limited to trees. In particular, we are interested in how the time complexity of Broadcast and Consensus depends on the density of the network. To this end, we consider \emph{directed Erdős–Rényi graphs}, a directed version of the classic and well-studied random graphs. This graph family is parameterized by the number of edges $m$ and hence allows us to shed light on the impact of the density. 
Specifically in this model, in each round the network is formed by sampling $m$ edges. We again study two extensions: in the first extension some nodes behave as Byzantine nodes, while in the second extension, up to $k \le m$ edges are chosen by an adversary, and then the remaining edges are sampled.
While results for this model can be found in some cases where $m$ is chosen so that the graph is an expander w.h.p. in each round by using the results from Augustine et al~\cite{DBLP:conf/soda/AugustinePRU12}, in the case where $m$ is small, our results are novel.

We show that all our results extend to multiple other problems, namely All-to-All Broadcast, Byzantine Consensus and Reliable Broadcast.

\subsection{Model} 
Let $n$ be the number of nodes, and let each node %
have a unique identifier from $[n]$. Time proceeds in a sequence of rounds $t = 1, 2, \dots, $ such that in each round $t$ the communication network is chosen according to one of the models defined below. In each round, every honest node sends a message to all of its out-neighbors before receiving one from its in-neighbor. There is no message size restriction. We will study the following models of communication:

\paragraph*{Uniformly Random Trees.} In the \emph{Uniformly Random Trees} model, let $\RT_n$ be the set of all directed rooted trees on $n$ nodes (where all edges are pointed away from the root).  In each round, the communication network is chosen uniformly at random among graphs in $\RT_n$, independently from other rounds. All nodes are honest.

\paragraph*{Uniformly Random Trees with Byzantine Nodes.} In the \emph{Uniformly Random Trees with Byzantine Nodes} model, in each round, the communication network is chosen uniformly at random among graphs in $\RT_n$, independently from other rounds. We have $n-f$ honest nodes, and $f$ nodes are Byzantine, that is, they might behave arbitrarily (and even coordinate to make the protocol fail). We assume access to cryptographic tools that allow nodes to sign and encrypt messages. We restrict $f \le \frac {2n} 3 -1$. 

\paragraph*{Uniformly Random Trees with Adversarial Edges.} In the \emph{Uniformly Random Trees with Adversarial Edges} model, in each round, the communication network is chosen as follows: A randomized oblivious message adversary chooses $k$ directed edges, then a graph is chosen uniformly at random among all graphs in $\RT_n$ that include those $k$ edges, and the choise is independent from other rounds. All nodes are honest. We restrict $k \le \frac {2n} 3 -1$.

\paragraph*{Directed Erdős–Rényi graphs.} In the \emph{directed Erdős–Rényi graphs} model, let $m \in [n^2]$. In each round, the communication network is chosen by uniformly sampling without replacement $m$ edges out of the possible $n^2$ edges of the graph, independently from other rounds. All nodes are honest. 

\paragraph*{Directed Erdős–Rényi graphs with Byzantine Nodes.} In the \emph{directed Erdős–Rényi graphs with Byzantine nodes} model, let $m \in [n^2]$. In each round, the communication network is chosen by uniformly sampling without replacement $m$ edges out of the possible $n^2$ edges of the graph, independently from other rounds. We have $n-k$ honest nodes, and $k$ nodes are Byzantine, that is, they might behave arbitrarily (and even coordinate to make the protocol fail). We assume access to cryptographic tools that allow nodes to sign and encrypt messages. We restrict $k < \frac {2n} 3$.

\paragraph*{Directed Erdős–Rényi graphs with Adversarial Edges.} In the \emph{directed Erdős–Rényi graphs with Adversarial Edges} model, let $0 \le k \le m \le n^2$.  In each round, the communication network is chosen as follows: A randomized oblivious message adversary chooses $k$ edges, $m-k$ edges are sampled without replacement out of the remaining $n^2 - k$ edges. All nodes are honest. We restrict $k < \frac {3} 4 n^2$.

In those models, we will study the following problems:

\paragraph*{Broadcast.}
For the \emph{Broadcast}\footnote{The Broadcast problem can also be seen as computing the \emph{dynamic eccentricity} of the source node. Other flavors of Broadcast have also been studied under the name \emph{dynamic radius}~\cite{fugger2020radius}. } problem, we start by giving a message to \emph{one} (honest) node. Each honest node that received the message will replicate it as many times as needed, and start forwarding it to its neighbors\footnote{This is known as "flooding" or "rumor passing"}. Then Broadcast \emph{completes} when the message has been forwarded to all other nodes. 

\paragraph*{All-to-All Broadcast.} In the \emph{All-to-All Broadcast} problem, we start by giving a distinct message to \emph{each} node. Each honest node that received a message will replicate it as many times as needed, and start forwarding it as well. Then All-to-All Broadcast \emph{completes} when each honest node receives a copy of every message. In each round, each honest node forwards all the messages it has received in previous rounds to all its out-neighbors. 

\paragraph*{Consensus.}

In the \emph{Consensus} problem, we start by giving a value $v_p \in \{0, 1\}$ to each node $p$, and Consensus completes when each honest node decided %
on a value in $\{0,1\}$. This should satisfy the following conditions:

\begin{itemize}
    \item \textbf{Agreement:} No two honest nodes decide differently.
    \item \textbf{Termination:} Every honest node eventually decides.
    \item \textbf{Validity:} The value the honest nodes agree on should be one of the input values $v_p$.
\end{itemize}

\subsection{Our Results}

We study Broadcast in the above mentioned models, then apply those results to All-to-All broadcast and Consensus. We prove the following theorems:

\begin{restatable}{theorem}{Broadcast}\label{thm:Broadcast}
For any $c\ge 1$ and $n \ge 5$, Broadcast on Uniformly Random Trees completes within $32\cdot c \cdot \ln n$ rounds with probability $p>1-\frac 1 {n^c}$.    
\end{restatable}

We also show that these results are asymptotically tight. Indeed, we cannot hope for a similar probability for a number of rounds that is $o(\ln n)$:

\begin{restatable}{theorem}{lowerbounded}
If $n\ge 2$, then the probability that Broadcast (and All-to-All Broadcast) on Uniformly Random Trees fails to complete within $\log n$ rounds is at least $\frac 1 4 $.
\end{restatable}

We have similar results for all the combinations of model and problem, which we summarize in Table~\ref{tab:results}.

\begin{figure}
    \centering
    \scriptsize
\begin{center}
\resizebox{\linewidth}{!}{
\begin{tabular}{|c |c |c |c|}
	\hline
	             & Broadcast                                                                     & All-to-All Broadcast                                                        & Consensus                                                                   \\ \hline
	 Uniformly   &                                                                               &                                                                             &                                                                             \\
	   Random    & $O(c \cdot \log n), q \le n^{-c}$                                             & $O(c \cdot \log n), q \le n^{1-c}$                                          & $O(c \cdot \log n), q \le n^{-c}$                                           \\
	Trees (URT)  & $\Omega(\log n)$                                                              & $\Omega(\log n)$                                                            &                                                            \\ \hline
	  URT with   &                                                                               &                                                                             &                                                                             \\
	 Byzantine   & $O(c \cdot \log n), q \le n^{-c}$                                             & $O(c \cdot \log n), q \le n^{1-c}$                                          & $O(f\cdot c \cdot \log n), q \le n^{-c}$                                           \\
	   Nodes     & $\Omega(\log n)$                                                              & $\Omega(\log n)$                                                            &                                       \\ \hline
	  URT with   &                                                                               &                                                                             &                                                                             \\
	Adversarial  & $O(c \cdot (\log n+k)), q \le n^{-c}$                                         & $O(c \cdot (\log n+k)), q \le n^{1-c}$                                      & $O(c \cdot  (\log n+k)), q \le n^{-c}$                                      \\
	   Edges     & $\Omega(\log n+k)$                                                            & $\Omega(\log n + k)$                                                        &                                           \\ \hline
	  Directed   & $O\left(\ceil{\frac{c}{m/n}} \log n\right)$,$q \le n^{-c}\log n$              & $O\left(\ceil{\frac{c}{m/n}} \log n\right)$, $q \le n^{1-c}\log n$          & $O\left(\ceil{\frac{c}{m/n}} \log n\right)$, $q \le n^{-c}\log n$           \\
	Erdős–Rényi  & $O\left(\frac{c\log n}{\log (1+\frac m n)}\right)$ 
 if $\frac m n \ge \ln n$ & $O\left(\frac{c\log n}{\log (1+\frac m n)}\right)$ if $\frac m n \ge \ln n$ & $O\left(\frac{c\log n}{\log (1+\frac m n)}\right)$ if $\frac m n \ge \ln n$ \\
	             & with $q \le n^{-c}\log n$                                                     & with $q \le n^{1-c}\log n$                                                  & with $q \le n^{-c}\log n$                                                   \\
	graphs (DER) & $\Omega\left(\frac{\log n}{\log (1+m/n)}\right)$                              & $\Omega\left(\frac{\log n}{\log (1+m/n)}\right)$                            &                          \\ \hline
	  DER with   &                                                                               &                                                                             &                                                                             \\
	 Byzantine   & $O\left(\ceil{\frac{c}{m/n}} \log n\right)$, $q \le n^{-c}\log n$             & $O\left(\ceil{\frac{c}{m/n}} \log n\right)$, $q \le n^{1-c}\log n$          & $O\left(f\cdot\ceil{\frac{c}{m/n}} \log n\right)$, $q \le n^{-c}\log n$           \\
	   Nodes     & $\Omega\left(\frac{\log n}{\log (1+m/n)}\right)$                              & $\Omega\left(\frac{\log n}{\log (1+m/n)}\right)$                            &                     \\ \hline
	   DER with   &  $O\left(\ceil{\frac{c \cdot (n^2-k)}{(m-k)n}} \log n\right)$            &  $O\left(\ceil{\frac{c \cdot (n^2-k)}{(m-k)n}} \log n\right)$           &   $O\left(\ceil{\frac{c \cdot (n^2-k)}{(m-k)n}} \log n\right)$           \\
	Adversarial  &    with                                           $q \le n^{-c}\log n$        & with $    q \le n^{1-c}\log n$                                              & with $    q \le n^{-c}\log n$                                                 \\
	   Edges     & $\Omega\left(\frac{\log n}{\log (1+m/n)}\right)$                              & $\Omega\left(\frac{\log n}{\log (1+m/n)}\right)$                            &                              \\ \hline
\end{tabular}}
\end{center}
\normalsize
    \caption{Our main results, where $c >0$ is any constant and $q$ is the failure probability.}
    \label{tab:results}
\end{figure}

\paragraph*{Applications.}
Our results have some interesting applications.
In an idea similar to Ghaffari, Kuhn and Su's work~\cite{simulation},
All-to-All Broadcast allows us, e.g.,~to implement algorithms that run on a clique in a synchronous setting in our sparser graphs.
Indeed, if All-to-All Broadcast needs $R$ rounds to complete with high probability, then each round of communication of a clique can be simulated by $R$ rounds of Uniformly Random Trees with high probability.
Essentially, if an algorithm runs in $T$ rounds, with $T \le n^{c-1}$, %
in a clique network, we can implement it with high probability in $R\cdot T$ rounds in the Uniformly Random Trees network, which is essentially a logarithmic overhead. In particular, in the Uniformly Random Trees with Byzantine Nodes model, we have:

\begin{restatable}{theorem}{cliquereduc}\label{sec.5.thm:reduction}
    Let $\Alg$ be a distributed synchronous algorithm that runs on a static clique in $T$ rounds, where $T \le \alpha n^x$ for some constant $\alpha, x\in \R_+$, and has a probability of success $p$. Assume $\Alg$ is robust to $f$ Byzantine nodes, and $f \le \frac 2 3 n -1$.
    Then, assuming standard cryptographic tools\footnote{Specifically, our approach requires authenticated messages. Encryption may also be needed, only if the protocol $\Alg$ is vulnerable to eavesdropping. Both can be implemented using standard cryptographic tools.},  there exists a distributed algorithm $\Alg'$ that runs on Uniformly Random Trees in $T \cdot 144\cdot \log n \cdot c$ rounds, and has a probability of success $p' \ge p(1-\alpha n^{1+x-c})$, for any $c\ge 1+x$. Moreover, $\Alg'$ is robust to $f$ Byzantine nodes.
\end{restatable}

In particular, we can apply known results on reliable Broadcast and Byzantine Consensus to show the following results:

\begin{restatable}{corollary}{reliableBroadcast}
    For any $c\ge 1$, and $f \le \frac 2 3 n -1$, in the Uniformly Random Trees with $f$ Byzantine nodes, there exists an algorithm for Reliable Broadcast, that is robust to $f$ Byzantine nodes, that runs in $(f+1)\cdot 144 \cdot c \cdot \log n$ rounds, and succeeds with probability $p \ge 1-n^{2-c}$.
\end{restatable}

\begin{restatable}{corollary}{ConsensusByzantine}
    For any $c\ge 1$ and $f < \frac n 3 $, in the Uniformly Random Trees with $f$ Byzantine nodes, there exists an algorithm for Byzantine Consensus, that is robust to $f$ Byzantine nodes, that runs in $3(f+1)\cdot 144 \cdot c \cdot \log n$ rounds, and succeeds with probability $p \ge 1-2n^{2-c}$.
\end{restatable}

Throughout the paper, the filtration of the process is denoted as $\{\F_t\}_{t \in \N}$, that is, $\F_t$ is the amount of information available after timestep $t$.

\paragraph*{Organization}
The paper is organized as follows. First, we give a new result on counting rooted trees in Section~\ref{sec:crt}, which will be useful in our analysis. Afterwards, we explore the Uniformly Random Trees model in Section~\ref{sec:broadcast}. Then, in Section~\ref{sec:Byzantine}, we expand our analysis to the Uniformly Random Trees with Byzantine Nodes model. In Section~\ref{sec:adversary}, we explore the case where the adversary controls $k$ edges in each round. We study the directed Erdős–Rényi graphs model and its adversarial variants in Section~\ref{sec:erdos}. We review related work in Section~\ref{sec:relwork}. Appendix~\ref{appendix:lowerbound} gives a lower bound for deterministic Broadcast in constant-height trees. Appendix~\ref{sec:tree} gives the full details of Section~\ref{sec:crt}. In Appendix~\ref{appendix:prob}, we give some probability theory results that are useful throughout the paper. Finally, in Appendix~\ref{sec:apptrees} and~\ref{appendix:adversary}, we include omitted proofs from Sections~\ref{sec:broadcast} and~\ref{sec:adversary} respectively, while Appendix~\ref{app:Byzantine} and~\ref{app:erdos} give the full details of Sections~\ref{sec:Byzantine} and~\ref{sec:erdos}.

\section{Counting Rooted Trees}\label{sec:crt}

Given a graph consisting of $n$ vertices together with a  directed rooted forest $F$ of $e$ edges on them, Pitman~\cite{pitman1999coalescent} showed in 1999 that there are $n^{n-1-e}$ many directed rooted trees over these vertices that contain $F$. While useful, this result is not sufficient for our purposes as we need to count the number of trees with a given node $v$ as root.

Thus, we show the following extended result:

\begin{restatable}{theorem}{treecountroot}\label{thm:treecountroot}
    Let us be given a directed rooted forest $F$ on $n$ vertices, let $v \in [n]$ be the root of a component in $F$, and $f$ be the number of vertices of that component (note that we can have $f=1$ if $v$ is an isolated vertex). Then the number of directed rooted trees $T$ on $n$ vertices, such that $F$ is contained in $T$, and such that $v$ is the root of $T$, is $f n^{n-2-\card E}$. 
\end{restatable}

To show our result, we develop techniques which differ significantly from Pitman's proof. Indeed, Pitman relies on the symmetry of the vertices in the rooted tree. However, for our result, the symmetry is broken as one vertex is different from the others with the new requirement that it is the root. We hence make use of another type of symmetry in the trees in our analysis that is based on group actions.

We first ignore the orientations of the edges in $F$ and find the set $A_F$ of all undirected trees that contain $F$. We can compute the cardinality of that set with a result by Lu, Mohr and Székely~\cite{lu2012quest}. We then root each of those trees at $v$. This will give a direction to every edge that might or might not agree with its direction in $F$.  We now want to partition $A_F$ into subsets such that all subsets have the same size and only one tree from each subset has edges that agree with the direction of $F$. The number we are looking for is then the number of subsets, which is the ratio between the cardinality of $A_F$ and the size of the subsets.

To create the subsets, we introduce a specific group tailored to $F$, and an action of that group on $A_F$. It is known that the set of all orbits of the action partition $A_F$, and we show that exactly one element in each orbit has edges in the same direction as $F$. To see unicity, we take an element $T$ of $A_F$ that has edges in the same direction as $F$, and take an element $T' \neq T$ in its orbit, that is there exists a nontrivial group element $g$ such that $T'$ is obtained from $T$ by applying the action of $g$ to $T$. We show that this action must change the direction of at least one edge of $F$, and thus $T'$ does not have edges in the same direction as $F$. For existence, we show that for every $T \in A_F$, we can find a group element $g$ such that, if applied to $T$, yields a tree that  has edges in the same direction as $F$.
We then show how to compute the size of each orbit. This allows us to deduce the number of orbits, which equals the number of trees that we want to count.

The full details of the proof can be found in Appendix~\ref{sec:tree}.

\section{The Uniformly Random Trees Model}
\label{sec:broadcast}
We now give a precise description of how information flows in Uniformly Random Trees over time. In this section, we will use Theorem~\ref{thm:treecount}, which states that the number of rooted trees on $n$ nodes containing a given directed rooted forest $F$ with $e$ edges is $n^{n-1-e}$. Since all nodes are equivalent, we will at each step, divide the nodes into two sets: the set $I$ of nodes that have received the message, called \emph{informed} nodes, and the set $S$ of remaining nodes, called \emph{uninformed} nodes.
We study how $I$ grows over time.

For the rest of the section, $I_t$ and $S_t$ will, respectively, be the set of nodes that are informed and uninformed after round $t$. 
We set $I_0=\{v_0\}$ and $S_0=[n]-\{v_0\}$, where $v_0$ is the node that initially holds the message, $N_t=\card{I_t}$ to be the number of informed nodes after $t$ rounds, and $T_t$ to be the tree chosen at random in round $t$. 
For a tree $T$, for each node $p$, $P_T(p)$ is the (unique) parent of node $p$ in $T$, unless $p$ is the root of $T$, in which case $P_T(p)=p$. 
Simplifying the notation, we also use $P_t(p)$ to denote $P_{T_t}(p)$. All skipped proofs can be found in Appendix~\ref{sec:apptrees}.

The central claim of the proof is the following lemma, which characterizes how many new nodes get informed in each round, depending on how many were informed after the previous round. This lemma shows that uninformed nodes get informed independently from each other.
\begin{lemma} \label{lem:binom}
    For any $t > 0$,
    $N_{t+1}-N_t$ follows a binomial distribution with parameters $\left( n-N_t, \frac {N_t} n\right)$.
\end{lemma}

The proof of this lemma shows that every uninformed node has probability $\frac {N_t} n$ of having an informed parent in round $t+1$, independently of whether the other uninformed nodes have an uninformed parent.

\begin{proof}
    Let $I_t=\{i_1, \dots , i_{N_t}\}$ and $S_t=\{s_1, \dots , s_{n-N_t}\}$. We then have, for any integer $x$:

$$
        \Proba(N_{t+1}-N_t=x| \F_t) =\sum_{J\in {A(S_t, x)}}\Proba\left(\Inter_{y \in J} (P_{t+1}(y) \in I_t) \Inter_{y \in S_t\setminus J}(P_{t+1}(y) \notin I_t)\middle| \F_t \right),
$$
where
${A(S, x)}$ with $S$ being a set and $x$ an integer denotes the set of subsets of $S$ of size $x$.
Our goal is to show that the events $P_t(y) \in I_t$ for different 
$y \in S_t$ are mutually independent.
Let us look at the event $\Inter_{y \in J} (P_t(y) \in I_t)$ for any $J\subseteq S_t$ (note that we do not require that $J$ has a specific size here). 
We can then write, indexing $a$ on $J$:

\begin{multline*}
\Proba\left(\Inter_{y \in J} (P_{t+1}(y) \in I_t)\middle| \F_t\right)=\sum_{a\in[N_t]^{\card J}} \Proba\left(\Inter_{y \in J} (P_{t+1}(y) =i_{a_y})\middle| \F_t\right)\\= \sum_{a\in[N_t]^{\card J}} \frac{\card{\{ T\in \RT_n:P_T(y) =i_{a_y}, \forall y \in J\}}}{\card{\RT_n}}
\end{multline*}

Now consider the forest that is composed of stars whose centers are the $i_{a_y}$ and whose leaves are the nodes $y \in J$. More specifically, consider the forest that contains the edges
 $(i_{a_y}, y), \forall y \in J$. Note that $\card{\{ T\in \RT_n:P_T(y) =i_{a_y}, \forall y \in J\}}$ equals the number of rooted trees that are compatible with this forest.
By Theorem~\ref{thm:treecount}, we have that $\card{\{ T\in \RT_n:P_T(y) =i_{a_y}, \forall y \in J\}}=n^{n-1-\card{J}}$.
This allows us to compute the above probability as follows:

$$
\Proba\left(\Inter_{y \in J} (P_{t+1}(y) \in I_t)\middle| \F_t\right)= \sum_{a\in[N_t]^{\card J}}\frac{n^{n-1-\card{J}}}{n^{n-1}}=\left(\frac {N_t} n\right)^{\card J}
$$

This proves that the events $P_{t+1}(y) \in I_t$ for any two $y \in S_t$ are mutually independent (Definition~\ref{def:mutually}),
each having probability $\frac {N_t} n$. Going back to the first equation of this proof, we can now compute with Lemma~\ref{lem:mutually} and using $\Delta_t := N_{t+1}-N_t$ as a shorthand:

\begin{align*}
        \Proba(\Delta_t=x | \F_t) &=\sum_{J\in {A(S_t, x)}}\prod_{y \in J}\Proba\left(P_{t+1}(y) \in I_t\middle| \F_t\right)\prod_{y \in S_t\setminus J}\Proba \left(P_{t+1}(y) \notin I_t) \middle| \F_t\right) \\&={{n-N_t} \choose x} \left(\frac {N_t} n\right)^{x}\left(1-\frac {N_t} n\right)^{n-N_t-x}
\end{align*}
\end{proof}

Our next goal is to show that $N_t = n$ with high probability for all $t \ge 32\cdot c\cdot \ln n$.
To do so we introduce a random variable $X_t$ that we use to lower bound $N_t$.
\begin{definition}
    Let $X_t$ be the random variable defined recursively: $X_0=1$, and 
    \small
    \begin{align*}
        X_{t+1}&=X_t+(n-X_t) \cdot \frac {X_t} n &\text{if} \quad N_{t+1}-N_t \geq (n-N_t) \cdot \frac {N_t} n\\
        X_{t+1}&=X_t & \text{if} \quad N_{t+1}-N_t < (n-N_t) \cdot \frac {N_t} n
    \end{align*}
    \normalsize
\end{definition}

Intuitively, $X_t$ is a lower bound for $N_t$ that increases if and only if $N_{t+1}-N_t$ exceeds its expectation. Therefore, we always have $n \ge N_t \ge X_t \ge 1$ (full proof in Appendix~\ref{sec:apptrees}), and we can claim that $N_t = n$ as soon as $X_t > n-1$. Moreover, we can compute the values of $X_t$ after each increase:

\begin{restatable}{lemma}{valuesofX}
\label{lem:valueofX}
    Let $u_t \in \N$ be the $t$-th round such that $X_{u_t+1}>X_{u_t}$ and let $u_0 = 0$. Then $X_{u_t}=n-n\left(\frac {n-1} n \right) ^{2^t}$. Moreover, we have that $X_{u_{t+1}}=X_{u_t}+(n-X_{u_t})\cdot \frac {X_{u_t}} n$.
    \end{restatable}

This allows us to estimate when $N_t$ reaches $n$, as when $s \ge 2 \ln n$, $X_{u_s} > n-1$.

\begin{restatable}{lemma}{Nequalsn}\label{lem:Nequalsn}
    If $t \ge u_{2\ln n}$, then $N_{t}=n$.
\end{restatable}

All we need to show now is that $X_{t+1} - X_t$ exceeds its expectation often enough. For that, we use a result due to Greenberg and Mohri~\cite{greenberg2014tight}, that will give us an estimate of the probability of $X_t$ strictly increasing in a given round.

\begin{restatable}[Theorem 1 of \cite{greenberg2014tight}]{theorem}{binestimate}\label{thm:binestimate}

For any positive integer $m$ and any probability $p$ such that $p > \frac 1 m$, let $B$ be a binomial random variable of parameters $(m,p)$. Then, the following inequality holds:
$$
\Proba(B \ge mp) > \frac 1 4
$$
\end{restatable}

In fact, in Lemma~\ref{prob:bin}, we are able to relax the condition to $p > \frac 1 {3m}$ while keeping the same inequality. We use this lemma to lower bound the probability of $X_t$ increasing in any round: 

\begin{restatable}{lemma}{xbinomial}
    If $n> 4$, for every $t \in \N$, we have that $\Proba\left(X_{t+1} > X_{t}\right) \geq \frac{1}{4}$.
\end{restatable}

    We now show that, for any $c\ge 1$, over $32\cdot c\cdot \ln n$ rounds, $X_t$ increases at least $2\ln n$ times with high probability. For that, let $(B_t)_{t \in \N}$ be Bernoulli independent random variables of parameter $\frac{1}{4}$. Let $Z^B_{\le t}=\sum_{z \in [t]} B_z$ and $Z_{\le t}=\sum_{z \in [t]} \1 \left( X_{z+1} > X_z\right)$. 
    By Lemma 1.8.5 of~\cite{DBLP:series/ncs/Doerr20} the following trivial corollary follows.
 \begin{corollary}  \label{cor:upbound} 
    For any $\ell \in \N$, we have that $\Proba (Z_{\le t} \le \ell) \le \Proba (Z_{\le t}^B \le \ell)$.
\end{corollary}

\begin{lemma}\label{lem:manyincreases}
    Let $t=32\cdot c\cdot \ln n$ for any $c \ge 1$. Then $\Proba (Z_{\le t} \le 2\ln n) \le \frac 1 {n^c}$.
\end{lemma}

\begin{proof}

Note that $Z_{\le t}^B$ is a binomial distribution of parameters $(t, \frac 1 4)$. Using Hoeffding's inequality (Lemma~\ref{lem:hoeffding}), we have that:
$$\Proba (Z_{\le t}^B \le 2\ln n)\leq \exp \left(-2t\left(\frac 1 4-{\frac {2\ln n}{t}}\right)^{2}\right)\le \exp \left(-2\cdot 32 c\ln n \left(\frac 1 4-{\frac {2}{16}}\right)^{2}\right)=n^{-c}$$
    Corollary~\ref{cor:upbound} then gives the desired result.
\end{proof}

We now have all the tools to prove Theorem~\ref{thm:Broadcast}, which we recall here:

\Broadcast*

\begin{proof}

    By Lemma~\ref{lem:manyincreases}, we have that, with probability $p \le 1-\frac 1 {n^c}$, $X_{t+1}>X_t$ for at least $2\ln n$ many rounds within the $32\cdot c\cdot \ln n$ first rounds. Recall that $u_{2\ln n}$ is the $2\ln n$-th round where $X_{t+1}>X_t$. We thus have that  $\Proba \left( u_{2\ln n} \le 32 \cdot c \cdot \ln n \right)\ge 1-n^{-c}$. But, by Lemma~\ref{lem:Nequalsn} the event $u_{2\ln n} \le 32\cdot c\cdot \ln n$ 
    implies the event $N_{32\cdot c\cdot \ln n}=n$, therefore $\Proba \left( N_{32\cdot c\cdot \ln n}=n\right) \ge 1-n^{-c}$.
\end{proof}

We now show that this result is asymptotically tight. Indeed, we can show that if at most $\log n$ rounds are allowed, then with probability $q\ge \frac 1 4$, Broadcast does not complete:

\lowerbounded*

\begin{proof}
    We will only show the result for Broadcast, as the result for All-to-All Broadcast follows immediately. We will first show by induction that $\expect (N_t) \le X_{u_t}$ for every $t \in \N$. We will then conclude using Markov's inequality.

    The induction basis is clear as $N_0=X_0=1$. For the induction step, assume that for some $t \in \N$, we have that $\expect (N_t) \le X_{u_t}$. Let us show that this implies that $\expect (N_{t+1}) \le X_{u_{t+1}}$. 
    Indeed, by Lemma~\ref{lem:binom}, $N_{t+1}-N_t$ has a binomial distribution of parameters $n-N_t$ and $\frac {N_t} n$. This implies that:
    $$ \expect [N_{t+1}| \F_t]=N_t + \frac {N_t} n \cdot (n-N_t) = 2N_t-\frac {N_t^2} n$$
    Therefore:
    $$\expect [N_{t+1}]=\expect\left[\expect [N_{t+1}| \F_t]\right] = 2 \expect [N_t] - \frac {\expect[N_t^2]} n$$
    As $Var(N_t) = \expect [N_t^2] - \expect [N_t]^2 \ge 0$, we have that $- \expect [N_t^2] \le - \expect [N_t]^2$. This implies:
$$\expect [N_{t+1}] \le 2 \expect [N_t] - \frac {\expect[N_t]^2} n$$

Note that we have that, by Lemma~\ref{lem:valueofX}:
$$X_{u_{t+1}} = 2 X_{u_{t}} - \frac {X_{u_{t}}^2} n$$

Since $x\mapsto 2x-\frac{x^2}n$ is strictly increasing between $0$ and $n$, with both $X_t$ and $\expect [N_t]$ falling in that range (Lemmata~\ref{lem:couple} and~\ref{lem:couple2}), the induction hypothesis implies that $2 \expect [N_t] - \frac {\expect[N_t]^2} n \le 2 X_{u_{t}} - \frac {X_{u_{t}}^2} n$. This implies $\expect[N_{t+1}] \le X_{u_{t+1}}$.   

We know the value of $X_{u_t}$ from Lemma~\ref{lem:valueofX}. We can thus give the upper bound $\expect[N_{\log n}] \le X_{u_{\log n}} = n (1 - ((n-1)/n)^n ) \le  n(1-\frac 1 4)$, since $n \ge 2$. Using Markov's inequality, we thus have:
$$
\Proba(N_{\log n}\ge n)\le \frac{\expect [N_{\log n}]} {n} = 1-\frac 1 4
$$
\end{proof}

We now use Theorem~\ref{thm:Broadcast} result to get a similar result for All-to-All Broadcast. Using a union-bound, we obtain:

\begin{restatable}{theorem}{alltoallBroadcast}\label{thm:alltoallBroadcast}
For any $c \ge 1$ and $n\ge 5$, All-to-All Broadcast on Uniformly Random Trees completes within $32\cdot c\cdot \ln n$ rounds with probability $p>1-\frac 1 {n^{c-1}}$.    
\end{restatable}

We now finally show a result on Consensus, which uses the following algorithm:

\begin{algorithmm}\label{alg:Consensus}
    The protocol works as follows: each node waits for $32\cdot c\cdot \ln n$ rounds, during which if it receives the initial value of node $1$, it starts forwarding it as well.
    After the $32\cdot c\cdot \ln n$ rounds have passed, it outputs that value, or $\bot$ if it hasn't received it.
\end{algorithmm}
\begin{restatable}{theorem}{Consensus}\label{thm:Consensus}
For any $c\ge 1$ and $n\ge 5$, There exists a protocol for Consensus on Uniformly Random Trees that satisfies Agreement and Validity, terminates within $32\cdot c\cdot \ln n$ rounds with probability $p>1-\frac 2 {n^c}$, and only requires messages of 1 bit over each edge in each round.  
\end{restatable}

\begin{proof}
    Algorithm~\ref{alg:Consensus} is an algorithm where everyone agrees on $v_1$, the input to node $1$, and where only $v_1$ is passed along. Thus every node outputs either $v_1$ or $\bot$. However, if $v_1$ has Broadcast within the first $32\cdot c\cdot \ln n$ rounds, then everyone outputs $v_1$. This happens with probability $p\ge 1-n^{-c}$, by Theorem~\ref{thm:Broadcast}.
\end{proof}

Note that Algorithm~\ref{alg:Consensus} can be adapted to different variants of Consensus. To keep our presentation concise, we do not explore them further in detail.  For example, the version given here satisfies the condition that no node continues to communicate after it has decided on a value, but Consensus does not complete with probability $1$ after everyone has decided as some nodes might output $\bot$. A different definition of Consensus could allow each node to send messages after it decides on a value, in which case a different version of the algorithm could be given, where each node can decide as soon as it receives the value $v_1$.

\section{Adversarial Nodes: Trees with Byzantine Nodes}\label{sec:Byzantine}
In this section, we will discuss the case where some nodes are \emph{Byzantine}, that is, nodes that can arbitrarily deviate from the protocol. 
These nodes can stop functioning, send wrong messages, and coordinate to make the protocol fail. We will rely on cryptographic tools so that each node can sign and encrypt the message it sends. Then nodes can be confident about the sender of each message and its content and can forward the message along with its unchanged signature to other nodes. We will assume that there are up to $f$ Byzantine nodes, out of a total of $n$ nodes. We require that $f \le \frac 2 3 n-1$. Nodes that are not Byzantine are called honest.

We begin by analyzing Broadcast in this setting. We first give a message to a fixed honest node, and ask the node to forward it to all other honest nodes.
Note the difference between this model and the reliable Broadcast model, where the initial message could be from an honest node or a Byzantine node, and where if the initial message is from a Byzantine node, then the message accepted by each honest node must be the same. 

In our
setting, the best strategy for the Byzantine node is not to forward any message at all. Indeed, they cannot modify the content of a message because they cannot forge any signature, and, thus, their power is limited. Hence, we will  analyze this problem as if Byzantine nodes are just defunct
but the process that chooses the communication network, i.e., the random tree, does not know which nodes are Byzantine and, thus, they are part of the network as before, i.e., the tree still consists of $n$ nodes.

As most of the analysis resembles the one of the previous section, all details are delayed to Appendix~\ref{app:Byzantine}. We get the main theorem of this section:

\begin{restatable}{theorem}{secfivethmbroadcast}\label{sec.5.thm:Broadcast}
    For any $c \ge 1$, and $f \le \frac 2 3 n -1$, Broadcast on Uniformly Random Trees with $f$ Byzantine nodes completes within
    $144\cdot c \cdot \log n$ rounds with probability
    $p> 1 - \frac 1 {n^c}$.
\end{restatable}

We now use this result to get a similar result for All-to-all Broadcast. Using a union-bound, we obtain:

\begin{restatable}{theorem}{byzalltoall}\label{sec.5.thm:alltoall}
    For any $c \ge 1$, and $f \le \frac 2 3 n -1$, All-to-all Broadcast on Uniformly Random Trees with $f$ Byzantine nodes completes within $144 \cdot c\cdot \log n$ rounds with probability $p>1-n^{1-c}$.%
\end{restatable}

This allows us, e.g.,~to implement algorithms that run on a clique %
in a synchronous setting in our sparser graph.
Indeed, each round of communication of a clique can be simulated by $32\cdot c \cdot \tau$ rounds of Uniformly Random Trees with high probability, since All-to-All Broadcast needs $32\cdot c \cdot \tau$ rounds to complete with high probability.
Essentially, if an algorithm runs in $T$ time, with $T \le n^{c-1}$, in a clique network, we can implement it with high probability in $32\cdot c \cdot T \cdot \tau$ rounds in the Uniformly Random trees network, which is essentially a logarithmic overhead.
The only caveat is that if $T$ is too large, i.e. $T > n^{c-1}$, the probability of at least one of the $T$ All-to-All Broadcast rounds failing can become close to 1. To circumvent this, we restrict ourselves to the case where $T$ is a small enough polynomial in $n$.

\cliquereduc*

We now give two applications of this theorem, namely Reliable Broadcast and Byzantine Consensus.

\reliableBroadcast*

\begin{proof}
    Dolev and Strong~\cite{reliable} have given an algorithm that solves reliable Broadcast, is robust to $f$ Byzantine nodes, and runs in $T=f+1$ rounds. Since $T\le n$, we can apply Theorem~\ref{sec.5.thm:reduction} with $x=1, \alpha = 1$, and we get the desired result.
\end{proof}

\ConsensusByzantine*

\begin{proof}
   Berman, Garay and Perry~\cite{king} have given an algorithm (known as the King's algorithm) that solves Byzantine Consensus, is robust to $f$ Byzantine nodes, and runs in $T=3(f+1)$ rounds. Since $T\le 2n$, we can apply Theorem~\ref{sec.5.thm:reduction} with $x=1, \alpha = 2$, and we get the desired result.
\end{proof}

\section{Adversarial Edges: Trees with Adversarial Topology}\label{sec:adversary}
In this section, we consider a more general model where a parametrized adversary controls a certain number of edges in every round, and the others are chosen randomly. More specifically, in each round, the adversary $A$ chooses $k$ edges such that the resulting graph is a directed rooted forest $F$, and then a tree is chosen uniformly at random among the rooted trees that are compatible with $F$. We consider the model where the adversary has access to the randomly chosen trees of all previous rounds, but has no information on the random coin flips of the current and future rounds.

All the proofs missing in this section can be found in Appendix~\ref{appendix:adversary}. 

We will prove the following theorem:
\begin{theorem}
    If the adversary controls $k$ edges in each round, for $k \le \frac 2 3 n -1$, then for any $c\ge 1$, with probability $p\ge 1-n^{-c}$, Broadcast completes within $O(k+\log n)$ rounds.%
\end{theorem}

We will in fact show that Broadcast completes within $O(k+\tau)$ rounds, where $\tau = \frac{\log n}{\log (1+\frac{n-k}{2n})} = \Theta(\log n)$.

For the rest of the section, $I_t$ and $S_t$ will, respectively, be the set of nodes that are informed and uninformed after round $t$. 
We set $I_0=\{1\}$ and $S_0=[n]-\{1\}$, $N_t=\card{I_t}$ to be the number of informed nodes after $t$ rounds, and $T_t$ to be the tree chosen at random in round $t$. 
For a tree $T$, for each node $p$, $P_T(p)$ is the (unique) parent of node $p$ in $T$, unless $p$ is the root of $T$, in which case $P_T(p)=p$. 
Simplifying the notation, we also use $P_t(p)$ to denote $P_{T_t}(p)$.

We start by finding the best strategy $A$ could use and then analyze that strategy.

\subsection{\texorpdfstring{Best Strategy for the Adversary $A$}{Best Strategy for the Adversary A}}

\begin{figure}
    \centering
    \includegraphics[width=0.6\linewidth]{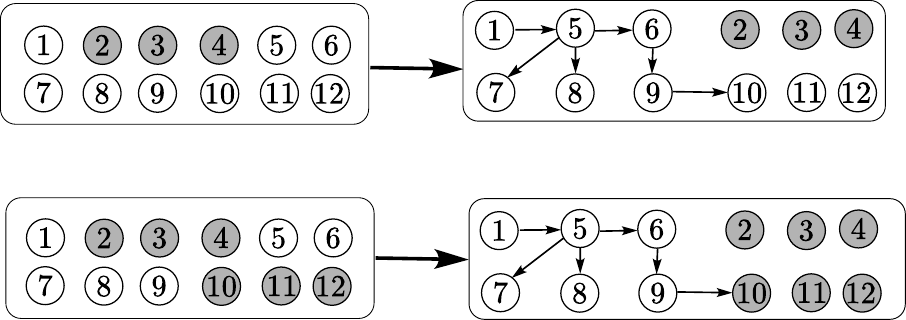}
    \caption{The best strategy for the adversary $A$, with $k=6$. Shaded nodes are informed nodes. In the top example, nodes $5, 6, 7, 8, 9$ and $10$ are safe from being informed, whereas node $1$ can still be informed. In the bottom example, nodes $5, 6, 7, 8,$ and $9$ are safe, whereas node $1$ can still be informed. However, node $1$ is safe from being informed by node $10$.}
    \label{fig:beststrategy}
\end{figure}

In this subsection, we will show that the best strategy for the adversary is to use all the edges to form one tree, with as many uninformed nodes as possible. The main idea is that an uninformed node with an uninformed parent is ``protected" in the round, that is, it cannot have an informed parent. Hence the adversary will try to protect as many uninformed nodes as possible by creating a tree connecting them. If the adversary still has edges left after it protected all the nodes it could, they use the remaining edges to connect as many informed nodes as possible to the tree such that there is no edge from an informed to an uninformed node
to prevent that any of them becomes a parent of  the root of the tree. 

An illustration of this strategy is shown in Figure~\ref{fig:beststrategy}. This section is dedicated to formalizing and proving these ideas. We will use the notion of stochastic dominance. Intuitively, if a strategy yields more informed nodes than another one, then the adversary will choose the latter one. Stochastic dominance is the tool we use to formalize this. Note that we define stochastic dominance for two types of random variables, namely random variables that are real numbers and random variables that are sets. Thus, for any set $S$, let $\parts(S)$ be the set of all subsets of $S$.

\begin{definition}[Stochastic Dominance]
    (1) A real random variable $Y_1$ stochastically dominates another real random variable $Y_2$, if, for every $x \in \R$, we have that $\Proba(Y_1 \ge x) \ge \Proba (Y_2 \ge x)$.
    
    (2) A random variable $Y_1$ with values in $\parts([n])$ stochastically dominates another random variable $Y_2$ with values in $\parts([n])$, if, for every $x \in \N$, we have that $\Proba(\card {Y_1} \ge x) \ge \Proba (\card{Y_2} \ge x)$.
\end{definition}

With stochastic dominance, we will use a related notion, that is coupling. Coupling is a useful tool to compare two random variables, and in particular, it helps translate probabilistic events into deterministic ones, which are easier to analyze. 

\begin{definition}[Coupling]
   A coupling of two random variables $Y_1, Y_2$ is a third random variable $(\hat{Y_1}, \hat{Y_2})$ such that $Y_1$ has the same distribution as $\hat{Y_1}$, and $Y_2$ has the same distribution as $\hat{Y_2}$.
\end{definition}

Next we state two coupling theorems, namely one for random variables that are real numbers and one for random variables that are sets.
 
\begin{theorem}[Stochastic Dominance and Coupling, Theorem 7.1 of~\cite{den2012probability}]
If a real random variable $Y_1$ stochastically dominates another real random variable $Y_2$, then there exists a coupling $(\hat{Y_1}, \hat{Y_2})$ of $Y_1$ and $Y_2$ such that
$$
\Proba (\hat{Y_1}\ge \hat{Y_2})=1
$$ 
\end{theorem}
\begin{theorem}[Stochastic Dominance and Coupling, Theorem 7.8 of~\cite{den2012probability}]\label{thm:dominancecoupling}
If a random variable $Y_1$ with values in $[n]$ stochastically dominates another random variable $Y_2$ with values in $[n]$, then there exists a coupling $(\hat{Y_1}, \hat{Y_2})$ of $Y_1$ and $Y_2$ such that
$$
\Proba \left(\card{\hat{Y_1}}\ge \card{\hat{Y_2}}\right)=1
$$
\end{theorem}
The next lemma contains a crucial observation: being greedy in each round is an optimal strategy for the adversary.
\begin{restatable}[Distribution Domination]{lemma}{distributiondomination}\label{lem:domination}
    Let $t$ be a round.
    Let $E_1, E_2$ be two sets of edges the adversaries could choose for round $t$. Let $N^{(1)}_t$ (resp.~$I^{(1)}_t$) be the number (resp.~set) of informed nodes after round $t$ if $E_1$ is chosen, and $N^{(2)}_t$  (resp.~$I^{(2)}_t$) if $E_2$ is chosen. Then if $\Proba (N^{(1)}_t \ge m) \ge \Proba (N^{(2)}_t \ge m)$ for every $m \in \N$ (that is, if $N_t^{(1)}$ stochastically dominates $N^{(2)}_t$), then choosing $E_2$ is a better strategy for the adversary than choosing $E_1$. 
\end{restatable}

Intuitively, the way to prove this is to build, for any strategy the adversary might use after choosing $E_1$, another strategy that would work better if used after choosing $E_2$. To prove that it is indeed the case, we couple these two strategies to prove that after any round, the number of informed nodes in one strategy stochastically dominates the number of informed nodes in the other one. The full details of the proof can be found in Appendix~\ref{appendix:adversary}.

The next step is to show that the adversary will never force an edge from an informed node to an uninformed one. 
Indeed, intuitively, this means the adversary forces a node to be informed, which is against its interests. 
To do so, we introduce the notions of \emph{non-increasing} and \emph{increasing} trees, and show that $A'$ will never choose an increasing tree.
An illustration is given in Figure~\ref{fig:correction}.

\begin{definition}
    A rooted tree $U$ at round $t$ is said to be \emph{non-increasing in round $t$}  if all edges in $U$ whose source is in $I_{t-1}$ have their target in $I_{t-1}$ as well. Otherwise a tree is \emph{(information)-increasing in round $t$}.
\end{definition}

To show that the worst-case adversary never uses an increasing tree, we introduce the notion of a \emph{correction} of an increasing tree, which will be non-increasing, and show that choosing the correction is a better strategy for the adversary than choosing the increasing tree.

\begin{definition}[Isomorphism]
    We say that a rooted tree $U$ on $n$ nodes is isomorphic to a rooted tree $U'$ on $n$ nodes if there exists a bijection $b$ from $[n]$ to $[n]$ such that for every (directed) edge $(u,v) \in U$, we have that $(b(u), b(v)) \in U'$, and for every (directed) edge $(u, v) \in U'$, we have that $(b^{-1}(u), b^{-1}(v)) \in U$.
\end{definition}

In particular, if $r$ is the root of $U$, then $b(r)$ is the root of $U'$.

\begin{definition}
    A \emph{correction} of a tree $U$ that is increasing at round $t$ is a tree $U'$ over the same nodes as $U$ that (1) is isomorphic to $U$, (2) is non-increasing in round $t$, and (3) whose root is a node $s \in S_{t-1}$ such that $P_U(s) \in I_{t-1}$.
\end{definition}

Intuitively, if a tree is increasing, one can correct it by putting all the informed nodes at the bottom of the tree.

\begin{restatable}{lemma}{lembij}
\label{lem:bij}
    For any increasing tree $U$, there exists a correction $U'$.
\end{restatable}

The following lemma proves that the worst-case adversary will never choose a set of edges such that one (or more) component is increasing. Indeed, if such components existed, then the adversary would have replaced all of them  with non-increasing ones, as this will lead to no fewer and potentially more rounds. Therefore, we can assume in the following that  all components are non-increasing.

\begin{restatable}{lemma}{correction}
    Let $t$ be a round and $N_{t-1}$ be the number of informed nodes after round $t-1$. Let $E_1, E_2$ be two sets of edges that the adversary could choose for round $t$ such that
    \begin{enumerate}
    \item $E_1$ is a collection of rooted trees such that at least one tree $U$ is information-increasing, and 
    \item $E_2$ is obtained from $E_1$ by replacing $U$ with a correction $U'$ of $U$. 
    \end{enumerate}
    Let $N^{(1)}_t$ be the number of informed nodes after round $t$ if $E_1$ is chosen, and let $N^{(2)}_t$  be that number if $E_2$ is chosen. Then choosing $E_2$ is a better strategy for the adversary than choosing $E_1$.
\end{restatable}

\begin{figure}
    \centering
    \includegraphics[width=0.15\linewidth]{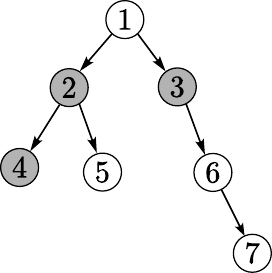}\hspace{0.1\linewidth}\includegraphics[width=0.15\linewidth]{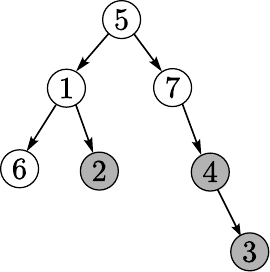}
    \caption{Shaded nodes are informed nodes. Left: A tree $U$ that is information increasing. Right: A tree $U'$ that is a correction of~$U$.  
    }
    \label{fig:correction}
\end{figure}

The next step is to show that if the adversary chooses a forest, all edges will be used in one component. For that, we introduce the notion of \emph{merging trees}, and show that if the adversary chooses a forest with 2 or more non-trivial components, then merging two of those non-trivial components will yield a better strategy for the adversary. 
We start by computing the probability that a set of roots of the forest given by the adversary get informed:

\begin{figure}
    \centering
    \includegraphics[width=0.3\linewidth]{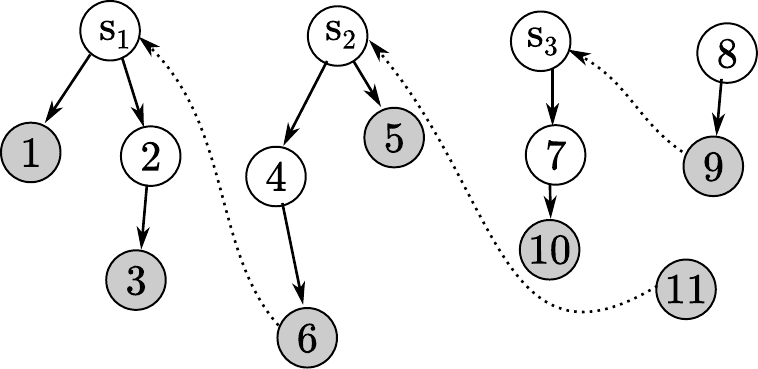}\hspace{0.1\linewidth}\includegraphics[width=0.3\linewidth]{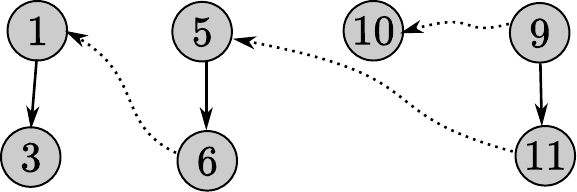}
    \caption{Illustration for the proof of Lemma~\ref{lem:iinformed}, Case 1. Shaded nodes are informed nodes. Left: Solid lines represent $E$. Dotted lines are a suitable choice of $a$. Right: Solid lines represent $F$ associated to $E$. Dotted lines represent $b(a)$. Any tree rooted at $9$ on the right yields a suitable choice of $a$ on the left. 
    }
    \label{fig:whyrootedtrees}
\end{figure}

\begin{lemma}\label{lem:iinformed}
    Let $t$ be a round, let $E$ be the set of $k$ edges forming a directed rooted forest over $[n]$ which the adversary chooses in round $t$ such that each component of $E$ is non-increasing, and let $s_1, \dots, s_x$ be uninformed %
    nodes that are roots of their component (which might have size only 1). Note that $\{s_1, \dots, s_x\}$ needs not be the set of the roots of all components, simply a collection of some of them.  Let $\eta_1, \dots, \eta_x$ be the number of informed nodes in the component of $s_1, \dots, s_x$ respectively, and let $\eta$ be the number of informed nodes outside the components of $s_1, \dots, s_x$. Then we have that:

    $$
    \Proba \left( \inter_{j \in [x]} (P_{t}(s_j)\in I_{t-1})\middle| \F_{t-1}\right)=\frac{\eta(\eta+\sum_{j\in [x]}\eta_j)^{x-1}}{n^x}=\frac{\eta(N_{t-1})^{x-1}}{n^x}
    $$
\end{lemma}

\begin{proof}
We have that:
$$
         \Proba \left( \inter_{j \in [x]} (P_{t}(s_j)\in I_{t-1})\middle| \F_{t-1}\right)= \sum_{a\in (I_{t-1})^{x}}\Proba \left( \inter_{j \in [x]} (P_{t}(s_j)= a_j)\middle| \F_{t-1}\right)
$$

However, many terms of that sum are equal to $0$. Indeed, for example, if $a_1$ is one of the $\eta_1$ informed nodes in the component of $s_1$, then $\Proba (P_t(s_1)=a_1)=0$. More generally, if the choice of $a$ is such that $E\cup \Union_{j \in [i]} (a_j, s_j)$ contains an (undirected cycle), in other words, is incompatible with a rooted tree, then $\Proba (P_t(s_1)=a_1)=0$. If, on the other hand, the choice of $a$ is compatible with a rooted tree, then, applying Theorem~\ref{thm:treecount}, we have:

$$
\Proba \left( \inter_{j \in [x]} (P_{t}(s_j)= a_j)\middle| \F_{t-1}\right)=\frac{\card{T \in \RT_n: (E\Union_{j \in [x]} (a_j, s_j)) \subset T }}{\card{T \in \RT_n: E \subset T }}=\frac{n^{n-1-\card E-x}}{{n^{n-1-\card E}}}=n^{-x}
$$

We now have to count how many choices of $a$ are compatible with a rooted tree. %
Let us call these the \emph{suitable} choices of $a$. To do so we create a bijection between the set of all suitable choices of $a$ and a simple set of forests consisting only of trees that are line graphs. This basically says that for counting the number of suitable choices, we can ignore the internal structure of each tree.

\emph{Case 1:} A figure for that case is given in Figure~\ref{fig:whyrootedtrees}. Let us first assume that none of the $\eta_j$ nor $\eta$ is equal to $0$. Let $\alpha$ denote the set of all such values of $a$, and define $\beta$ as follows: create a forest $F$ with $x+1$ (directed) line graphs, each line having respectively $\eta_1, \dots, \eta_x, \eta$ nodes. Then $\beta$ is the set of all rooted trees that are compatible with $F$, and whose root is the root of the last tree of $F$. 

To determine $|\alpha|$, we show that there is a bijection between $\alpha$ and $\beta$ and determine $|\beta|$. To create the bijection first take an arbitrary but fixed bijection $b$ that maps every informed node from $I_{t-1}$ to a node from $F$, such that an informed node from the component of $s_j$ is mapped to a node of the $j$-th line of $F$. Recall that each $a \in \alpha$ assigned a parent $a_j$ to each node $s_j$ with $1 \le j \le x$. We can map a choice of $a\in \alpha$ to a tree $T\in \beta$ by setting the parent  in $T$ of the root of the $j$-th line to be $b(a_j)$ for every $j$. Note that this uniquely identifies a tree of $\beta$. Conversely, to find a choice $a\in \alpha$ from a tree $T\in B$, set $a_j=b^{-1}(p_j)$ where $p_j$ is the parent of the root of the $j$-th line of $F$ in $T$. Now note that $\beta$ is the set of all rooted trees that are compatible with $F$, and whose root is the root of the last tree of $F$. By Theorem~\ref{thm:treecountroot}, $\card \beta = \eta (\eta+\sum_{j\in [x]}\eta_j)^{x-1}$, which concludes the proof for this case.%

\emph{Case 2:} If $\eta = 0$, it is easy to see that no choice of $a$ is compatible with a rooted tree, as $a$ assigns a parent to each root $s_j$ for $1 \le j \le x$.

\emph{Case 3:} If there exists some values of $j$ such that $\eta_j=0$, then assume wlog that $\eta_1=\dots=\eta_\ell=0$, and $\eta_j >0$ for every $j > \ell$. By the same arguments as in Case 1, there will be $\eta (\eta+\sum_{j\in [x]}\eta_j)^{x-\ell-1}$ suitable choices for $(a_{\ell+1}, \dots, a_x)$. Once this choice is made, for every $1\le j\le \ell$, $a_j$ can take any value in $I_{t-1}$, where $\card{I_{t-1}}=\eta+\sum_{j\in [x]}\eta_j$. Thus, the total number of choices for $a$ is $\eta (\eta+\sum_{j\in [x]}\eta_j)^{x-1}$.
\end{proof}
The following merge operation combines two trees such as to make a uninformed root the root of the merged tree, if at least one of the roots is uninformed.

\begin{definition}
    We say that we \emph{merge} two non-trivial trees $U$ and $U'$ with respective roots $r$ and $r'$ in round $t$ when we apply the following operation:
    \begin{itemize}
        \item If $r \in I_{t-1}$, then for every $p \in U$ with $(r,p)\in U$, replace edge $(r, p)$ with the edge $(r',p)$.
        \item If $r \notin I_{t-1}$, then for every $p \in U'$ with $(r',p)\in U'$, replace edge $(r', p)$ with the edge $(r,p)$. %
    \end{itemize}
\end{definition}

\begin{figure}
    \centering
    \includegraphics[width=0.5 \linewidth]{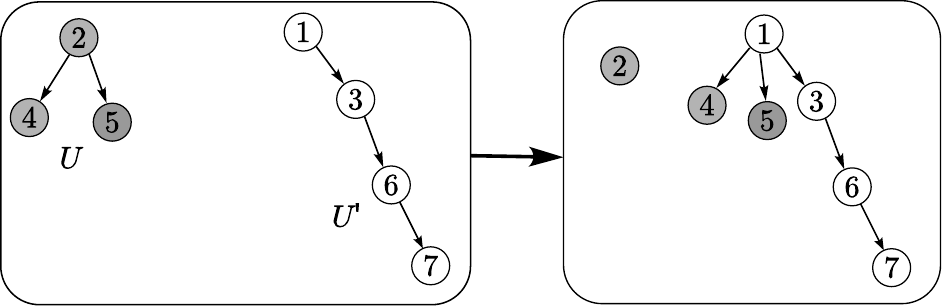}
    \includegraphics[width=0.5 \linewidth]{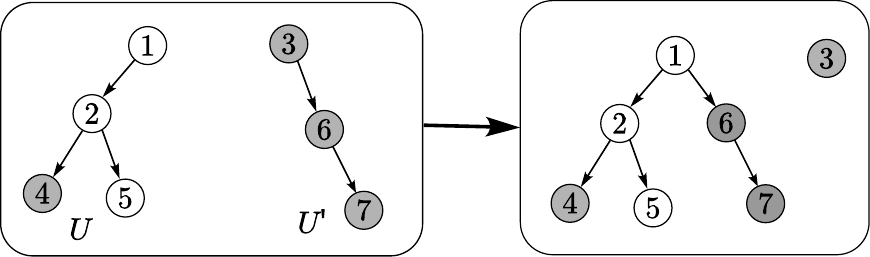}
    \caption{Merging examples}
    \label{fig:merging}
\end{figure}

\begin{restatable}{lemma}{merging}
    Let $t$ be a round and $N_{t-1}$ be the number of informed nodes after round $t-1$. Let $E_1, E_2$ be two sets of edges that the adversary could choose for round $t$, as follows: let $E_1$ be a collection of rooted trees such that every tree is non-increasing, with at least two non-trivial components $U$ with root $r$ and $U'$ with root $r'$, and let $E_2$ be obtained from $E_1$ by merging $U$ and $U'$. %
    Let $N^{(1)}_t$ %
    be the number of informed nodes after round $t$ if $E_1$ is chosen, and $N^{(2)}_t$ if $E_2$ is chosen. Then choosing $E_2$ is a better strategy for the adversary than choosing $E_1$.
\end{restatable}

This lemma implies that the adversary will never choose a set of edges with more than one non-trivial component, i.e., the adversary will choose \emph{one} tree with $k+1$ nodes.
We already showed that the adversary will only choose non-increasing components. 
Therefore, we  are left with analyzing the case where the adversary  chooses one non-trivial non-increasing tree with $k+1$ nodes.

\begin{restatable}{lemma}{binomial}\label{lem:binomialadv}
    Let $t$ be a round and $N_{t}$ be the number of informed nodes after round $t$. Let $U$ be a non-increasing tree over $k+1$ nodes in round $t+1$. Let $\sigma$ be the number of uninformed nodes in $U$ and $\eta$ the number of informed nodes in $U$. Then the distribution of $N_{t+1}-N_t$ %
    equals the  sum of  of $n-N_t-\sigma$ independent Bernoulli random variables of parameter $ \frac {N_t} n$ plus one Bernoulli random variable of parameter $\frac{N_t-\eta} n$.
\end{restatable}

\begin{corollary}\label{cor:eta}
    Let $t$ be a round and $N_{t}$ be the number of informed nodes after round $t$. Let $U$  be a non-increasing tree over $k+1$ nodes in round $t+1$ and let $\eta$ be its number of informed nodes in $U$.
    The optimal strategy for the adversary is to minimize $\eta$ in every round.
\end{corollary}

\begin{proof}
    Recall $\sigma$ is the number of uninformed nodes in $U$.
    Note that we always have $\sigma + \eta = k+1$. 
    Let us consider two non-increasing trees $U$ and $U'$ over $k+1$ nodes. Let $\eta_1$ (resp.~$\sigma_1$) be the number of informed (resp.~uninformed) nodes in $U$, and $\eta_2$ (resp.~$\sigma_2$) be the number of informed (resp.~uninformed) nodes in $U'$. Assume wlog that $\eta_1 >\eta_2 \ge 0$. Then $\sigma_1 < \sigma_2$. Let $N^{(1)}_{t+1}$ and $N^{(2)}_{t+1}$ be the number of informed nodes after round $t+1$ if the adversary chooses respectively tree $U$ or $U'$. The distribution of $N^{(1)}_{t+1}-N_t$ is %
    the sum of at least $n-N_t-\sigma_1$ independent Bernoulli variables of parameter $\frac{N_t}{n}$, while $N^{(2)}_{t+1}-N_t$ is %
    the sum of at most $n-N_t-\sigma_2+1$ independent Bernoulli variables of parameter at most $\frac{N_t}{n}$. The first distribution clearly dominates the second, and by the Distribution Domination Lemma (Lemma~\ref{lem:domination}), the result holds.
\end{proof}

This shows that the optimal strategy for the adversary is always to choose the number $\sigma$ of uninformed nodes in the tree  $U$ chosen by the adversary equal to $k+1$, unless the number of informed nodes $N_{t-1}$ is so large that $\sigma$
is smaller than $k+1$, in which case $\sigma=n-N_{t-1}$, which is the number of uninformed nodes. As the number  $N_t$ of informed nodes never decreases, this leads to the following partitioning of the rounds into  two phases:  one phase  which contains all rounds $t$ where the number of uninformed nodes is at least $k+1$, i.e.,~$n-N_{t-1} \ge k+1$%
, in which case $\sigma=k+1$, and another phase which contains all rounds $t$ with $n-N_{t-1} < k+1$, in which case $\sigma=n-N_{t-1}$. We will show that the first phase takes $O(\log n)$ rounds, while the second one takes $O(k+\log n)$ rounds.

\subsection{Phase 1}
In phase 1, we have that $N_{t+1}-N_t$ follows a binomial distribution of parameters $(n-k-N_t, \frac {N_t} n)$. This is exactly the same evolution as the case detailed in Section~\ref{sec:Byzantine}, or more precisely the result of Lemma~\ref{sec:5.lem:binom}. We hence get the same result as in Theorem~\ref{sec.5.thm:Broadcast}:

\begin{restatable}{lemma}{phaseone}\label{lem:phase1}
    If $k \le \frac 2 3 n -1$ then, for any $c \ge 1$,
    Phase 1 ends within $32 \cdot c \cdot \tau$ rounds with probability $p \ge 1-n^{-c}$, where $\tau=\frac {\log n}{ \log \left( 1 + \frac{n-k}{2n}\right)}$.
\end{restatable}

\subsection{Phase 2}

Phase two starts when there are only $k$ uninformed nodes. This essentially means that the adversary can protect all uninformed nodes but one, as the trees they will choose will have an uninformed root, which might get informed in this round. Note that all uninformed nodes below it will not become informed in the current round.

\begin{lemma}\label{lem:phase2}
     If $k \le \frac 2 3 n -1$, for any $c \ge 1$, Phase 2 ends within $8c\cdot \min\{\ln n,  \frac {kn}{n-k-1}\} \leq 12 c\cdot \min\{\ln n, k\}$  rounds with probability $p \ge 1-n^{-c}$.
\end{lemma}

\begin{proof}
 Recall that in this phase every uninformed node belong to $U$. Thus, there is only one node, namely the root of $U$ that can be informed through the randomly chosen edges.
    In each round, by Lemma~\ref{lem:binomialadv} with $\sigma = n - N_t$,  exactly one node, which as discussed must be the root, gets informed with probability $\frac {n-k-1} {n}$. Assimilating this to a flip of a coin where the coin has probability $\frac{N_t-((k+1) - |S_t|))} {n} = \frac {n-k-1} {n}$ of landing on heads, and flipping the coin $8c\cdot \min\{\ln n,  \frac {kn}{n-k-1}\}$ times, we are asking what the probability $p$ of the coin landing on heads at least $k$ times is. Again, using Hoeffding's inequality (Lemma~\ref{lem:hoeffding})%
    , we have that:%
\begin{multline*}
1-p \leq \exp \left(-2\times 8c\cdot \max\{\ln n , \frac {kn}{n-k-1}\}\left(\frac {n-k-1} {n}-{\frac {k}{8c\cdot \max\{\ln n , \frac {kn}{n-k-1}\}}}\right)^{2}\right)
\\ \le\exp \left(-2\times 8c\cdot \ln n \cdot \left(\frac {n-k-1} {n}\right) ^2\left(1-\frac {1}{8c}\right)^{2}\right)\\
\le\exp \left(-2\times 8c\cdot \ln n \cdot \frac {1}{9}\left(1-\frac {1}{4}\right)^{2}\right)
\le \exp \left(-c\ln n\right)\le n^{-c}
\end{multline*} 
Where we use that $k(n-k-1) \ge n-2 \ge 2n$ if $n \ge 4$.
\end{proof}

\subsection{Combining Phase 1 and 2}

We now just have to combine the results for Phases 1 and 2 to show that Broadcast completes in $O(\ln n + k)$ rounds:

\begin{theorem}\label{thm:smallk}
    If the adversary can control $k$ edges in each round, with $k \le \frac 2 3 n -1$, Broadcast completes within $32\cdot \tau \cdot c + 12c\cdot \max\{\ln n, k\}=O(c \cdot (\ln n + k))$ rounds  with probability $p \ge 1-2n^{-c}$.
\end{theorem}

\begin{proof}
    This is a direct result of Lemmata~\ref{lem:phase1} and~\ref{lem:phase2}
\end{proof}

In order to understand how tight this bound is, we give a lower bound on how many rounds the adversary can delay Broadcast:

\begin{restatable}{theorem}{lowerboundk}
If $n\ge 2$, and the adversary controls $k$ edges in each round, then there exists a strategy for the adversary that delays Broadcast with a Randomized Oblivious Message Adversary to at least $\frac{kn}{2(n-k-1)} \ge \frac k 2$ rounds with probability at least $\frac 1 4 $.
\end{restatable}

\begin{proof}
   Let us look at an adversary that only uses a non-increasing tree over nodes $1$ to $k+1$, where the root is always chosen to be the one with the smallest ID among those that are still uninformed. Let $N'_t$ be the number of informed nodes after $t$ rounds among $[k+1]$. Clearly $N'_0 \le 1$. By Lemma~\ref{cor:eta}, the root of the tree has probability at most $\frac {n-k-1}{n}$ of being informed in each round, and thus in expectation, $\expect[N'_t] \le 1+ \frac {t(n-k-1)} n$. Therefore, using Markov's inequality, we have that, for $t=\frac{kn}{2(n-k-1)}$:
   $$
   \Proba (N_t =n) \le \Proba(N'_t \ge k+1) \le \frac{1+\frac{kn(n-k-1)}{2(n-k-1)n}}{k+1} = \frac{1}{k+1} + \frac{k}{2k + 2}\le \frac 3 4 .
   $$
\end{proof}

Applying a union-bound on the result for Broadcast, we get a result for All-to-All Broadcast:

\begin{theorem}\label{thm:alltoalladv}
    For any $c \ge 1$ and $n\ge 5$, All-to-All Broadcast on Uniformly Random Trees with adversarial edges completes within $O(c \cdot (\ln n + k))$ rounds with probability $p>1-\frac 1 {n^{c-1}}$.    
\end{theorem}

\subsection{Consensus}
Finally, we see that a direct application of Theorem~\ref{thm:alltoalladv} gives us a reliable algorithm for Consensus with a Randomized Oblivious Message Adversary of parameter $k$:

\begin{restatable}{theorem}{Consensusadv}\label{thm:Consensusadv}
There exists a protocol for Consensus with a Randomized Oblivious Message Adversary that satisfies Agreement and Validity, and terminates in $O(c\cdot( \ln n +k))$ rounds with probability $p \ge 1-\frac  {2}{n^c}$, and only requires messages of 1 bit over each edge in each round, as long as $k\le \frac  2 3 n-1$. 
\end{restatable}

\begin{proof}
    By Theorem~\ref{thm:smallk}, node 1 Broadcasts within $O(c \cdot (\ln n + k))$ rounds with probability $p \ge 1-2n^{-c}$. Therefore, using the same arguments as in the proof of Theorem~\ref{thm:Consensus} it follows that Algorithm~\ref{alg:Consensus} achieves Consensus within $O(c \cdot (\ln n + k))$ rounds with probability $p \ge 1-2n^{-c}$, as long as we let the for loop run for $O(c \cdot (\ln n + k))$ rounds instead of $32\cdot c\cdot\ln n$ rounds.
\end{proof}
\section{Beyond Trees: Broadcast and Consensus in directed Erdős–Rényi graphs}
\label{sec:erdos}

Directed Erdős–Rényi graphs consist of $m$ edges chosen uniformly at random among the $n^2$ potential edges. Intuitively they have less structure than uniformly random trees, which makes the analysis of Broadcast simpler. We present the main ideas below. Note that we also analyze Byzantine nodes and adversarial edges in that model in Section~\ref{sec:erdos}, but omit these extensions in this overview.

Sampling a directed Erdős–Rényi graph is equivalent to choosing $m$ edges \emph{without} replacement from the set of all possible edges. We call that Scheme~1. Then we observe, using a coupling argument, that Scheme~1 requires no more rounds than Scheme~2, where in each round $m$ edges are chosen \emph{with} replacement. Finally, to analyze Scheme~2, we basically partition the sequence of rounds of Scheme~2 into $2 \ceil{(\log n) /2}$ phases, such that for each of the first $\ceil{(\log n)/2}$ phases the number of informed nodes doubles in each phase and for each of the last  $\ceil{(\log n)/2}$ phases the number of uninformed nodes halves in each phase.  Note that Broadcast completes after the last phase.
Using Hoeffding's inequality for binomial distributions we show that phase $i$ for  $1 \le i \le \ceil{\log n/2}$ requires with high probability at most $O(\max\{ \log n, 2^{i-1} \} n/2^{i-1})$ sampled edges, and, thus,  $O(\ceil{\max\{ \log n, 2^{i-1} \} /(2^{i-1} m/n )})$ rounds,
and for
$\ceil{\log n/2} + 1 \le i \le 2\ceil{\log n/2}$
 phase $i$ requires with high probability at most 
$O(\max\{ \log n, 2^{j-2} \} n/2^{j-1})$ sampled edges with with $j := 2 \ceil{\log n/2} - i$,
and, thus, $O(\ceil{\max\{ \log n, 2^{j-1} \} /(2^{j-1} m/n )})$ rounds.
Summed over all phases this shows that with high probability $O(\ceil{n/m} \log n)$ rounds suffice for Scheme~2 to reach Broadcast. 
We thus get the result:

\begin{restatable}{theorem}{thmerdos}
    For any $c \ge 1$, in scheme 2, and therefore scheme 1, Broadcast completes within $O\left(\ceil{\frac{cn}{m}} \log n\right)$ rounds with probability $p \ge 1 -  n^{-c}\log n$.
\end{restatable} 
Note that the analysis extends to the setting when the graph in each round contains \emph{at least} $m$ edges. We also show that a lower bound that implies that this upper bound is tight for $m \le n$. 

\begin{restatable}{theorem}{erdoslowerbound}
    In scheme 1, and thus in scheme 2, Broadcast fails to complete within $\frac {\log (n) -1} {\log (1+m/n)}$ rounds with probability at least $  \frac 1 2$.
\end{restatable}

We also give somewhat different analysis where the number of informed resp. uninformed nodes does not double, but increases by $(1 + m/n)$ that is tight for $m \ge n \ln n$.

\begin{restatable}{theorem}{thmerdostwo}\label{thm:largerthanlog}
    For any $c \ge 1$ and $m \in [n^2] $ such that $m/n \ge \ln n$, in scheme 2 and in scheme 1, Broadcast completes within $O\left(\frac{c\cdot \log n}{\log (1+m/n)}\right)$ rounds with probability $p \ge 1 -  n^{-c}\log n$.
\end{restatable}

We also extend those results to Consensus, Byzantine nodes and with adversarial edges. All details are delayed to Appendix~\ref{app:erdos}.

\section{Related Work}\label{sec:relwork}

Information dissemination in general and Broadcasting and Consensus in particular are fundamental topics in distributed computing. In contrast to this paper, most classic literature on network Broadcast as well as on related tasks such as gossiping and Consensus, considers a static setting, e.g., where in each round each node can send information to one neighbor~\cite{hromkovivc1996dissemination,fraigniaud1994methods}. 

Especially the Byzantine setting has received much attention in the literature. Important results include Dolev and Strong~\cite{reliable} on reliable Broadcast which is robust to $f$ Byzantine nodes, and runs in $T=f+1$ rounds, or Berman, Garay and Perry~\cite{king} on King's algorithm that solves reliable Broadcast, is robust to $f$ Byzantine nodes, and runs in $T=3(f+1)$ rounds. To just name a few.

In terms of dynamic networks, Kuhn, Lynch and Oshman~\cite{kuhn2010distributed} explore the all-to-all data dissemination problem (gossiping) in an undirected setting, where nodes do not know beforehand the total number of nodes and must decide on that number.
Dutta, Pandurangan, Rajaraman, Sun and Viola~\cite{DBLP:conf/soda/DuttaPRSV13} generalize the model to when not all nodes need to forward their message, but only $k$ tokens must be forwarded.
Augustine, Pandurangan, Robinson and Upfal~\cite{DBLP:conf/soda/AugustinePRU12} show that if the graph is an expander in every round, broadcast is complete within $O(\log n)$ rounds, even if a small enough constant fraction of nodes get churned in each round.
Ahmadi, Kuhn, Kutten, Molla and Pandurangan~\cite{DBLP:conf/icdcs/AhmadiKKMP19} study the message complexity of Broadcast also in an undirected dynamic setting, where the adversary pays up a cost for changing the network. 

In dynamic networks, the oblivious message adversary is a commonly considered model, especially for Broadcast and Consensus problems, first introduced by Charron-Bost and Schiper~\cite{charron2009heard}. The Broadcast problem under oblivious message adversaries has been studied for many years. A first key result for this problem was the $n\log n$ upper bound  by Zeiner, Schwarz, and Schmid~\cite{schwarz2017linear} who also gave a $\ceil{\frac{3n-1}{2}}-2$ lower bound. Another important result is by   F{\"u}gger, Nowak, and Winkler~\cite{fugger2020radius} who presented an $O( \log\log n)$ upper bound if the adversary can only choose nonsplit graphs; combined with the result of Charron-Bost, F{\"u}gger, and Nowak~\cite{charron2015approximate} that states that one can simulate $n-1$ rounds of rooted trees with a round of a nonsplit graph, this gives the previous $O(n\log\log n)$ upper bound for Broadcasting on trees. Dobrev and Vrto~\cite{dobrev2002optimal, dobrev1999optimal} give specific results when the adversary is restricted to hypercubic and tori graphs with some missing edges.
El-Hayek, Henzinger, and Schmid~\cite{podc22,itcs23broadcast}  recently settled the question about the asymptotic time complexity of Broadcast by giving a tight $O(n)$ upper bound, also showing the upper bound still holds in more general models. 
Regarding Consensus, Coulouma, Godard and Peters in~\cite{coulouma2015characterization} presented a general characterization on which dynamic graphs Consensus is solvable, based on Broadcastability. Winkler, Rincon Galeana, Paz, Schmid, and Schmid~\cite{itcs23Consensus} recently presented an explicit decision procedure to determine if Consensus is possible under a given adversary, enabling a time complexity analysis of Consensus under oblivious message adversaries, both for a centralized decision procedure as well as for solving distributed Consensus. They also showed that reaching Consensus under an oblivious message adversary can take exponentially longer than Broadcasting.

In contrast to the above works,  in this paper  we study a more randomized message adversary, considering a stochastic model where adversarial graphs are partially chosen uniformly at random. While a randomized perspective on dynamic networks is natural and has been considered in many different settings already, existing works on random dynamic communication networks, e.g., on the radio network model~\cite{ellen2021constant}, 
on rumor spreading~\cite{clementi2013rumor}, as well as on epidemics \cite{durrett2022susceptible}, do not consider oblivious message adversaries.
Note, however, that  the information dissemination considered in this paper is similar to the SI model for virus propagation, with results having implications in both directions \cite{eugster2004epidemic}. 
For example, Doerr and Fouz \cite{doerr2011asymptotically}  introduced an information dissemination protocol inspired by epidemics. 
More generally, randomized information dissemination protocols can be well-understood from an epidemiological point-of-view, and are very similar to the SI model which has been very extensively studied.
In contrast to the typical SI models considered in the literature \cite{murray2002mathematical}, however, our model in this paper revolves around tree communication structures which introduce additional technical challenges. 
Furthermore, existing literature often provides results in expectation, while we in this paper provide tail bounds. 

Many papers have tried to bridge the gap between the deterministic and random case, using smoothed analysis. In~\cite{modelsofsmoothing}, Meir, Paz and Schwartzman study the broadcast problem in noisy networks, under different definitions on noise. 
In particular, if in each round the graph given by the adversary is replaced by a graph chosen uniformly at random among graphs at hamming distance at most $k$ from the original graph, in the case where the adversary can suggest any connected graph, then Broadcast is reduced from $n$ rounds to $O(\min\{n, n\sqrt{\frac {\log n} k}\})$ rounds, in the case of an adaptive adversary. If the adversary is oblivious, then Dinitz, Fineman, Gilbert and Newport~\cite{smoothedanalysis} showed that it is further reduced to $O(n^{2/3}/k^{1/3} \times \log n)$.

\bibliographystyle{plain}
\bibliography{broadcast}

\newpage
\appendix
\section{Lower Bound for Deterministic Broadcast in Constant Height Trees}
\label{appendix:lowerbound}

In this section, we consider a very similar model to~\cite{itcs23broadcast}, the only difference being that the adversary is restricted to choosing trees of height at most $2$.

\paragraph*{Model} We are given $n$ nodes, and these nodes can communicate in synchronous rounds. Each node has a distinct I.D., and aims to share this I.D. with as many nodes as possible. In the beginning, each node only knows its own I.D.. An adversary chooses for each round a directed network along which nodes can communicate, among a set $A$ of allowed networks. In each round, each node sends all I.D.s it has received in previous rounds to each one of its out-neighbors. The adversary's goal it to maximize the number of rounds until broadcast, that is, until one I.D. has been received by everyone. The question is: how many rounds can the adversary delay broadcast, depending on $A$?

Authors in~\cite{itcs23broadcast} have shown that if $A$ is the set of rooted trees, then the adversary can delay broadcast for a linear number of rounds. Since a linear number of rounds is easily achievable by the adversary simply by taking a line graph $L$, and using $L$ as the communication network in each round, one would think that the height of the trees allowed play an important role to determine broadcast time. We give in Figure~\ref{fig:height} a counter example, where $A$ is the set of rooted trees of height at most $2$, and where broadcast needs at least a linear number of rounds.

\begin{figure}[h]
    \centering
    \includegraphics{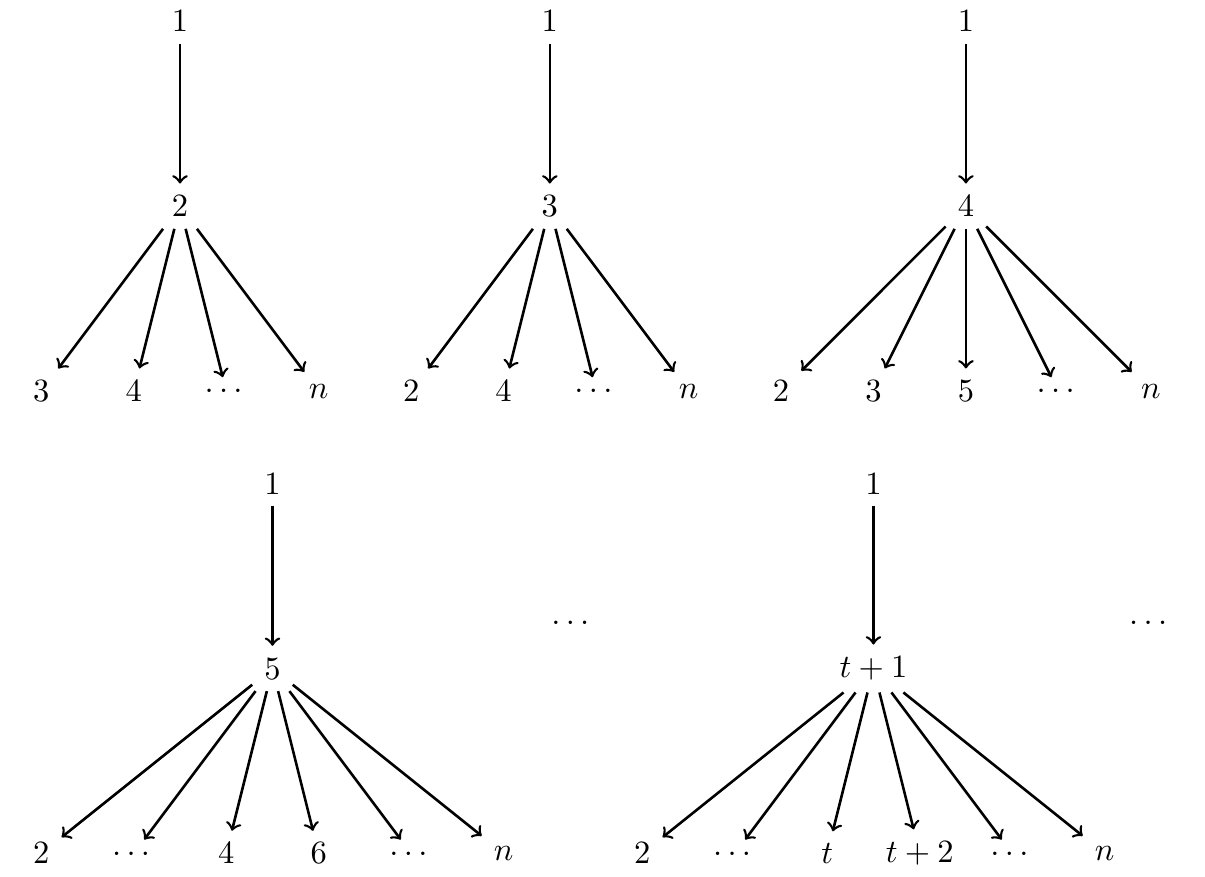}
    \caption{Lower Bound for (deterministic) Broadcast when the adversary is restricted to trees of height at most $2$.}
    \label{fig:height}
\end{figure}

In this example, in round $t$, for $t < n-2$, the adversary chooses the tree rooted at node $1$, with edges $(1,t+1)$ and $(t+1, i)$ for every $i \in [n]\setminus \{1,t+1\}$. Since node $1$ never has an in-neighbor, broadcast completes when the I.D. of node $1$ is shared to every node. It is easy to see that this only happens after round $n-2$.
\section{Counting Trees}\label{sec:tree}

In this section, we will present previously known and new results on the number of rooted trees that satisfy given properties. This will be helpful for computing probabilities in later sections. Namely, we are particularly interested in the two following results:

\begin{restatable}[Lemma 1 of \cite{pitman1999coalescent}]{theorem}{treecount}\label{thm:treecount}
Let us be given a directed rooted forest $F$ on $n$ vertices, and let $\card E$ be the number of edges in $F$. Then, the number of directed rooted trees $T$ over $n$ vertices, such that $F \subseteq T$, is $n^{n-1-\card E}$.
\end{restatable}

\treecountroot*

We will also prove Theorem~\ref{thm:treecountroot}. To do so, we start by recalling Cayley's formula~\cite{cayley1889theorem}:

\begin{theorem}[Cayley's formula]
    The number of undirected trees on $n$ vertices is $n^{n-2}$.
\end{theorem}

As a corollary of this theorem, we can compute the number of rooted trees on $n$ vertices, as choosing a rooted tree is equivalent to choosing an undirected tree, and then choosing a root:

\begin{corollary}
    The number of rooted trees on $n$ vertices is $n^{n-1}$.
\end{corollary}
Throughout this section we use $F$ to denote an \emph{undirected or directed} forest  and
 $C_1, C_2, \hdots , C_m$ of $f_1, \hdots, f_m$ vertices with integer $m \ge 1$ to denote the connected components of (the undirected version of) $F$.
The next theorem on undirected trees gives
the number of undirected trees 
which respect a set of fixed edges. It was shown by Lu, Mohr and Sz{\'e}kely~\cite{lu2012quest}.

\begin{theorem}[Lemma 6 of \cite{lu2012quest}]\label{thm:undirectedFtoT}

Let us be given an undirected forest $F$ on $n$ vertices, with connected components $C_1, C_2, \hdots , C_m$ of $f_1, \hdots, f_m$ vertices with integer $m \ge 1$. Let $\card E$ be the number of edges in $F$. Then, the number of undirected trees $T$ on $n$ vertices, such that $F \subseteq T$, is:

$$
\left(\prod_{i \in [m]}f_i\right) n^{n-2- \card E}
$$
\end{theorem}

We also recall the definition of a directed rooted forest, illustrated in Figure~\ref{fig:forestsandtrees}:

\begin{definition}[Directed Rooted Forest]
    A directed rooted forest is a collection of disjoint directed rooted trees.
\end{definition}

For simplicity, we will always require that $\sum_{i \in [m]}f_i =n$, which is always achievable by putting isolated vertices in trivial components. We will also assume that $v \in C_1$. For any directed graph $G$, $u(G)$ will represent its undirected version (Illustrated in Figure~\ref{fig:forestsandtrees}). For any directed rooted tree $T$, its root is denoted by $r(T)$. We will also use the following bijection.  Recall that $\RT_n$ is the set of all directed rooted trees on $n$ vertices. We use $T_n$ to denote the set of all undirected trees on $n$ vertices.

\begin{definition}
    Let $T_n$ be the set of all undirected trees on $n$ vertices. We define $\pi$ to be the following bijection:
\begin{align*}
  \pi \colon \RT_n &\to T_n\times [n]\\
  T &\mapsto \left(u(T), r(T)\right)
\end{align*}
\end{definition}

\begin{figure}
    \centering
    \includegraphics[width=.25\linewidth]{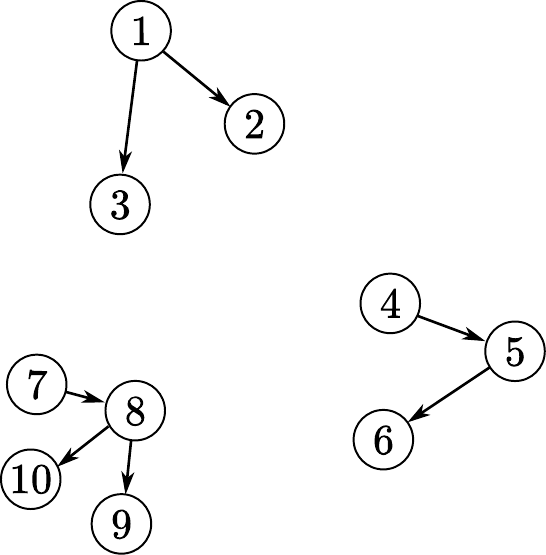} \hspace{.1\linewidth} \includegraphics[width=.25\linewidth]{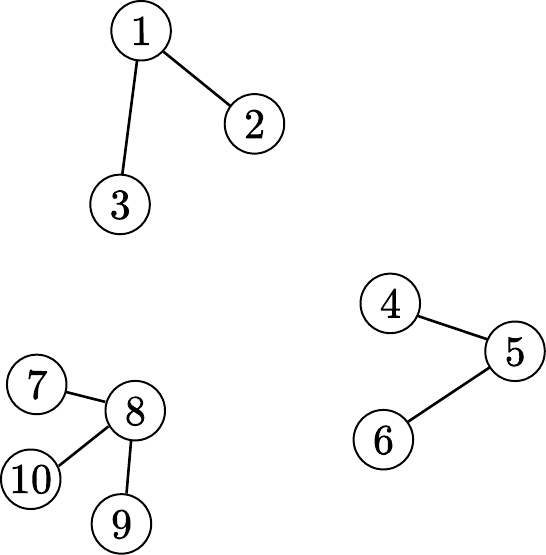} \hspace{.1\linewidth} \includegraphics[width=.25\linewidth]{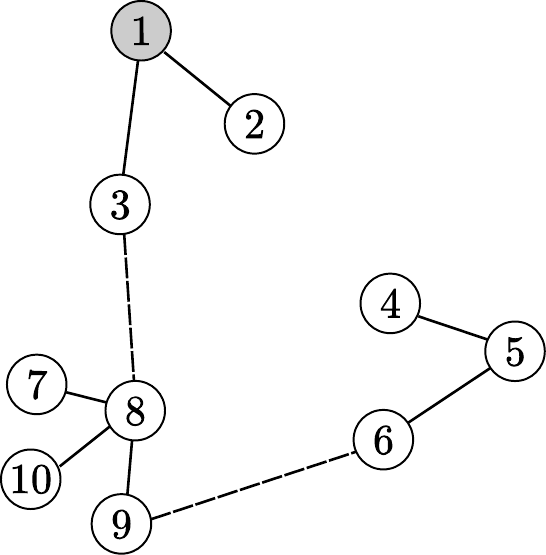}
    \caption{Left: A directed rooted forest $F$. Middle: $u(F)$. Right: A directed rooted tree $T$ in $A_F$, as seen as an undirected tree with a choice of a root. Note that $F \not\subset T$.}
    \label{fig:forestsandtrees}
\end{figure}

To prove Theorem~\ref{thm:treecountroot}, we will first look at all the trees rooted at $v$ that agree with $F$ if edge directions are ignored. Choosing such a tree is equivalent to choosing an undirected tree that contains $F$, then choosing $v$ as the root. This results in $\left(\prod_{i \in [m]}f_i\right) n^{n-2-E}$ trees. However, while all of them agree with $F$ on the undirected edges, the direction of those edges will not correspond for a majority of them. We will then partition this set of trees such that only one element of each set of the partition agrees with $F$ on the directed edges, and counting the number of sets in the partition will yield the desired result. To do so, we will use group actions.

\begin{definition}[Group action]
If $G$ is a group with identity element $e$, and $X$ is a set, then a (left) group action $\alpha$ of $G$ on $X$ is a function
\begin{align*}
  \alpha \colon G \times X &\to X
\end{align*}
that satisfies the following two axioms:

\begin{itemize}
    \item  Identity: 	$\alpha( e , x ) = x,  \forall x \in X$, where $e$ is the identity element of $G$.
    \item Compatibility: 	$\alpha ( g , \alpha ( h , x ) ) = \alpha ( g h , x ), \forall g,h \in G, \forall x \in X$
\end{itemize}   
\end{definition}

\begin{definition}[Rotations]
Let $k >0$ be an integer and let  $R_k$ be the group of all rotations of $[k]$, that is, the set of functions:
\begin{align*}
  \sigma_i^k \colon \Z/k\Z &\to \Z/k\Z\\
  x &\mapsto (x+i) \mod k
\end{align*}
\end{definition}

\begin{definition}
    Let $F$ be a forest with vertices in $[n]$ (rooted and directed or undirected), and $T$ a tree with vertices in $[n]$ (rooted and directed or undirected). We say that they are \emph{undirected-compatible} if $u(F)\subseteq u(T)$, where $u(G)$ represents the undirected version of graph $G$. If $F$ and $T$ are both rooted and directed or both undirected, we say that they are \emph{compatible} if $F \subseteq T$. 
\end{definition}

\begin{definition}\label{def:A}
    Let us be given a directed rooted forest $F$ with vertices in $[n]$.  $A_F$ is the set of directed rooted trees on $n$ vertices, rooted at $v$, that are undirected-compatible with $F$. 
\end{definition}
An example of a tree in $A_F$ is given in Figure~\ref{fig:forestsandtrees}. The following lemma follows almost immediately from Theorem~\ref{thm:undirectedFtoT}.
\begin{lemma}\label{lem:sizeA}
Let $F$ be a directed rooted forest with $n$ vertices and $E$ edges. Then
    $\card{A_F} = \left(\prod_{i \in [m]}f_i\right) n^{n-2-\card E}$.
\end{lemma}

\begin{proof}
    Let $B_F$ be the set of all undirected rooted trees that are undirected-compatible with $F$. $\pi$ induces a bijection between $A_F$ and $B_F \times \{v\}$. Therefore, $\card{A_F} = \card{B_F} \cdot 1$. By Theorem~\ref{thm:undirectedFtoT}, $\card{B_F} = \left(\prod_{i \in [m]}f_i\right) n^{n-2-\card E}$.    
\end{proof}

\begin{definition}
    For any $i \in [m]$, there exists a bijection between $\Z/f_i\Z$ and $C_i$. Let $b_i$ be that bijection.
\end{definition}

Let $R=R_{f_2} \times \hdots \times R_{f_k}$, an illustration of an element of $R$ is given in Figure~\ref{fig:action}. Note that $R$ is a group as a cartesian product of groups. 
We now define a group action of $R$ on $A_F$. This group action will allow us to partition $A_F$ as desired.

\begin{definition}[Group Action of $R$ on $A_F$]\label{def:action}
Given a forest $F$ with connected components $C_i$ with $1  \le i \le m$ and corresponding bijections $b_i$,
let $\alpha$ be the group action of $R$ on $A_F$ defined as follows:
Let $\sigma = (\sigma_{a_2}^{f_2}, \dots, \sigma_{a_m}^{f_m})$ for some $(a_2, \dots, a_m) \in \Z/ f_2\Z\times \dots \times \Z/f_m\Z$ be an element of $R$ and let $T \in A_F$. Then
$\alpha(\sigma, T)$ is obtained from $T$ by making the following modifications to $\pi(T)=(u(T), v)$: 

 For every $i \in \{2, \dots m\}$, there is one (and only one) path from $v$ to $C_i$ in $u(T)$. Let $(x,y)$ be the only edge on that path such that $x \notin C_i, y\in C_i$. Replace edge $(x,y)$ with edge $(x, b_i \sigma_{a_i}^{f_i} b_i^{-1} (y))$.
\end{definition}

To prove that this is indeed a group action, we need to verify (1) that $\alpha(\sigma, T)$ is indeed in $A_F$, (2) that the identity element $e=(\sigma_{0}^{f_1}, \dots, \sigma_{0}^{f_m})$ of $R$ verifies $\alpha(e, T)=T$ for any $T \in A_F$, and (3) that for any two $\sigma, \tau \in R$, for any $T\in A_F$, we have $\alpha(\sigma, \alpha(\tau, T))=\alpha(\sigma\tau, T)$. The second condition being trivial as $\sigma^{f_i}_0$ is the identity function for any value of $f_i$, we only prove the other two.

\begin{figure}
    \centering
    \includegraphics[width=0.25\linewidth]{third.pdf}\hspace{0.1\linewidth}\includegraphics[width=0.25\linewidth]{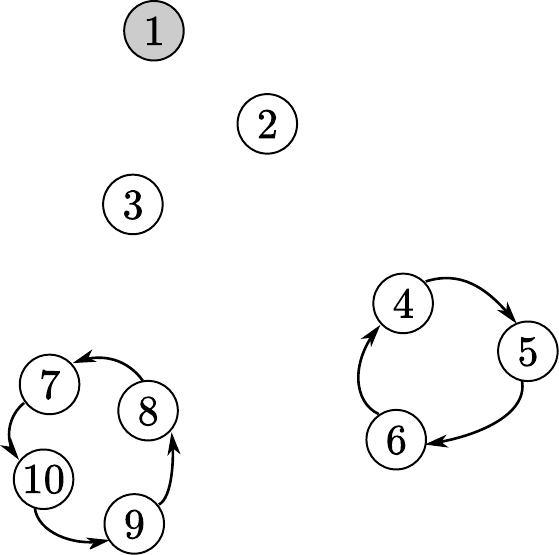}\hspace{0.1\linewidth}\includegraphics[width=0.25\linewidth]{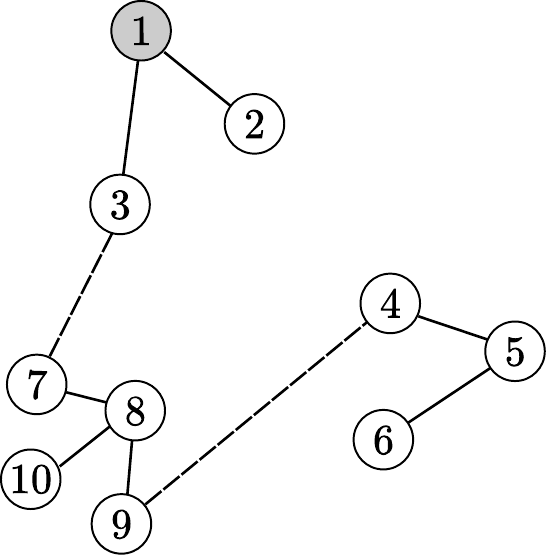}
    \caption{Left: An element $T$ in $A_F$. Middle: An element $\sigma$ of $R$. Right: The result $\alpha(\sigma, T)$ of the action $\alpha$ of $\sigma$ on $T$. Note that $F \subseteq \alpha(\sigma, T)$, where $F$ is the example from Figure~\ref{fig:forestsandtrees}, with this particular choice of $\sigma$.}
    \label{fig:action}
\end{figure}

\begin{lemma} 
\label{lem:inA}
    $\alpha(\sigma, T)\in A_F$.
\end{lemma}
\begin{proof}
   Let us first show that $u(\alpha(\sigma, T))$ is an undirected tree. As it has $n-1$ edges, we only need to show that it is connected. Let $u$ be a vertex. We need to show that it can be reached from $v$. Let $P$ be the (only) path from $v$ to $u$ in $T$, written as a sequence of vertices. Then we can split up $P$ into $P=P_1P_2\dots P_z$, where each $P_j$ is a sequence of vertices that all belong to the same $C_i$ for some $i\in [m]$. We will now replace each of the $P_j$ by another path to make a path from $v$ to $u$ in $u(\alpha(\sigma, T))$.

   Consider every edge $(x,y)$ where $x$ is the last vertex of $P_j$ for some $j$, and $y$ is the first vertex of $P_{j+1}$. There exists some $k$ such that $y \in C_k$. Then $P_1P_2\dots P_{j}y$ is the path from $r(T)$ to $C_k$ in $u(T)$. Then $(x, b_k^{-1}\sigma_{a_k}^{f_k} b_k(y)) \in u(\alpha(\sigma, T))$. Replace $y$ by $b_k^{-1}\sigma_{a_k}^{f_k} b_k(y)$ in $P$.

Let us now look at a particular $P_j$, and let $i$ be such that all of the vertices of $P_j$ belong to $C_i$, then its first vertex has been changed to another vertex of $C_i$, while all others are unchanged. Hence, the first and last vertex still belong to $C_i$.
As $C_i$  is connected in $u(\alpha(\sigma,T))$ since no edge inside $C_i$ has been modified, there exists a path $P'_j$ in $u(\alpha(\sigma,T))$ that connects that first and last vertex of $P_j$. We can thus replace $P_j$ by $P'_j$.

The new path now correctly connects $v$ and $u$ in $u(\alpha(\sigma, T))$, which shows that it is connected.
Hence $\alpha(\sigma, T)$ is a tree. Since no edge in any particular $C_i$ has been modified, $\alpha(\sigma, T)$ is compatible with $F$.    
\end{proof}

\begin{lemma}
    For any $T\in A_F$, and $\sigma, \tau \in R$ we have that $\alpha(\sigma, \alpha(\tau, T))=\alpha(\sigma\tau, T)$.
\end{lemma}

\begin{proof}

Let $\sigma=(\sigma_{a_2}^{f_2}, \dots,\sigma_{a_m}^{f_m})$ and $\tau=(\sigma_{c_2}^{f_2}, \dots,\sigma_{c_m}^{f_m})$. 
And then, for every $i \in \{2, \dots, m\}$, the path from $v$ to $C_i$ in $T$ will include the edge $(x,y)$ such that $x \notin C_i, y \in C_i$, then the corresponding edge is $(x, b_i\sigma_{c_i}^{f_i}b_i^{-1}(y))$ in $\alpha(\tau, T)$, and $(x, b_i\sigma_{a_i}^{f_i}\sigma_{c_i}^{f_i}b_i^{-1}(y))$ in $\alpha(\sigma\tau, T)$. Hence, it is $(x, b_i\sigma_{a_i}^{f_i}b_i^{-1}b_i\sigma_{c_i}^{f_i}b_i^{-1}(y))$ in $\alpha(\sigma,\alpha(\tau, T))$. We thus have $\alpha(\sigma,\alpha(\tau, T))=\alpha(\sigma\tau, T)$.
\end{proof}

As we plan to use Lagrange's theorem for group actions, we now compute the \emph{stabilizer} of a tree $T$, which is the set of all rotations that do not modify the tree:

\begin{lemma}
    $R_T := \{\sigma \in R : \alpha(\sigma, T)=T\} = \{e\}$, for every $T \in A_F$.%
\end{lemma}

\begin{proof}
    Let $\sigma\in R$ be a rotation such that $\alpha(\sigma, T)=T$. For every $i \in \{2, \dots m\}$, we look at the path from $v$ to $C_i$ in both $T$ and $\alpha(\sigma, T)$. These two paths must be the same. However, if the first element of that path in $T$ that is in $C_i$ is some vertex $y$, then in $\alpha(\sigma, T)$, it is $b_i\sigma_{a_i}^{f_i} b_i^{-1} (y)$. We conclude that $b_i\sigma_{a_i}^{f_i} b_i^{-1} (x)=x$ and thus $a_i=0$.  

We therefore have that $a_i=0$ for every $i \in \{2, \dots , m \}$, which proves that $\sigma = (\sigma^{f_2}_0, \dots, \sigma^{f_m}_0) = e$.
\end{proof}

We now take a look at the orbit $R\cdot T$ of a tree $T \in A_F$. The group action ensures that the orbits in $A_F$ form a partition of $A_F$.

\begin{theorem}[Lagrange's Theorem, Corollary 10.23 of~\cite{smith2015introduction}]\label{thm:orbits}
Let $G$ be a group, $X$ a set and $\alpha$ a group action of $G$ on $X$. Let $x$ be an element of $X$, $G_x:=\{g \in G: \alpha(g,x)=x\}$ and $G.x := \{y \in X: \exists g \in G, y=\alpha(g,x)\}$. Then we have that:
$$
\card{G.x}=\frac{\card{G}}{\card{G_x}}
$$
\end{theorem}

\begin{lemma}\label{lem:orbitsize}
     Let, for every $T \in A_F$, $R\cdot T := \{T' \in A_F : \exists \sigma \in R, \alpha(\sigma, T)=T'\}$. Then $\card {R\cdot T}=\prod_{i \in \{2, \dots, m \}}f_i$.
\end{lemma}

\begin{proof}
    By Theorem~\ref{thm:orbits}, we have that $\card {R\cdot T}=\frac {\card R}{R_T}=\frac {\prod_{i \in \{2, \dots, m\}}f_i}{1}$.
\end{proof}

We now show that exactly one tree in each orbit is compatible with $F$.

\begin{lemma}\label{lem:uniqueinorbit}
    Let $T\in A_F$. Then there exists exactly one $T' \in R\cdot T$ such that $T'$ is compatible with $F$.
\end{lemma}

\begin{proof}
    Let $T' \in R\cdot T$ be a tree such that $T'$ is compatible with $F$, and let $\sigma$ be the rotation such that $T'=\alpha(\sigma, T)$. Let, for every $i \in [m]$, $r_i$ be the root of $C_i$ in $F$. 
    
    For every $i \in \{2, \dots, m\}$, look at the path from $v$ to $C_i$ in $T$, and its corresponding path in $T'$, computed similarly to the proof of Lemma~\ref{lem:inA}. In $T'$, the first vertex of that path in $C_i$ must be $r_i$, but it also is $b_i\sigma_{a_i}^{f_i}b_i^{-1}(y)$, where $y$ is the first vertex of the path in $T$. Hence $a_i = b_i^{-1}r_i-b_i^{-1}(y)$.

    These conditions uniquely determine $\sigma$, and, thus, $T'$. Conversely, setting $\sigma$ with each $a_i$ defined as above gives a tree $T'$ that is compatible with $F$.
\end{proof}

We can now prove Theorem~\ref{thm:treecountroot}, which we recall below:

\treecountroot*

\begin{proof}
    Consider set $A_F$ as defined in Definition~\ref{def:A}. We know that every directed rooted spanning tree $T$ in $K_n$ such that $F$ is contained by $T$, and such that $v$ is the root of $T$, is in $A_F$. We can partition $A_F$ in orbits of the group action defined in Definition~\ref{def:action}. By Lemma~\ref{lem:orbitsize}, each orbit has $\prod_{i \in \{2, \dots, m\}}f_i$ elements, and thus we have $\frac{\card{A_F}}{\prod_{i \in \{2, \dots, m\}}f_i}$ orbits, which is equal to $f_1 n^{n-2-\card E}$ by Lemma~\ref{lem:sizeA}. Lemma~\ref{lem:uniqueinorbit} ensures that exactly one element in each orbit is a directed rooted spanning tree $T$ in $K_n$ such that $F$ is contained by $T$. Note that $f_1=f$ to conclude the proof.
\end{proof}

\section{Probabilities tools}
\label{appendix:prob}
\begin{lemma}\label{prob:easy}
    Let $X_1, \dots , X_m$ and $Y_1, \dots Y_m$ be binary random variables such that for every $I \subseteq [m]$ we have that $\Proba (\inter_{i \in I} (X_i=1)\inter_{i \notin I} (X_i=0)) = \Proba \left(\inter_{i \in I} (Y_i=1)\inter_{i \notin I} (Y_i=0)\right)$, then the probability distribution of $\sum_{i \in [m]} X_i$ is equal to the probability distribution of $\sum_{i \in [m]} Y_i$.
\end{lemma}

\begin{proof}
We have, for every $k \in [m]$:
    \begin{align*}
        \Proba\left(\sum_{i\in [m]}X_i=k\right)&= \sum_{\card I = k}\Proba \left(\inter_{i \in I} (X_i=1)\inter_{i \notin I} (X_i=0)\right)\\
        &=\sum_{\card I = k}\Proba \left(\inter_{i \in I} (Y_i=1)\inter_{i \notin I} (Y_i=0)\right)\\
        &=\Proba\left(\sum_{i\in [m]}Y_i=k\right)
    \end{align*}
\end{proof}

\begin{lemma}\label{prob:easy1}
    Let $X_1, \dots , X_m$ and $Y_1, \dots Y_m$ be binary random variables such that for every $\ell \in \N$, $\sum_{\card I = \ell}\Proba (\inter_{i \in I} (X_i=1)\inter_{i \notin I} (X_i=0)) = \sum_{\card I = \ell}\Proba (\inter_{i \in I} (Y_i=1)\inter_{i \notin I} (Y_i=0))$, then the probability distribution of $\sum_{i \in [m]} X_i$ is equal to the probability distribution of $\sum_{i \in [m]} Y_i$.
\end{lemma}

\begin{proof}
We have, for every $k \in [m]$:
    \begin{align*}
        \Proba\left(\sum_{i\in [m]}X_i=k\right)&= \sum_{\card I = k}\Proba (\inter_{i \in I} (X_i=1)\inter_{i \notin I} (X_i=0))\\
        &=\sum_{\card I = k}\Proba (\inter_{i \in I} (Y_i=1)\inter_{i \notin I} (Y_i=0))\\
        &=\Proba\left(\sum_{i\in [m]}Y_i=k\right)
    \end{align*}
\end{proof}

\begin{lemma}
    Let $X_1, \dots , X_m$ and $Y_1, \dots Y_m$ be binary random variables such that for every $I \subset [m], \Proba (\inter_{i \in I} (X_i=1)) = \Proba (\inter_{i \in I} (Y_i=1))$, then the probability distribution of $\sum_{i \in [m]} X_i$ is equal to the probability distribution of $\sum_{i \in [m]} Y_i$.
\end{lemma}
\begin{proof}
    We start by proving by induction on the size of $J$, $\Proba (\inter_{i \in I} (X_i=1) \inter _{j \in J} (X_j=0)) = \Proba (\inter_{i \in I} (Y_i=1)\inter _{j \in J} (Y_j=0))$ for any $I, J \subseteq [n]$ such that $I \inter J = \varnothing$. This is clear for $\card J = 0$. 

    Let $I, J \subseteq [n]$ such that $I \inter J = \varnothing$ and $\card J > 0$. Let $a$ be an element of $J$. Then we have:
    \begin{multline*}
        \Proba (\inter_{i \in I} (X_i=1) \inter _{j \in J\setminus \{a\}} (X_j=0))\\= \Proba (\inter_{i \in I} (X_i=1) \inter _{j \in J} (X_j=0))+ \Proba (\inter_{i \in I\union \{a\}} (X_i=1) \inter _{j \in J \setminus \{a\}}(X_j=0))        
    \end{multline*}
    Similarly:
    \begin{multline*}
        \Proba (\inter_{i \in I} (Y_i=1) \inter _{j \in J\setminus \{a\}} (Y_j=0))\\= \Proba (\inter_{i \in I} (Y_i=1) \inter _{j \in J} (Y_j=0))+ \Proba (\inter_{i \in I\union \{a\}} (Y_i=1) \inter _{j \in J \setminus \{a\}}(Y_j=0))        
    \end{multline*}
    By induction hypothesis, we have:
\begin{align*}
    \Proba (\inter_{i \in I} (X_i=1) \inter _{j \in J\setminus \{a\}} (X_j=0))&=\Proba (\inter_{i \in I} (Y_i=1) \inter _{j \in J\setminus \{a\}} (Y_j=0))\\
    \Proba (\inter_{i \in I\union \{a\}} (X_i=1) \inter _{j \in J \setminus \{a\}}(X_j=0)) &=\Proba (\inter_{i \in I\union \{a\}} (Y_i=1) \inter _{j \in J \setminus \{a\}}(Y_j=0))
\end{align*}
Hence:
$$
\Proba (\inter_{i \in I} (X_i=1) \inter _{j \in J} (X_j=0))=\Proba (\inter_{i \in I} (Y_i=1) \inter _{j \in J} (Y_j=0))
$$    
The result follows from Lemma~\ref{prob:easy}, when we take $J=[m]\setminus I$
\end{proof}

\begin{lemma}\label{prob:hard}
    Let $X_1, \dots , X_m$ and $Y_1, \dots Y_m$ be binary random variables such that $\sum_{\card I = \ell}\Proba (\inter_{i \in I} (X_i=1)) = \sum_{\card I = \ell} \Proba (\inter_{i \in I} (Y_i=1))$ for every $\ell \in \N$, then the probability distribution of $\sum_{i \in [m]} X_i$ is equal to the probability distribution of $\sum_{i \in [m]} Y_i$.
\end{lemma}
\begin{proof}
    We start by proving by induction on $k$, that for every $\ell, k \in \N, \sum_{\card I = \ell}\Proba (\inter_{i \in I} (X_i=1) \inter _{j \in J_I} (X_j=0)) =\sum_{\card I = \ell} \Proba (\inter_{i \in I} (Y_i=1)\inter _{j \in J_I} (Y_j=0))$ for any choice of $J_I \subseteq [n]$ such that $I \inter J_I = \varnothing$ and $\card{J_I}=k$. This is clear for $k = 0$. 

    For the induction case, let us assume, that for $k>1$, we have that for every $\ell \in \N, \sum_{\card I = \ell}\Proba (\inter_{i \in I} (X_i=1) \inter _{j \in J_I} (X_j=0)) =\sum_{\card I = \ell} \Proba (\inter_{i \in I} (Y_i=1)\inter _{j \in J_I} (Y_j=0))$ for any choice of $J_I \subseteq [n]$ such that $I \inter J_I = \varnothing$ and $\card{J_I}=k-1$.

    Let us fix $\ell$, and for every $I \subseteq [m]$ such that $\card I = \ell$, let $J_I \subseteq [m]$ be such that $I \inter J_I = \varnothing$ and $\card {J_I}=k> 0$. Let $a_I$ be an element of $J_I$. Then we have:
    \begin{multline*}
         \sum_{\card I = \ell}\Proba (\inter_{i \in I} (X_i=1) \inter _{j \in J_I\setminus \{a_I\}} (X_j=0))\\= \sum_{\card I = \ell} \Proba (\inter_{i \in I} (X_i=1) \inter _{j \in J_I} (X_j=0))+ \sum_{\card I = \ell} \Proba (\inter_{i \in I\union \{a_I\}} (X_i=1) \inter _{j \in J_I \setminus \{a\}}(X_j=0))        
    \end{multline*}
    Similarly:
    \begin{multline*}
         \sum_{\card I = \ell}\Proba (\inter_{i \in I} (Y_i=1) \inter _{j \in J_I\setminus \{a_I\}} (Y_j=0))\\=  \sum_{\card I = \ell}\Proba (\inter_{i \in I} (Y_i=1) \inter _{j \in J_I} (Y_j=0))+  \sum_{\card I = \ell}\Proba (\inter_{i \in I\union \{a_I\}} (Y_i=1) \inter _{j \in J_I \setminus \{a_I\}}(Y_j=0))        
    \end{multline*}
    By induction hypothesis, we have:
\begin{align*}
     \sum_{\card I = \ell}\Proba (\inter_{i \in I} (X_i=1) \inter _{j \in J_I\setminus \{a_I\}} (X_j=0))&= \sum_{\card I = \ell}\Proba (\inter_{i \in I} (Y_i=1) \inter _{j \in J_I\setminus \{a_I\}} (Y_j=0))\\
     \sum_{\card I = \ell}\Proba (\inter_{i \in I\union \{a_I\}} (X_i=1) \inter _{j \in J_I \setminus \{a_I\}}(X_j=0)) &= \sum_{\card I = \ell}\Proba (\inter_{i \in I\union \{a_I\}} (Y_i=1) \inter _{j \in J \setminus \{a_I\}}(Y_j=0))
\end{align*}
Hence:
$$
 \sum_{\card I = \ell}\Proba (\inter_{i \in I} (X_i=1) \inter _{j \in J_I} (X_j=0))= \sum_{\card I = \ell}\Proba (\inter_{i \in I} (Y_i=1) \inter _{j \in J_I} (Y_j=0))
$$
This concludes the induction step.

The result follows from Lemma~\ref{prob:easy1}, when we take for every $I$, $J_I=[m]\setminus I$, for $k=m-\ell$.
\end{proof}

\begin{lemma}\label{prob:easydom}
    Let $X_1, \dots , X_m$ and $Y_1, \dots Y_m$ be binary random variables, $\alpha \in \R, \alpha \ge 1$ and $r\in \N $ such that for any $I \subseteq [m]\setminus \{r\},  \Proba (\inter_{i \in I} (X_i=1)) = \Proba (\inter_{i \in I} (Y_i=1))$, and $ \Proba (\inter_{i \in I\union\{r\}} (X_i=1)) = \alpha \Proba (\inter_{i \in I\union \{r\}} (Y_i=1))$ then $\sum_{i \in [m]} X_i$ stochastically dominates $\sum_{i \in [m]} Y_i$.
\end{lemma}
\begin{proof}
    We start by proving by induction on the size of $J$, for every $I,J \subset [m] \setminus \{r\}$ such that $I \inter J = \varnothing, \Proba (\inter_{i \in I} (X_i=1) \inter _{j \in J} (X_j=0)) =\Proba (\inter_{i \in I} (Y_i=1)\inter _{j \in J} (Y_j=0))$. This is clear for $\card J = 0$. 

    Let $I, J \subseteq [n]$ such that $I \inter J = \varnothing$ and $\card J > 0$. Let $a$ be an element of $J$. Then we have:
    \begin{multline*}
        \Proba (\inter_{i \in I} (X_i=1) \inter _{j \in J\setminus \{a\}} (X_j=0))\\= \Proba (\inter_{i \in I} (X_i=1) \inter _{j \in J} (X_j=0))+ \Proba (\inter_{i \in I\union \{a\}} (X_i=1) \inter _{j \in J \setminus \{a\}}(X_j=0))        
    \end{multline*}
    Similarly:
    \begin{multline*}
         \Proba (\inter_{i \in I} (Y_i=1) \inter _{j \in J\setminus \{a\}} (Y_j=0))\\=  \Proba (\inter_{i \in I} (Y_i=1) \inter _{j \in J} (Y_j=0))+  \Proba (\inter_{i \in I\union \{a\}} (Y_i=1) \inter _{j \in J \setminus \{a\}}(Y_j=0))        
    \end{multline*}
    By induction hypothesis, we have:
\begin{align*}
     \Proba (\inter_{i \in I} (X_i=1) \inter _{j \in J\setminus \{a\}} (X_j=0))&= \Proba (\inter_{i \in I} (Y_i=1) \inter _{j \in J\setminus \{a\}} (Y_j=0))\\
     \Proba (\inter_{i \in I\union \{a\}} (X_i=1) \inter _{j \in J \setminus \{a\}}(X_j=0)) &= \Proba (\inter_{i \in I\union \{a\}} (Y_i=1) \inter _{j \in J \setminus \{a\}}(Y_j=0))
\end{align*}
Hence:
$$
 \Proba (\inter_{i \in I} (X_i=1) \inter _{j \in J} (X_j=0))= \Proba (\inter_{i \in I} (Y_i=1) \inter _{j \in J} (Y_j=0))
$$    

We then show by induction on the size of $J$, that for every $I,J \subset [m]$ such that $r \in I$, $I \inter J = \varnothing, \Proba (\inter_{i \in I} (X_i=1) \inter _{j \in J} (X_j=0)) =\alpha\Proba (\inter_{i \in I} (Y_i=1)\inter _{j \in J} (Y_j=0))$ for any $I, J \subseteq [n]$ . This is clear for $\card J = 0$. 

    Let $I, J \subseteq [n]$ such that $r \in I$, $I \inter J = \varnothing$ and $\card J > 0$. Let $a$ be an element of $J$. Then we have:
    \begin{multline*}
        \Proba (\inter_{i \in I} (X_i=1) \inter _{j \in J\setminus \{a\}} (X_j=0))\\= \Proba (\inter_{i \in I} (X_i=1) \inter _{j \in J} (X_j=0))+ \Proba (\inter_{i \in I\union \{a\}} (X_i=1) \inter _{j \in J \setminus \{a\}}(X_j=0))        
    \end{multline*}
    Similarly:
    \begin{multline*}
         \Proba (\inter_{i \in I} (Y_i=1) \inter _{j \in J\setminus \{a\}} (Y_j=0))\\=  \Proba (\inter_{i \in I} (Y_i=1) \inter _{j \in J} (Y_j=0))+  \Proba (\inter_{i \in I\union \{a\}} (Y_i=1) \inter _{j \in J \setminus \{a\}}(Y_j=0))        
    \end{multline*}
    By induction hypothesis, we have:
\begin{align*}
     \Proba (\inter_{i \in I} (X_i=1) \inter _{j \in J\setminus \{a\}} (X_j=0))&=\alpha \Proba (\inter_{i \in I} (Y_i=1) \inter _{j \in J\setminus \{a\}} (Y_j=0))\\
     \Proba (\inter_{i \in I\union \{a\}} (X_i=1) \inter _{j \in J \setminus \{a\}}(X_j=0)) &= \alpha \Proba (\inter_{i \in I\union \{a\}} (Y_i=1) \inter _{j \in J \setminus \{a\}}(Y_j=0))
\end{align*}
Hence:
$$
 \Proba (\inter_{i \in I} (X_i=1) \inter _{j \in J} (X_j=0))=\alpha \Proba (\inter_{i \in I} (Y_i=1) \inter _{j \in J} (Y_j=0))
$$    

We now show that for any $x \in \N$, we have that $\Proba\left(\sum_{i \in [m]}X_i \ge x\right)\ge \Proba\left(\sum_{i \in [m]}Y_i \ge x\right)$. Indeed, we have:

\begin{align*}
    \Proba\left(\sum _{i \in [m]} X_i\ge x \right)
    &= \sum_{\card I = x: r \in I}\Proba\left(\inter_{i \in I} (X_i=1) \inter _{j \in [m]\setminus I} (X_j=0)\right)\\
    &+\sum_{I \subset [m]: r \notin I, \card I \ge x}\Proba\left(\inter_{i \in I} (X_i=1) \inter _{j \in [m] \setminus I} (X_j=0)\right)\\
    &+\Proba\left(\inter_{i \in I\union \{r\}} (X_i=1) \inter _{j \in [m] \setminus (I\union \{r\})} (X_j=0)\right)\\
    &= \alpha \sum_{\card I = x: r \in I}\Proba\left(\inter_{i \in I} (Y_i=1) \inter _{j \in [m]\setminus I} (Y_j=0)\right)\\&+\sum_{I \subset [m]: r \notin I, \card I \ge x}\Proba\left(\inter_{i \in I} (X_i=1) \inter _{j \in [m] \setminus (I\union\{r\})} (X_j=0)\right)\\
    &\ge \sum_{\card I = x: r \in I}\Proba\left(\inter_{i \in I} (Y_i=1) \inter _{j \in [m]\setminus I} (Y_j=0)\right)\\&+\sum_{I \subset [m]: r \notin I, \card I \ge x}\Proba\left(\inter_{i \in I} (Y_i=1) \inter _{j \in [m] \setminus (I\union\{r\})} (Y_j=0)\right)\\
    &=\Proba\left(\sum _{i \in [m]} Y_i\ge x \right)
\end{align*}

\end{proof}

\begin{lemma}\label{prob:bin}
    Let $X$ be a random variable that has a binomial distribution of parameters $(m, p)$. Then if $0 < p \le \frac 1 m$, we have that $    \Proba(X\ge mp) \ge \frac 1 4$ as long as $p \ge 1-\left(\frac 3 4 \right)^{\frac 1 m}$.

    In particular, it suffices for $p$ to be larger than $\frac 1 {3m}$.
\end{lemma}

\begin{proof}
    If $p\le \frac 1 m$ then $0 < mp \le 1$ which means that the events $X\ge mp$ and $X\ge 1$ are the same since the binomial distribution takes only integer values. Hence:

    \begin{align*}
        \Proba(X\ge mp)&=\Proba(X\ge 1)=1-\Proba(X=0)=1-(1-p)^m\\
        &\ge 1-\left(1-1+\left(\frac 3 4\right)^{\frac 1 m} \right)^m
        \ge \frac 1 4
    \end{align*}

    As the function $\frac 1 {3m}-1+\left(\frac 3 4 \right)^{\frac 1 m}$ is positive for $m=1$ and strictly decreasing towards $0$ with increasing $m$, it is always positive and thus the second claim holds.
\end{proof}

\begin{definition}[Mutually independent events]\label{def:mutually}
Let $A_1, \dots, A_n$ be events. They are said to be mutually independent if and only if, for every subset $I \subseteq [n]$, we have that:

$$
\Proba \left(\inter_{i \in I}A_i\right) = \prod_{i \in I}\Proba(A_i)
$$
\end{definition}

\begin{lemma}\label{lem:mutually}
    If $A_1, \dots, A_n$ are mutually independent events, then we have that, for every subsets $I,J \subseteq [n]$, $I\inter J = \varnothing$:

    $$
    \Proba\left(\inter_{i \in I}A_i\inter_{j\in J}\overline{A_j}\right)=\prod_{i \in I}\Proba(A_i)\prod_{j \in J}(1-\Proba(A_j))
    $$
\end{lemma}

\begin{proof}
    We will show by induction on the size of $J$ that this holds for every $I \subseteq [n]$ such that $I \inter J = \varnothing$. It is clear for the case $\card J = 0$.

    Let $J$ be a nonempty subset of $[n]$ and $I$ a subset of $[n]$ such that $I \inter J = \varnothing$. Let $a\in J$. Then we have that, by the induction hypothesis:
    $$
    \Proba\left(\inter_{i \in I}A_i\inter_{j\in J\setminus \{a\}}\overline{A_j}\right)=\prod_{i \in I}\Proba(A_i)\prod_{j \in J\setminus \{a\}}(1-\Proba(A_j))
    $$

    However, we also have that:
    $$
    \Proba\left(\inter_{i \in I}A_i\inter_{j\in J\setminus \{a\}}\overline{A_j}\right) = \Proba\left(\inter_{i \in I}A_i\inter_{j\in J}\overline{A_j}\right)+\Proba\left(\inter_{i \in I\union\{a\}}A_i\inter_{j\in J\setminus \{a\}}\overline{A_j}\right)
    $$
    Again, by the induction hypothesis, we have that 

    $$
    \Proba\left(\inter_{i \in I\union\{a\}}A_i\inter_{j\in J\setminus \{a\}}\overline{A_j}\right)=\prod_{i \in I\union\{a\}}\Proba(A_i)\prod_{j \in J\setminus \{a\}}(1-\Proba(A_j))
    $$

Piecing everything together, we get that:

\begin{align*}
    \Proba\left(\inter_{i \in I}A_i\inter_{j\in J}\overline{A_j}\right)&=\Proba\left(\inter_{i \in I}A_i\inter_{j\in J\setminus \{a\}}\overline{A_j}\right)-\Proba\left(\inter_{i \in I\union\{a\}}A_i\inter_{j\in J\setminus \{a\}}\overline{A_j}\right)\\
    &=\prod_{i \in I}\Proba(A_i)\prod_{j \in J\setminus \{a\}}(1-\Proba(A_j))-\prod_{i \in I\union\{a\}}\Proba(A_i)\prod_{j \in J\setminus \{a\}}(1-\Proba(A_j))\\
    &=\prod_{i \in I}\Proba(A_i)\prod_{j \in J}(1-\Proba(A_j))
\end{align*}
\end{proof}

\begin{lemma}[Hoeffding's inequality for binomial distributions~~\cite{hoeffding1963probability}]\label{lem:hoeffding}
Let $Y$ be a binomial random variable with parameters $(t,p)$. We then have, for any $x \le tp$:

$$
\Proba(Y \le x) \le \exp\left(-2 t \left(p-\frac x t\right)^2\right)
$$
    
\end{lemma}

\section{Ommitted Lemmas and Proofs from Section~\ref{sec:broadcast}}\label{sec:apptrees}

\begin{lemma}\label{lem:couple}
    For every $t\in \N$, we have that $n\ge N_t \ge X_t$.
\end{lemma}

\begin{proof}
    Note that $N_{t}$ cannot go higher than $n$ because it is the number of nodes informed after round $t$, which is at most $n$. 
    
    We will prove the rest by induction on $t$. For the induction basis note that by definition $N_0=1= X_0$. For the induction step let us assume that $n\ge N_t \ge X_t$ for some $t \in \N$. Consider first the case that $N_{t+1} - N_t < (n-N_t)\cdot \frac {N_t} n$. Since no informed node can become uninformed, we have that $N_{t+1} \ge N_t \ge X_t = X_{t+1}$, as desired. 
    Next consider the case that $N_{t+1} - N_t \ge (n-N_t)\cdot \frac {N_t} n$. Then $N_{t+1} \ge N_t + (n-N_t)\cdot \frac {N_t} n$ and $X_{t+1}=X_t + (n-X_t) \cdot \frac{X_t} n$. As the function $x \mapsto x+(n-x)\frac{x}{n}$ is strictly increasing for $x\le n$, this proves that $N_{t+1} \ge X_{t+1}$, as desired.
\end{proof}

\begin{lemma}\label{lem:couple2}
    For every $t \in \N$, we have that $X_t \ge 1$.
\end{lemma}

\begin{proof}
    We will again show this by induction. For the induction basis note that by definition $1= X_0$. For the induction step let us assume that $X_t \ge 1$ for some $t \in \N$. We then have two cases, either $X_{t+1}=X_t$ and the result holds trivially, or $X_{t+1} = X_t + (n-X_t) \cdot \frac{X_t} n$. Since $1\le X_t \le n$, we have that $X_{t+1} \ge X_t \ge 1$.
\end{proof}

\begin{lemma}\label{lem:Xincreases}
    For every $t\in \N$, we have that $n> X_t$, if $n>1$.
\end{lemma}

\begin{proof}
    We show this claim by induction on $t$. As $n>1$ and $X_1 = 1$, it is trivially true for $t=1$. Assume it is true for $t \in \N$. Then $X_{t+1} \le X_t + (n-X_t) \cdot \frac {X_t} n = n(\frac {X_t} n + \frac{n-X_t} n \frac {X_t} n) <n$, where the last inequality holds by noting that $(\frac {X_t} n + \frac{n-X_t} n \frac {X_t} n)$ is a convex combination of $1$ and $\frac {X_t} n$,  the latter of which being strictly smaller than $1$.
\end{proof}

Essentially, this means that $X_t$ never reaches $n$, and thus that $X_{t+1}$ is always strictly larger than $X_t$ if $N_{t+1}-N_t \geq (n-N_t) \cdot \frac {N_t} n$:

\begin{corollary}\label{cor:Xincreases}
    We have that $X_{t+1}>X_t$ if and only if $N_{t+1}-N_t \geq (n-N_t) \cdot \frac {N_t} n$.
\end{corollary}

\valuesofX*

\begin{proof}
    We show the claim by induction on $t$.
    By definition of $u_0$ we have that $X_{u_0} =1$. Thus the induction basis $X_{u_0}=1 = n-n\left(\frac {n-1} n\right)^{2^0}$ follows. 
    
    For the induction step assume next the result is true for some $t\in \N$. 
    Note that for every $t \in \N$, it holds that $X_{u_{t+1}}=X_{u_t}+(n-X_{u_t})\cdot \frac {X_{u_t}} n$. Indeed, we have that $X_{u_{t+1}}=X_{u_{t+1}-1}=\dots=X_{u_{t}+1}=X_{u_t}+(n-X_{u_t})\cdot \frac {X_{u_t}} n$. Thus,
    
    \begin{align*}
        X_{u_{t+1}}&=X_{u_t}+(n-X_{u_t})\cdot \frac {X_{u_t}} n=n-n\left(\frac {n-1} n\right)^{2^t} \\
        &\qquad+\left(n-n+n\left(\frac {n-1} n\right)^{2^t}\right)\frac{n-n\left(\frac {n-1} n\right)^{2^t}} n\\
        &=n-n\left(\frac {n-1} n\right)^{2^t} +\left(\frac {n-1} n\right)^{2^t}\left({n-n\left(\frac {n-1} n\right)^{2^t}}\right)\\
        &=n-n\left(\frac {n-1} n\right)^{2^{t+1}}
    \end{align*}
\end{proof}

\Nequalsn*
\begin{proof}
Let $Y_{u_i} = \frac{n-X_{u_i}} {n}$. Then
$X_{u_t} > n-1$ iff $Y_{u_t} < 1/n$.
Further 
$Y_{u_0} = \frac{n-1}{n}$ and
$Y_{u_{t+1}} = \frac{n-X_{u_{t+1}}} {n} = \left(\frac {n-1} n\right)^{2^{t+1}} = (Y_{u_t})^2$.
Now note that $Y_{u_t} < 1/n$ iff
$\left(\frac {n-1} n\right)^{2^{t}} < 1/n$ iff
$2^t < \frac{\log (1/n)} {\log (\frac {n-1} n)} = 
\frac{\log n}{\log n - \log (n-1)}$. Now note that ${\log n - \log (n-1)} > 1/n$ and, thus, 
$\frac{\log n}{\log n - \log (n-1)} < n \log n \le n^2$. 
It follows that for $t \le 2 \log (n)$ it holds that
$X_{u_t} > n-1$ and, hence, $N_{u_t} \ge X_{u_t} > n-1.$
As $N_{u_t}$ is an integer, the result follows.\end{proof}

\xbinomial*

\begin{proof}
By Corollary~\ref{sec:5.cor:Xincreases}, we have that $X_{t+1} > X_t$ if and only if $N_{t+1} - N_t \ge (n-N_t)\cdot \frac {N_t} n$. Thus,
$\Proba\left(X_{t+1} > X_{t}\right) = \sum_{\ell\in [n]}
\Proba\left(N_{t+1}-N_t \geq (n-N_t) \cdot \frac {N_t} n\middle| N_t =\ell\right) \Proba(N_t = \ell)$.
Hence we only have to show that $\Proba\left(N_{t+1}-N_t \geq (n-N_t) \cdot \frac {N_t} n\middle| N_t =\ell\right) \ge \frac 1 4$ for every $\ell=N_t$.
Recall that $N_{t+1}-N_t$ follows a binomial distribution with parameters $m = n-N_t$ and $p = \frac{N_t}{n}$. 

If $N_t=n$, the result holds trivially. Otherwise, if $pm \ge 1$, then the result holds by applying Theorem~\ref{thm:binestimate}. If on the other hand $pm \le 1$, then we note that $pm = \frac{(n-N_t) N_t}{n}$ which is at its minimum for $N_t=1$. Therefore $pm \ge \frac{n-1}{n}\ge \frac 1 3$ if $n \ge 2$. The result then holds by applying Lemma~\ref{prob:bin}.\end{proof}

\alltoallBroadcast*
\begin{proof}
    Let $N_t^{(i)}$ be the random variable that represents the number of nodes that are informed after round $t$ of the message given to node $i$. By Theorem~\ref{thm:Broadcast}, we know that $\Proba \left( N_{32\cdot c\cdot \ln n}^{(i)}<n\right) \le n^{-c}$ for every $i \in [n]$. Using a union-bound, we get that:
    $$
    \Proba \left(\Union_{i \in [n]} N_{32\cdot c\cdot \ln n}^{(i)}<n\right) \le n^{-c+1}
    $$
    And thus:
    $$
    \Proba \left(\Inter_{i \in [n]} N_{32\cdot c\cdot \ln n}^{(i)}=n\right) = 1-\Proba \left(\Union_{i \in [n]} N_{32\cdot c\cdot \ln n}^{(i)}<n\right) \ge 1- n^{-c+1}
    $$
\end{proof}
\section{Adversarial Nodes: Trees with Byzantine Nodes}\label{app:Byzantine}
In this section, we will discuss the case where some nodes are \emph{Byzantine}, that is, nodes that can arbitrarily deviate from the protocol. 
These nodes can stop functioning, send wrong messages, and coordinate to make the protocol fail. We will rely on cryptographic tool so that each node can sign and encrypt the message it sends. Then nodes can be confident about the sender of each message and its content and can forward the message along with its unchanged signature to other nodes. We will assume that there are up to $f$ Byzantine nodes, out of a total of $n$ nodes. We require that $f \le \frac 2 3 n-1$. Nodes that are not Byzantine are called honest.

We begin by analyzing Broadcast in this setting. We first give a message to a fixed honest node, and ask the node to forward it to all other honest nodes.
Note the difference between this model and the reliable Broadcast model, where the initial message could be from an honest node or a Byzantine node, and where if the initial message is from a Byzantine node, then the message accepted by each honest node must be the same. 

In our
setting, the best strategy for the Byzantine node is not to forward any message at all. Indeed, they cannot modify the content of a message because they cannot forge any signature, and, thus, their power is limited. Hence, we will  analyze this problem as if Byzantine nodes are just defunct
but the process that chooses the communication network, i.e., the random tree, does not know which nodes are Byzantine and, thus, they are part of the network as before, i.e., the tree still consists of $n$ nodes.

As in the previous section, $I_t$ and $S_t$ will, respectively, be the set of informed and uninformed nodes after round $t$. 
We set $I_0=\{v\}$ and $S_0=[n-f]-\{v\}$, where $v$ is the node that initially holds the message, $N_t=\card{I_t}$ to be the number of informed and honest nodes after $t$ rounds, and $T_t$ to be the tree chosen at random in round $t$. 
For a tree $T$, for each node $p$, $P_T(p)$ is the (unique) parent of node $p$ in $T$, unless $p$ is the root of $T$, in which case $P_T(p)=p$. 
Simplifying the notation, we also use $P_t(p)$ to denote $P_{T_t}(p)$. We use
${A(S, x)}$ where $S$ is a set and $x$ an integer to represent the set of subsets of $S$ of size $x$.

For the rest of the paper, we will denote $\tau = \frac {\log n}{\log \left({1+\frac{n-f}{2n}}\right)}$. Note that $\log n \le \tau \le 4.5 \log n$.

Again, the central lemma will characterize how many new nodes get informed in each round, depending on how many were informed after the previous round. This lemma shows that uninformed nodes get informed independently from each other.
\begin{lemma} \label{sec:5.lem:binom}
    For any $t > 0$,
    $N_{t+1}-N_t$ follows a binomial distribution with parameters $\left(n-f-N_t,  \frac {N_t} n\right)$.
\end{lemma}

This lemma proves that every uninformed node has probability $\frac {N_t} n$ of having an informed parent in round $t+1$, regardless of whether the other uninformed nodes have an uninformed parent.

\begin{proof}
    Let $I_t=\{i_1, \dots , i_{N_t}\}$ and $S_t=\{s_1, \dots , s_{n-f-N_t}\}$. We then have, for any integer $x$:

$$
        \Proba(N_{t+1}-N_t=x | \F_t) =\sum_{J\in {A(S_t, x)}}\Proba\left(\Inter_{y \in J} (P_{t+1}(y) \in I_t) \Inter_{y \in S_t\setminus J}(P_{t+1}(y) \notin I_t) \middle| \F_t\right)
$$

Our goal is to show that the events $P_t(y) \in I_t$ for different 
$y \in S_t$ are mutually independent.
Let us look at the event $\Inter_{y \in J} (P_t(y) \in I_t)$ for any $J\subseteq S_t$ (note that we do not require that $J$ has a specific size here). 
We can then write, indexing $a$ on $J$:

\begin{align*}
\Proba\left(\Inter_{y \in J} (P_{t+1}(y) \in I_t)\middle| \F_t\right)&=\sum_{a\in[N_t]^{\card J}} \Proba\left(\Inter_{y \in J} (P_{t+1}(y) =i_{a_y})\middle| \F_t\right)\\
&= \sum_{a\in[N_t]^{\card J}} \frac{\card{\{ T\in \RT_n:P_T(y) =i_{a_y}, \forall y \in J\}}}{\card{\RT_n}}
\end{align*}

Now consider the forest that is composed of stars whose centers are the $i_{a_y}$ and whose leaves are the nodes $y$. More specifically, consider the forest that contains the edges
 $(i_{a_y}, y), \forall y \in J$. Note that $\card{\{ T\in \RT_n:P_T(y) =i_{a_y}, \forall y \in J\}}$ equals the number of rooted trees that are compatible with this forest.
By Theorem~\ref{thm:treecount}, we have that $\card{\{ T\in \RT_n:P_T(y) =i_{a_y}, \forall y \in J\}}=n^{n-1-\card{J}}$.
This allows us to compute the above probability as follows:

$$
\Proba\left(\Inter_{y \in J} (P_{t+1}(y) \in I_t)\middle| \F_t\right)= \sum_{a\in[N_t]^{\card J}}\frac{n^{n-1-\card{J}}}{n^{n-1}}=\left(\frac {N_t} n\right)^{\card J}
$$

This proves that the events $P_{t+1}(y) \in I_t$ for any two $y \in S_t$ are mutually independent (Definition~\ref{def:mutually}),
each having probability $\frac {N_t} n$. Going back to the first equation of this proof, we can now compute, using Lemma~\ref{lem:mutually}:

\begin{align*}
        \Proba(N_{t+1}-N_t=x| \F_t) &=\sum_{J\in {A(S_t, x)}}\prod_{y \in J}\Proba\left(P_{t+1}(y) \in I_t\middle| \F_t\right)\prod_{y \in S_t\setminus J}\Proba \left(P_{t+1}(y) \notin I_t) \middle| \F_t\right)\\&={{n-f-N_t} \choose x} \left(\frac {N_t} n\right)^{x}\left(1-\frac {N_t} n\right)^{n-f-N_t-x}
\end{align*}
\end{proof}
Our next goal is to show that $N_t = n-f$ with high probability for all $t \ge 32\cdot \tau \cdot c$, for any $c \ge 1$, and $\tau = \frac {\log n} {\log (1+\frac{n-f}{2n})}$. 
To do so we introduce a random variable $X_t$ that we use to lower bound $N_t$ and we show right below that $X_t \le N_t$ for all $t$.
\begin{definition}
    Let $X_t$ be the random variable that is defined as follows:

    \begin{align*}
    X_0&=1\\
        X_{t+1}&=X_t+(n-f-X_t) \cdot \frac {X_t} n &\text{if} \quad N_{t+1}-N_t \geq (n-f-N_t) \cdot \frac {N_t} n\\
        X_{t+1}&=X_t & \text{if} \quad N_{t+1}-N_t < (n-f-N_t) \cdot \frac {N_t} n
    \end{align*}
\end{definition}

\begin{lemma}\label{sec:5.lem:couple}
    For every $t\in \N$, we have that $n-f\ge N_t \ge X_t$.
\end{lemma}

\begin{proof}
    Note that $N_{t}$ cannot go higher than $n-f$ because it is the number of honest nodes informed after round $t$, which is at most $n-f$. 
    
    We will prove the rest by induction on $t$. For the induction basis note that by definition $N_0=1= X_0$. For the induction step let us assume that $n-f\ge N_t \ge X_t$ for some $t \in \N$. Consider first the case that $N_{t+1} - N_t < (n-f-N_t)\cdot \frac {N_t} n$. Since no informed node can become uninformed, we have that $N_{t+1} \ge N_t \ge X_t = X_{t+1}$, as desired. 
    Next consider the case that $N_{t+1} - N_t \ge (n-f-N_t)\cdot \frac {N_t} n$. Then $N_{t+1} \ge N_t + (n-f-N_t)\cdot \frac {N_t} n$ and $X_{t+1}=X_t + (n-f-X_t) \cdot \frac{X_t} n$. As the function $x \mapsto x+(n-f-x)\frac{x}{n}$ is strictly increasing for $x\le n-f$, this proves that $N_{t+1} \ge X_{t+1}$, as desired.
\end{proof}

\begin{lemma}\label{sec:5.lem:couple2}
    For every $t \in \N$, we have that $X_t \ge 1$.
\end{lemma}

\begin{proof}
    We will again show this by induction. For the induction basis note that by definition $1= X_0$. For the induction step let us assume that $X_t \ge 1$ for some $t \in \N$. We then have two cases, either $X_{t+1}=X_t$ and the result holds trivially, or $X_{t+1} = X_t + (n-f-X_t) \cdot \frac{X_t} n$. Since $1\le X_t \le n-f$, we have that $X_{t+1} \ge X_t \ge 1$.
\end{proof}
Next we show that $X_t$ never reaches $n-f$ if there are at least 2 honest nodes.
\begin{lemma}\label{sec:5.lem:Xincreases}
    For every $t\in \N$, we have that $n-f> X_t$, if $n-f>1$.
\end{lemma}

\begin{proof}
    We show this claim by induction on $t$. As $n-f>1$ and $X_1 = 1$, it is trivially true for $t=1$. Assume it is true for $t \in \N$. Then $X_{t+1} \le X_t + (n-f-X_t) \cdot \frac {X_t} n = (n-f)(\frac {X_t} {n-f} + \frac{n-f-X_t} {n-f} \frac {X_t} n) <n-f$, where the last inequality holds by noting that $(\frac {X_t} {n-f} + \frac{n-f-X_t} {n-f} \frac {X_t} n)$ is a convex combination of $1$ and $\frac {X_t} n$,  the latter of which being strictly smaller than $1$.
\end{proof}

It follows  that $X_{t+1}$ is always strictly larger than $X_t$ if $N_{t+1}-N_t \geq (n-f-N_t) \cdot \frac {N_t} n$:

\begin{corollary}\label{sec:5.cor:Xincreases}
    We have that $X_{t+1}>X_t$ if and only if $N_{t+1}-N_t \geq (n-f-N_t) \cdot \frac {N_t} n$.
\end{corollary}
Next we introduce a variable that remembers the rounds that increase $X_t$. 
\begin{definition}\label{sec:5.lem:valueofX}
    Let $u_t \in \N$ be the $t$-th round such that $X_{u_t+1}>X_{u_t}$ and let $u_0 = 0$.
\end{definition}

We now show that once $X_t$ has been increased in $\Theta(\log n)$ rounds, it reaches $n-f$, i.e., all honest nodes have received the message.
Recall that $\tau=\frac{\log n}{\log\left({1+\frac{n-f}{2n}}\right)}$. 

\begin{lemma}\label{sec:5.lem:Nequalsn}
For all $t \ge u_{2\tau}$, $N_t = n-f$.
\end{lemma}

\begin{proof}
    Since $N_t$ is non-decreasing and upper-bounded by $n-f$, it suffices to show that $N_{u_{2\tau}} = n-f$. We will do so by using its lower bound $X_{u_{2 \tau}}$. 
    
    We  first show that $X_{u_\tau} \ge \frac {n-f} 2$. Indeed, for any $t$, if $X_{u_t} \le \frac {n-f} 2$, then :  
    $$X_{u_{t+1}}  = X_{u_t} + (n-f-X_{u_t}) \frac {X_{u_t}} n = X_{u_t} \left( 1+ \frac {n-f-X_{u_t}} n\right) \ge X_{u_t} \left( 1+ \frac {n-f} {2n}\right)$$
And $X_{u_t}$ thus only needs at most $\tau$ steps to increase from $1$ to $\frac{n-f}{2}$.

We now show that $X_{u_{2\tau}}> n-f-1$. Indeed, for any $t$, if $X_{u_t}\ge \frac {n-f} 2$, then we have:
\begin{multline*}
n-f-X_{u_{t+1}}  = n-f-X_{u_t} - (n-f-X_{u_t}) \frac {X_{u_t}} n \\= (n-f-X_{u_t}) \left( 1- \frac {X_{u_t}} n\right) \le (n-f-X_{u_t}) \left( 1- \frac {n-f} {2n}\right)    
\end{multline*}
And $n-f-X_{u_t}$ thus only needs at most  $\tau'$ steps to decrease from $\frac {n-f}2$ to $\frac 1 2$, where:
$$
\tau'=\frac{\log (n-f)}{\log\left(\frac 1 {1-\frac{n-f}{2n}}\right)}=\frac{\log (n-f)}{\log\left(\frac {2n}{n+f}\right)}\le \frac{\log (n-f)}{\log\left(\frac {3n-f}{2n}\right)}=\tau
$$
where the inequality holds since $\frac{3n-f}{2n} = \frac{2n+n-f}{n+f+n-f}\le \frac{2n}{n+f}$.

 Overall we have that $X_{u_{2\tau}} \ge X_{u_{\tau+\tau'}}> n-f-1$. Since $n-f\ge N_{u_{2\tau}}\ge X_{u_{2\tau}}$ by Lemma~\ref{sec:5.lem:couple}, and since $N_{u_{2\tau}}\in \N$, we have that $N_{u_{2\tau}}=n-f$.
\end{proof}

\begin{lemma}
    If $f \le \frac 2 3 n -1$, for every $t \in \N$, we have that $\Proba\left(X_{t+1} > X_{t}\right) \geq \frac{1}{4}$
\end{lemma}

\begin{proof}
By Corollary~\ref{sec:5.cor:Xincreases}, we have that $X_{t+1} > X_t$ if and only if $N_{t+1} - N_t \ge (n-f-N_t)\cdot \frac {N_t} n$. Thus,
$\Proba\left(X_{t+1} > X_{t}\right) = 
\sum_{\ell \in [n]}\Proba\left(N_{t+1}-N_t \geq (n-f-N_t) \cdot \frac {N_t} n|N_t=\ell\right)\Proba(N_t=\ell)$.
It thus suffices to show that for every $\ell$, $\Proba\left(N_{t+1}-N_t \geq (n-f-N_t) \cdot \frac {N_t} n|N_t=\ell\right)\ge \frac 1 4$.
Recall that $N_{t+1}-N_t$ follows a binomial distribution with parameters $m = n-f-N_t$ and $p = \frac{N_t}{n}$. 

If $N_t=n-f$, the result holds trivially. Otherwise, if $pm \ge 1$, then the result holds by applying Theorem~\ref{thm:binestimate}. If on the other hand $pm \le 1$, then we note that $pm = \frac{(n-f-N_t) N_t}{n}$ which is at its minimum for $N_t=1$. Therefore $pm \ge \frac{n-f-1}{n}\ge \frac {1/3 n} n = \frac 1 3$, where we used $f\le \frac 2 3 n -1$. The result then holds by applying Lemma~\ref{prob:bin}.

\end{proof}

    Let $(B_t)_{t \in \N}$ be Bernoulli independent random variables of parameter $\frac{1}{4}$. Let $Z^B_{\le t}=\sum_{z \in [t]} B_z$ and $Z_{\le t}=\sum_{z \in [t]} \1 \left( X_{z+1} > X_z\right)$. 
    
 \begin{corollary}  \label{sec:5.cor:upbound} 
    For any $\ell \in \N$, we have that $\Proba (Z_{\le t} \le \ell) \le \Proba (Z_{\le t}^B \le \ell)$.
\end{corollary}
Using the corollary we can now show that with high probability in the first $32\cdot c\cdot \tau$ $X_t$ will increase at least $2\tau$ often. This then immediately leads to our result.

\begin{lemma}\label{sec:5.lem:manyincreases}
    Let $t=32\cdot c\cdot \tau$ for any $c \ge 1$. Then $\Proba (Z_{\le t} \le 2\tau) \le \frac 1 {n^c}$.
\end{lemma}

\begin{proof}

Note that $Z_{\le t}^B$ is a binomial distribution of parameters $(t, \frac 1 4)$. Using Hoeffding's inequality, we have that:

$$\Proba (Z_{\le t}^B \le 2\tau)\leq \exp \left(-2t\left(\frac 1 4-{\frac {2\tau}{t}}\right)^{2}\right)\le \exp \left(-2\cdot 32 c\ln n \left(\frac 1 4-{\frac {2}{16}}\right)^{2}\right)=n^{-c}$$

    Where we used that $\tau \le \ln n $. %
Corollary~\ref{sec:5.cor:upbound} then gives the desired result.
\end{proof}

We now have all the tools to prove the main theorem of this section, which we state here:

\secfivethmbroadcast*

\begin{proof}

    By Lemma~\ref{sec:5.lem:manyincreases}, we have that, with probability $p \le 1-\frac 1 {n^c}$, $X_{t+1}>X_t$ for at least $2\tau$ many rounds within the $32\cdot c\cdot \tau$ first rounds. Recall that $u_{2\tau}$ is the $2\tau$-th round where $X_{t+1}>X_t$. We thus have that  $\Proba (Z_{\le 32\cdot c\cdot \tau} > 2 \tau) = \Proba \left( u_{2\tau} \le 32 \cdot c \cdot \tau \right)\ge 1-n^{-c}$. But, by Lemma~\ref{sec:5.lem:Nequalsn} the event $u_{2\tau} \le 32\cdot c\cdot \tau$ 
    implies the event $N_{32\cdot c\cdot \tau}=n-f$, therefore $\Proba \left( N_{32\cdot c\cdot \tau}=n-f\right) \ge 1-n^{-c}$. Upper-bounding $\tau \le 4.5 \log n$ gives the desired result.
\end{proof}

We now use this result to get a similar result for All-to-all Broadcast. Using a union-bound, we obtain:

\byzalltoall*

\begin{proof}
    Let $N_t^{(i)}$ be the random variable that represents the number of nodes that are informed after round $t$ of the message given to node $i$. By Theorem~\ref{sec.5.thm:Broadcast}, we know that %
    $\Proba \left( N_{32\cdot c\cdot \tau}^{(i)}<n-f\right) \le n^{-c}$ for every $i \in [n]$. Using a union-bound, we get that:
    $$
    \Proba \left(\Union_{i \in [n]} N_{32\cdot c\cdot \tau}^{(i)}<n-f\right) \le n^{-c+1}
    $$
    And thus:
    $$
    \Proba \left(\Inter_{i \in [n]} N_{32\cdot c\cdot \tau}^{(i)}=n-f\right) = 1-\Proba \left(\Union_{i \in [n]} N_{32\cdot c\cdot \tau}^{(i)}<n-f\right) \ge 1- n^{-c+1}
    $$
\end{proof}

This allows us, e.g.,~to implement algorithms that run on a clique %
in a synchronous setting in our sparser graph.
Indeed, each round of communication of a clique can be simulated by $32\cdot c \cdot \tau$ rounds of Uniformly Random Trees with high probability, since All-to-All Broadcast needs $32\cdot c \cdot \tau$ rounds to complete with high probability.
Essentially, if an algorithm runs in $T$ time, with $T \le n^{c-1}$, in a clique network, we can implement it with high probability in $32\cdot c \cdot T \cdot \tau$ rounds in the Uniformly Random trees network, which is essentially a logarithmic overhead.
The only caveat is that if $T$ is too large, i.e. $T > n^{c-1}$, the probability of at least one of the $T$ All-to-All Broadcast rounds failing can become close to 1. To circumvent this, we restrict ourselves to the case where $T$ is a small enough polynomial in $n$.

\cliquereduc*

\begin{proof}
    The algorithm $\Alg'$ works as follows: Each round of $\Alg$ will be simulated using $32\cdot c\cdot \tau $ rounds of Uniformly Random Trees. 
    Each round in $\Alg$ can be seen as (1) a \emph{computation step}, where each node decides what message to send to every other node, and (2) a \emph{communication step}, where each node sends the messages and receives the messages the other nodes have sent it. 
    In $\Alg'$ the computation step is unchanged, except that each node uses the PKI to sign and encrypt the messages. 
    The communication step on the other hand is extended over $32\cdot c\cdot \tau $ rounds, in which every (honest) node simply forwards all received messages to all its out-neighbors. 
    Theorem~\ref{sec.5.thm:alltoall} ensures that with probability at most $n^{1-c}$, at least one honest node fails to send a message to all other honest nodes. 
    By a union bound over the $T$ rounds, the probability that at least one message fails to be delivered is at most $\alpha n^{1+x-c}$. By taking its complement, the probability that all messages are delivered in time before the next round's computation step is $p' \ge 1-\alpha n^{x+1-c}$.
\end{proof}

We now give two applications of this theorem, namely Reliable Broadcast and Byzantine Consensus.

\reliableBroadcast*

\begin{proof}
    Dolev and Strong~\cite{reliable} have given an algorithm that solves reliable Broadcast, is robust to $f$ Byzantine nodes, and runs in $T=f+1$ rounds. Since $T\le n$, we can apply Theorem~\ref{sec.5.thm:reduction} with $x=1, \alpha = 1$, and we get the desired result.
\end{proof}

\ConsensusByzantine*

\begin{proof}
   Berman, Garay and Perry~\cite{king} have given an algorithm (known as the King's algorithm) that solves reliable Broadcast, is robust to $f$ Byzantine nodes, and runs in $T=3(f+1)$ rounds. Since $T\le 2n$, we can apply Theorem~\ref{sec.5.thm:reduction} with $x=1, \alpha = 2$, and we get the desired result.
\end{proof}

\section{Ommitted proofs of Section~\ref{sec:adversary}}
\label{appendix:adversary}

\distributiondomination*

\begin{proof}
Saying that $N_t^{(1)}$ stochastically dominates $N^{(2)}_t$ is equivalent to saying that $I_t^{(1)}$ stochastically dominates $I^{(2)}_t$.
    By Theorem~\ref{thm:dominancecoupling}, 
    there exists a coupling $(\hat{I}^{(1)}_t, \hat{I}_t^{(2)})$ of $I_t^{(1)}$ and $I_t^{(2)}$ such that $\Proba \left(\card{\hat{I}^{(1)}_t} \ge \card{\hat{I}_t^{(2)}}\right) =1 $. Consider that coupling and 
    let $\beta$ be a bijection from $[n]$ to $[n]$ such that $\beta(\hat{I}_t^{(2)}) \subseteq \hat{I}^{(1)}_t$. 
    By slight abuse of notation we naturally extend $\beta$ to also give a bijection between edges, by setting $\beta(u,v) = (\beta(u), \beta(v))$.
    For any $t'\ge t$, let $E^{t'}_1$ and $E^{t'}_2$ be respectively the edges chosen by the adversary in round $t'$ after choosing $E_1$, respectively $E_2$, in round $t$. Let $T^{t'}_1$ and $T^{t'}_2$ be respectively the trees in round $t'$ containing $E^{t'}_1$ and $E^{t'}_2$, respectively. 
    
    We introduce a coupling 
    $(\hat T^{t'}_1, \hat T^{t'}_2)$ such that if $E^{t'}_1=\beta(E^{t'}_2)$ then $\Proba(\hat T^{t'}_1= \beta( \hat T^{t'}_2))=1$ as follows: If $E^{t'}_1=\beta(E^{t'}_2)$, then note that $\beta$ induces a bijection between $\RT_n^{(2)}=\{T \in \RT_n:E^{t'}_2 \subseteq T\} $ and $\RT_n^{(1)}=\{T \in \RT_n:E^{t'}_1 \subseteq T\}$. In this case to define the coupling we choose $\hat T^{t'}_{2}$ uniformly at random from $\RT_n^{(2)}$ and set $\hat T^{t'}_1=\beta(\hat T^{t'}_2)$. 
    Thus, $\Proba(\hat T^{t'}_1= \beta( \hat T^{t'}_2))=1$.
      
    Otherwise, i.e., if $E^{t'}_1 \ne \beta(E^{t'}_2)$, to define the coupling we choose $\hat T^{t'}_2$ uniformly at random from $\RT_n^{(2)}$ and, independently, $\hat T^{t'}_1$ uniformly at random from $\RT_n^{(1)}$. 
    
    Let $\hat I_1^{t'}$ and $\hat I_2^{t'}$ be the set of informed nodes we get after round $t'$ if the trees in round $t'$ were $\hat T^{t'}_1,$ respectively $\hat T^{t'}_2$.
    We will show that $\Proba \left(\beta(\hat I_2^{t'}) \subseteq \hat I_1^{t'}\right)=1$ for all $t' \ge t$.

    Let us now assume that the sequence $(E^{t'}_1)_{t' > t}$ is an optimal strategy for the adversary after choosing $E_1$ in round $t$ and let us call this sequence (including $E_1$ in round $t$) $\sigma_1$. Then consider the sequence $\sigma_2$ which chooses $E_2$ in round $t$ then $E_2^{t'} := \beta^{-1}(E^{t'}_1)$ in rounds $t'\ge t$ and compare it with $\sigma_1$.
    Thus, $E^{t'}_1 = \beta(E_2^{t'})$ for all $t' > t$, which implies that $\Proba(\hat T^{t'}_1= \beta( \hat T^{t'}_2))=1$.
    
    Recall that $\sigma_1$ is optimal after choosing $E_1$ in round $t$.
    We show by induction that $\sigma_2$ is a better overall strategy for the adversary 
   than $\sigma_1$. Indeed, we can show by induction on the number of rounds $t' \ge t$ that $\Proba \left(\beta(\hat I_2^{t'}) \subseteq \hat I_1^{t'}\right)=1$. The induction basis is trivial as
   $\beta(\hat I_2^{t}) \subseteq \hat I_1^{t}$.
   For the induction step, note that for any $v \in \hat I_2^{t'}$, either (1) $v \in \hat I_2^{t'-1}$ which, using the induction assumption implies with probability 1 that  $\beta(v) \in \beta(\hat I_2^{t'-1}) \subseteq \hat I_1^{t'-1} \subseteq \hat I_1^{t'}$, or (2) $P_{\hat T_2^{t'}}(v) \in \hat I_2^{t'-1}$. 
   
   As $\Proba(\hat T^{t'}_1= \beta( \hat T^{t'}_2))=1$, Case (2) implies that with probability $1$, $\beta (P_{\hat T_2^{t'}}(v)) = P_{\hat T_1^{t'}}(\beta(v))$. As $P_{\hat T_2^{t'}}(v) \in \hat I_2^{t'-1}$ and by induction, with probability 1, $\beta (\hat I_2^{t'-1}) \subseteq \hat I_1^{t'-1}$, it follows that 
   with probability 1, it holds that $ P_{\hat T_1^{t'}}(\beta(v))=\beta(P_{\hat T_2^{t'}}(v)) \in \hat I_1^{t'-1}$ and, thus, $\beta(v) \in  \hat I_1^{t'}$. Hence,
   in both cases, with probability 1,  $\beta(I_2^{t'}) \in \hat I_1^{t'} $. 
   
   Now note that $\Proba \left(\beta(\hat I_2^{t'}) \subseteq \hat I_1^{t'}\right)=1$ implies that if $t'$ is the smallest round at which Broadcast completes after the adversary chooses $E_2$, that is, if $\card {\hat I_2^{t'-1}}=n$, then Broadcast completes in a no later round if the adversary chooses $E_1$.
\end{proof}

\lembij*

\begin{proof}
    Let $V(U)$ be the set of nodes of $U$ and let $|V(U) \cap S_{t-1}| = \ell$.
    To show the lemma we will give a bijection $b$ that maps the $\ell$ uninformed nodes of $V(U)$ to the $\ell$ first nodes of $U$ in bfs-order (on $U$) and the informed nodes to the remaining nodes of $U$ such that
    one of the nodes $u \in S_{t-1}$ with $P_U(u) \in I_{t-1}$ becomes the root of $U'$. The resulting tree will be the correction $U'$. As a result of this bijection every uninformed node of $U'$ has only uninformed ancestors and, thus, $U'$ is non-increasing. By construction, $U$ and $U'$ are isomorphic.

   More formally,
    if $U$ is increasing, then there exists an edge $(i,s)$ such that $i \in I_{t-1}, s \in S_{t-1}$. Let $\pi$ be a bijection from $[\card{V(U)}]$ to $V(U)$ such that $\pi(1)=s, \{\pi(2), \dots, \pi(\ell)\} \subset S_{t-1}$, and $\{\pi(\ell+1), \dots, \pi(\card{V(U)})\}\subseteq I_{t-1}$. 
    Furthermore, let $\rho$ be a bijection from $[\card{V(U)}]$ to $V(U)$ such that $\rho(j)$ is the $j$-th node encountered in a breadth-first traversal starting at the root of $U$. Then let $b=\pi \circ \rho^{-1}$. We will show next that the tree $U'$, whose set of edges is $\{(b(u), b(v)): (u,v) \in U\}$ is a correction of $U$. We first need to argue that $U'$ is a tree.
    This is the case as $U'$ is a graph over the same nodes as $U$ and 
    for every $(b(u),b(v)) \in U'$, $u$ is encountered in a BFS before $v$ in $U$, thus $\rho^{-1}(u)<\rho^{-1}(v)$. 
    Now we are ready to show that $U'$ fulfills the conditions of being a correction of $U$:
    (1) It follows that $U'$ cannot contain a cycle. As $U'$ is a graph on $|U|$ nodes with $|U|-1$ edges, it follows that $U'$ is a tree. Additionally, $U'$ isomorphic to $U$ as $b$ gives the required bijection between $U$ and $U'$.
    
    (2) Furthermore, $\rho^{-1}(u)<\rho^{-1}(v)$ implies that if $\rho^{-1}(u)\ge \ell$, we also have $\rho^{-1}(v) \ge \ell$ for any $\ell$. 
    Specifically for $\ell = |S_{t-1}|$ this means that if $b(u) \in I_{t-1}$, then $b(v) \in I_{t-1}$. Thus $U'$ is non-increasing in round $t$.

    (3) By construction we made sure that $s \in S_{t-1}$ with $P_U(s) \in I_{t-1}$ is the root of $U'$.
\end{proof}

\correction*
\begin{proof}
    We will build a bijection $\pi$ from $\RT_n^{(2)}=\{T\in \RT_n: E_2 \subseteq T\}$ %
    to $\RT_n^{(1)}=\{T\in \RT_n: E_1 \subseteq T\}$ such that for every
    $s \in S_{t-1}$ and any $ T\in \RT_n^{(2)}$ with $P_T(s) \in I_{t-1}$, we have that $P_{\pi(T)}(s) \in I_{t-1}$. Hence, $\pi(T)$ has more uninformed nodes that become informed than $T$. We will use this property to show that $N_t^{(1)}$ stochastically dominates $N_t^{(2)}$.

    To do so, let $b$ be the bijection that achieves the isomorphism of the proof of Lemma~\ref{lem:bij}
    from $U$ to $U'$. $\pi(T)$ is constructed in a way such that all nodes have the same parents as in $T$, unless they are in $U'$. More specifically, we let $\pi (T) :=\pi_b(T)$ where $\pi_b(T)$ is the tree obtained from $T$ by replacing every edge $(u,v) \in T$ as follows:
    \begin{itemize}
        \item if $u,v \in U'$, then replace it with the edge $(b^{-1}(u), b^{-1}(v))$.
        \item if $u \notin U'$, $v \notin U'$, then keep it the same.
        \item if $u \in U'$, $v \notin U'$, then keep it the same.%
        \item if $u \notin U', v \in U'$, then replace it with $(u, b^{-1}(v))$.
    \end{itemize}

    We clearly have that $U \subseteq E_1 \subset \pi(T)$ and $U' \subseteq E_2 \subset T$. Also, for any node $v$, the path from the root to $v$ in $T$ can be transformed into a path from the root in $\pi(T)$ by replacing the subpath $P=u_0, \dots, u_\ell$ that is in $U'$ with the path from $b^{-1}(u_0)$ to $u_\ell$ in $U$. Hence $\pi (T)\in \RT_n^{(1)}$. 
    Since $\pi_{b^{-1}}$ is clearly an inverse of $\pi_b$, we have that $\pi$ is a bijection.

    Let $s\in S_{t-1}$ be such that $P_T(s) \in I_{t-1}$. If $s \notin U'$, then it has the same parent in $T$ and in $\pi (T)$. If $s \in U'$, which is a non-increasing tree, then by the fact that the parent of $s$ in $T$ belongs to $I_{t-1}$ it follows that $s$ is the root of $U'$, and, thus, that its parent does not belong to $U'$. By the definition of a correction it follows that the parent of $s$ in $U$ is informed. 
    As $U \subseteq \pi(T)$ the parent of $s$ in $\pi(T)$ is a node of $I_{t-1}$.

    We, therefore, have that, for every $x \in \N$:
    \begin{align*}        
    \Proba(N^{(2)}_t-N_{t-1} \ge x | \F_{t-1}) &= \frac {\card {\{T \in \RT_n^{(2)}: \card {\{s \in S_{t-1}: P_T(s) \in I_{t-1}\}} \ge x\}}}{\card {\RT_n^{(2)}}}\\
    &\le \frac {\card {\{T \in \RT_n^{(2)}: \card {\{s \in S_{t-1}: P_{\pi(T)}(s) \in I_{t-1}\}} \ge x\}}}{\card {\RT_n^{(1)}}}\\
    &\le \frac {\card {\{T \in \RT_n^{(1)}: \card {\{s \in S_{t-1}: P_{T}(s) \in I_{t-1}\}} \ge x\}}}{\card {\RT_n^{(1)}}}
    \\&=\Proba(N^{(1)}_t-N_{t-1} \ge x|\F_t{t-1})
    \end{align*}
     The lemma now follows from the Distribution Domination Lemma (Lemma~\ref{lem:domination}).
\end{proof}

\merging*

\begin{proof}
    We will show  that for any $x \in \N$, we have that $\Proba (N^{(1)}_t-N_{t-1} = x| \F_{t-1})\ge\Proba (N^{(2)}_t-N_{t-1} = x| \F_{t-1})$. 
    Then the result will follow from the Distribution Domination Lemma.
    
    In the following let $S$ be a set of uninformed nodes $s_1, \dots, s_{|S|}$, let $\eta_j$ for $1 \le j \le |S|$ be the number of informed nodes in the connected component of $s_j$ in $E_1$ and let $\eta(S)$ be the number of informed nodes that do not belong to the connected component of any $s_j$.

We will analyze the value of $\Proba (\inter_{s\in S}P_{t}(s)\in I_{t-1}| \F_{t-1})$ when the adversary chooses $E_1$, and when it chooses $E_2$. Then two cases can arise: Either the value of $\sum_{\card S = \ell} \Proba (\inter_{s\in S}P_{t}(s)\in I_{t-1}| \F_{t-1})$ is equal whether the adversary chooses $E_1$ or $E_2$, and this for every $\ell$, and the result will follow Lemma~\ref{prob:hard}, or $\Proba (\inter_{s\in S}P_{t}(s)\in I_{t-1}| \F_{t-1})$ will be the same whether the adversary chooses $E_1$ or $E_2$ for every set $S$ except if $S$ includes a particular node, where there will be a constant factor difference between the two values, and the result will then follow from Lemma~\ref{prob:easydom}.

As all trees in $E_1$ (respectively $E_2$) are non-increasing, the parent in $E_1$ (respectively $E_2$) of every non-root node $s \in S$   is uniformed. Thus, if there exists a node  in $S$ that is not a root of $E_1$ (respectively $E_2$) then $\Proba (\inter_{s\in S}P_{t}(s)\in I_{t-1}| \F_{t-1})=0.$
Hence, we only need to analyze the setting where all nodes of $S$ are roots in $E_1$.    

    Case A: Let us first consider the case $r \in I_{t-1}$ in which case the merge of $U$ and $U'$ makes all children of $r$ to children of $r'$. In that case, $r \notin S$ and let $\gamma=\card U -1$. 
    We have two subcases: (A1) If $r' \notin S$, then the number of informed nodes in none of the components with roots in $S$, $\eta(S)$, remains unchanged. 
    It follows from Lemma~\ref{lem:iinformed} that $\Proba (\inter_{s\in S}P_{t}(s)\in I_{t-1}| \F_{t-1})$ is the same whether the adversary chooses $E_1$ or $E_2$. (A2) If $r' \in S$, wlog assume that $r'=s_1$. 
    Then we have that $\Proba (\inter_{s\in S}P_{t}(s)\in I_{t-1}| \F_{t-1})=\frac{\eta(S)(N_{t-1})^{\card {S}-1}}{n^{\card {S}}}$ if the adversary chooses $E_1$, while $\Proba (\inter_{s\in S}P_{t}(s)\in I_{t-1}| \F_{t-1})=\frac{(\eta(S)-\gamma)(N_{t-1})^{\card {S}-1}}{n^{\card {S}}}$ if the adversary chooses $E_2$.
    Applying Lemma~\ref{prob:easydom} where we set $X_s=P_t(s)\in I_{t-1}$ if the adversary chooses $E_1$, and $Y_s=P_t(s)\in I_{t-1}$ if the adversary chooses $E_2$, and $\alpha = \frac{\eta(S)}{\eta(S)-\gamma}$, we have that $ N_t^{(1)} - N_{t-1}=\sum_{s\in S_{t-1}}X_s$ stochastically dominates $ N_t^{(2)} - N_{t-1}=\sum_{s\in S_{t-1}}Y_s$.
    The result follows from the Distribution Domination Lemma.

    Case B: Let us now look at the case where $r \notin I_{t-1}$. In this case the merge of $U$ and $U'$ makes all children of $U'$ children of $U$. 

    We consider again two cases:
    (B1) If $r' \in I_{t-1}$, then  this case is symmetric to Case (A1) and the same proof as above applies.
    
    (B2) If $r' \notin I_{t-1}$, for any $\ell \in \N$, we have that:
    
    \begin{align*}
      \sum_{\card S = \ell}\Proba (\inter_{s\in S}P_{t}(s)\in I_{t-1}| \F_{t-1})  &= \sum_{\card S = \ell: r,r'\notin S}\Proba (\inter_{s\in S}P_{t}(s)\in I_{t-1}| \F_{t-1})+\sum_{\card S = \ell: r,r'\in S}\Proba (\inter_{s\in S}P_{t}(s)\in I_{t-1}| \F_{t-1})\\&+\sum_{\card S = \ell-1: r,r'\notin S}\Proba (\inter_{s\in S\union \{r\}}P_{t}(s)\in I_{t-1}| \F_{t-1})+\Proba (\inter_{s\in S\union \{r'\}}P_{t}(s)\in I_{t-1}| \F_{t-1})
    \end{align*} 
    
    We need to analyze the three sums. 
    
    For the first two sums, where both $r$ and $r'$ or neither belong to $S$, the number of informed nodes in none of the components with root in $S$, $\eta(S)$, is not different in $E_1$ and in $E_2$ and, thus, 
    Lemma~\ref{lem:iinformed} implies that $\Proba (\inter_{s\in S}P_{t}(s)\in I_{t-1})$ does not change whether the adversary chooses $E_1$ or $E_2$.
    
    For the third sum, let $\gamma, \gamma'$ be respectively the number of informed nodes in the component of $r, r'$ in $E_1$. Let us first consider the case where the adversary chooses $E_1$. We have, by Lemma~\ref{lem:iinformed}:
    $$\Proba (\inter_{s\in S\union \{r\}}P_{t}(s)\in I_{t-1}| \F_{t-1})= \frac{(\eta(S)-\gamma)(N_{t-1})^{\card {S}}}{n^{\card {S}+1}}$$
    
    and
    
    $$\Proba (\inter_{s\in S\union \{r'\}}P_{t}(s)\in I_{t-1}| \F_{t-1})= \frac{(\eta(S)-\gamma')(N_{t-1})^{\card {S}}}{n^{\card {S}+1}}$$
    
    therefore:
    
    $$\Proba (\inter_{s\in S\union \{r\}}P_{t}(s)\in I_{t-1}| \F_{t-1})+\Proba (\inter_{s\in S\union \{r'\}}P_{t}(s)\in I_{t-1}| \F_{t-1})= \frac{(2\eta(S)-\gamma-\gamma')(N_{t-1})^{\card {S}}}{n^{\card {S}+1}}$$
    
    If the adversary chooses $E_2$, then $r$ has $\gamma+\gamma'$ informed nodes in its component in $E_2$, while $r'$ has $0$ of them. by Lemma~\ref{lem:iinformed}:
    
    $$\Proba (\inter_{s\in S\union \{r\}}P_{t}(s)\in I_{t-1}| \F_{t-1})= \frac{(\eta(S)-\gamma-\gamma')(N_{t-1})^{\card {S}}}{n^{\card {S}+1}}$$
    
    and
    
    $$\Proba (\inter_{s\in S\union \{r'\}}P_{t}(s)\in I_{t-1}| \F_{t-1})= \frac{\eta(S)(N_{t-1})^{\card {S}}}{n^{\card {S}+1}}$$
    
    therefore:
    
    $$\Proba (\inter_{s\in S\union \{r\}}P_{t}(s)\in I_{t-1}| \F_{t-1})+\Proba (\inter_{s\in S\union \{r'\}}P_{t}(s)\in I_{t-1}| \F_{t-1})= \frac{(2\eta(S)-\gamma-\gamma')(N_{t-1})^{\card {S}}}{n^{\card {S}+1}}$$

    Therefore, $\sum_{\card S = \ell}\Proba (\inter_{s\in S}P_{t}(s)\in I_{t-1})$ has the same value whether the adversary chooses $E_1$ or $E_2$. Applying Lemma~\ref{prob:hard} where we set $X_s=P_t(s)\in I_{t-1}$ if the adversary chooses $E_1$, and $Y_s=P_t(s)\in I_{t-1}$ if the adversary chooses $E_2$, we have that $ N_t^{(1)} - N_{t-1}=\sum_{s\in S_{t-1}}X_s$ and $ N_t^{(2)} - N_{t-1}=\sum_{s\in S_{t-1}}Y_s$ have the same distribution. The result follows from the Distribution Domination Lemma.

\end{proof}

\binomial*
\begin{proof}
    Let $I_t=\{i_1, \dots , i_{N_t}\}$ and $S_t=\{s_1, \dots , s_{n-N_t}\}$ such that $i_1, \dots, i_\eta$ are nodes of $U$, $s_1, \dots, s_{\sigma-1}$ are uninformed nodes of $U$ that are not the root, and $s_\sigma$ is the root of $U$. As $U$ is non-increasing  $s_1, \dots s_{\sigma-1}$ cannot get informed in round $t+1$. As the parent of $s_{\sigma}$ does not belong to $U$, it cannot belong to $i_1, \dots, i_\eta$.
    We will show that the events \emph{uninformed node $s$ gets informed in round $t+1$} 
 for different uninformed nodes $s \in  [\sigma, n-N_t ]$ are mutually independent.
 To do so we take some $J\subseteq [\sigma, n-N_t]$
and analyze the event $\Inter_{y \in J} (P_t(s_y) \in I_t)$,  We distinguish two cases.

Case 1: If $\sigma \notin J$, then it holds that

\begin{align*}
\Proba\left(\Inter_{y \in J} (P_{t+1}(s_y) \in I_t)\middle| \F_t\right)&=\sum_{a\in[N_t]^{\card J}} \Proba\left(\Inter_{s_y \in J} (P_{t+1}(s_y) =i_{a_y})\middle| \F_t\right)\\
&= \sum_{a\in[N_t]^{\card J}} \frac{\card{\{ T'\in \RT_n:P_U(s_y) =i_{a_y}, \forall y \in J \land U \subset T'\}}}{\card{\{ T'\in \RT_n:U \subset T'\}}}
\end{align*}

By Theorem~\ref{thm:treecount}, we have that $\card{\{ T'\in \RT_n:U \subset T'\}}=n^{n-1-k}$, and $\card{\{ T'\in \RT_n:P_U(s_y) =i_{a_y}, \forall y \in J \land U \subset T'\}}=$ $n^{n-1-k-\card{J}}$. Therefore it follows that: 

$$
\Proba\left(\Inter_{y \in J} (P_{t+1}(s_y) \in I_t)\middle| \F_t\right)= \sum_{a\in[N_t]^{\card J}}\frac{n^{n-1-\card{J}}}{n^{n-1}}=\left(\frac {N_t} n\right)^{\card J}
$$

Case 2: If $\sigma \in J$, we have to take extra care of node $s_\sigma$:

\begin{align*}
\Proba\left(\Inter_{y \in J} (P_{t+1}(b_y) \in I_t)\middle| \F_t\right)&=\sum_{a\in[N_t]^{\card J-1}\times [\eta+1,N_t]} \Proba\left(\Inter_{s_y \in J} (P_{t+1}(b_y) =i_{a_y})\middle| \F_t\right)\\
&= \sum_{a\in[N_t]^{\card J-1}\times [\eta+1,N_t]} \frac{\card{\{ T'\in \RT_n:P_U(s_y) =i_{a_y}, \forall y \in J \land U \subset T'\}}}{\card{\{ T'\in \RT_n:U \subset T'\}}}
\end{align*}

By Theorem~\ref{thm:treecount}, we have that $\card{\{ T'\in \RT_n:U \subset T'\}}=n^{n-1-k}$, and $\card{\{ T'\in \RT_n:P_U(s_y) =i_{a_y}, \forall y \in J \land U \subset T'\}}=$ $n^{n-1-k-\card{J}}$. Therefore we have that: 

$$
\Proba\left(\Inter_{y \in J} (P_{t+1}(s_y) \in I_t)\middle| \F_t\right)= \sum_{a\in[N_t]^{\card J-1}\times [\eta+1,N_t]}\frac{n^{n-1-\card{J}}}{n^{n-1}}=\left(\frac {N_t} n\right)^{\card J-1}\frac {N_t-\eta} n
$$

This proves that the events $P_{t+1}(s_y) \in I_t$ are mutually independent for every $y\ge \sigma$, each having probability $\frac {N_t} n$, except if $y=\sigma$, which has probability $\frac {N_t-\eta} n$.

\end{proof}

\section{Beyond Trees: Broadcast and Consensus in directed Erdős–Rényi graphs}
\label{app:erdos}

In this section, we will study the case where in each round, the communication graph is a directed Erdős–Rényi Graph with $m$ edges. 
Choosing such a graph can be seen as choosing without replacement $m$ edges in the graph.
To do so, we will analyze three different schemes.

In \emph{scheme 1}, in each round, $m$ edges are chosen uniformly at random \emph{without} replacement among the $n^2$ possible edges. This is equivalent to our model.

In \emph{scheme 2}, in each round, $m$ edges are chosen uniformly at random \emph{with} replacement among the $n^2$ possible edges. This can only result in more rounds than scheme 1 as fewer disjoint edges are chosen in each round compared to scheme 1.

In \emph{scheme 3}, we start with a unique informed node $1$. 
There are two types of phases for a total of $2\ceil{\log \frac n 2 }$ phases. In each phase $i$ with $1 \le i \le \ceil{\log \frac n 2 }$%
, we have $N_i=2^{i-1}$ informed nodes in $I_i$ at the beginning of the phase and the goal is to double that number within the phase. We set $E=\varnothing$ and add to $E$ one edge at a time, sampled with replacement, until $\card {I_i \union \Out_E(I_i)} = \min \{2 ^i, \ceil{\frac n 2}\}$. We then set $I_{i+1}=I_i \union \Out_E(I_i)$.
Note that at the end of phase $i = \ceil{\log \frac n 2}$, $N_{i+1} = \ceil{\frac n 2}$. 
Note that this scheme is independent of $m$.

Then in each phase $i=2\ceil{\log \frac n 2 }-j$ with $0 \le j < \ceil{\log \frac n 2 }$, we initially have $N_i$ informed nodes in $I_i$, and  $\min\{2^{j}, \floor{\frac n 2}\} $ uninformed nodes in $S_i$ and the goal is to halve the number of uninformed nodes in each phase. We set $E=\varnothing$ and add to $E$, one edge at a time sampled with replacement, until $\card {S_i \setminus \Out_E(I_i)} =  2^{j-1}$. We then set $I_{i+1}=I_i \union \Out_E(I_i)$. %

Let us intuitively compare scheme 2 and scheme 3 and assume initially that $m = 1$. Then scheme 2 in each round chooses an edge uniformly at random, forwards information along it if its source is an informed node, and then moves on to the following round. Scheme 3 on the other hand, will continue to sample edges until enough progress can be made along those edges, and then forward information along those edges all at once, before moving to the next phase. Overall, scheme 3 will sample more edges than scheme 2, as in scheme 2 any progress is made as soon as possible, whereas in scheme 3 progress is only made when checkpoints are reached.

In this section, we will expand this intuition for any $m \in [n^2]$, give upper bounds on the number of edges sampled by scheme 3, then use those results to get an upper bound on the number of rounds needed by scheme 2, and use it to give an upper bound for scheme 1. Note that in scheme 3, we only count the number of edges sampled, and we will not be talking about rounds in that scheme. We thus start by analyzing scheme 3:

\begin{lemma}\label{lem:phasesbound}
    For any $c\ge 1$, any phase $i \le \ceil{\log \frac n 2}$ needs at most $ 8\cdot c \cdot \max\{\ln n, 2^{i-2}\} \frac{n}{2^{i-2}}$ sampled edges to complete with probability $p \ge 1-n^{-c}$. 
\end{lemma}

\begin{proof}
    We first note that in phase $i$ with  $i \le \ceil{\log \frac n 2}$, the probability that an edge that is sampled increases $I_i \union \Out_E(I_i)$  is at least $\frac {2^{i-2}} n$. Indeed, $\card{I_i \union \Out_E(I_i)} \le \frac n 2$ (otherwise the phase would have ended) and there are $\card{I_i}\cdot \card{[n]\setminus (I_i \union \Out_E(I_i))} \ge 2^{i-2} n$ edges that can increase $I_i \union \Out_E(I_i)$. Thus, picking any of those edges, out of $n^2$ possible ones, has probability at least $\frac {2^{i-2}} n$.

    Next, we remark that if we sample $\frac n {2^{i-2}}$ edges, then the probability that at least one of those edges increases $I_i \union \Out_E(I_i)$ is at least $\frac 1 2$. Indeed, the probability that all of those edges do not make $I_i \union \Out_E(I_i)$ larger is $(1-\frac {2^{i-2}} n)^{\frac n {2^{i-2}}} \le e^{-1} \le \frac 1 2$. We will, thus, group the sampled edges into disjoint ``buckets'' of $\frac n {2^{i-2}}$ consecutively sampled edges.

    We then use the fact that we only need $2^{i-1}$ edges that increase $I_i \union \Out_E(I_i)$  to end phase $i$. If we take $8\cdot c \cdot \max\{\ln n, 2^{i-2}\}$ buckets of $\frac n {2^{i-2}}$ edges each, then applying Hoeffding's inequality, we have:
    \begin{multline*}
        \Proba(\card{I_i \union \Out_E(I_i)} \le 2^{i}-1) \le \exp\left( - 2 \cdot 8 \cdot c \cdot \max\{\ln n, 2^{i-2}\} \left(\frac 1 2 - \frac{2^{i-1}} {8 \cdot c \cdot \max\{\ln n, 2^{i-2}\}}\right)^2\right) \\ \le \exp\left( -16\cdot c \cdot \ln n \left(\frac 1 2 - \frac 1 4\right)^2\right) = n^{-c}
    \end{multline*} 
    Which proves that with probability $p \ge 1-n^{-c}$, phase $i$ ends after sampling at most $8\cdot c \cdot \max\{\ln n, 2^{i-2}\} \frac{n}{2^{i-2}}$ edges.
\end{proof}

We have a symmetric result:

\begin{lemma}
    For any $c\ge 1$, any phase $i$ with  $i= 2\ceil{\log \frac n 2} -j$ for $1 \le j < \ceil{\log \frac n 2}$ needs at most $ 8\cdot c \cdot \max\{\ln n, 2^{j-2}\} \frac{n}{2^{j-2}}$ sampled edges to complete with probability $p \ge 1-n^{-c}$. 
\end{lemma}

\begin{proof}
    We first note that, by the stopping condition for phase $\ceil{\frac n 2 }$,  for any  phase $i$ with
    $i= 2\ceil{\log \frac n 2} -j$ for $1 \le j < \ceil{\log \frac n 2}$, it holds that
    $N_i \ge \ceil{\frac n 2}$. We first show that in phase $i$, the probability that a sampled edge decreases $S_i \setminus \Out_E(I_i)$ is at least $\frac {2^{j-2}} n$. Indeed, $\card{[n]\setminus (I_i \union \Out_E(I_i))} = \card{S_i \setminus \Out_E(I_i)} > 2^{j-1}$ (otherwise the phase would have ended) and, thus,  there are $N_i\cdot \card{[n]\setminus (I_i \union \Out_E(I_i))} \ge 2^{j-2} n$ edges that would decrease $S_i \setminus \Out_E(I_i)$. Thus, picking any of those edges, out of $n^2$ possible ones, has probability at least than $\frac {2^{j-2}} n$.

    The rest of the proof is symmetrical to the previous one.
\end{proof}

\begin{lemma}\label{lem:scheme2&1}
    For any $t \in \N$, $p \in [0,1]$, if Broadcast completes in scheme $2$ within $t$ rounds with probability at least $p$, then the same result holds for scheme $1$.
\end{lemma}

\begin{proof}
    We will couple schemes 1 and 2 as follows: One can see sampling without replacement of $m$ edges as sampling with replacement of as many edges as needed until the number of different edges sampled is equal to $m$. Indeed, to sample $m$ edges without replacement, one first must choose an edge uniformly at random, then another edge uniformly at random among the remaining edges, and so on until we have sampled $m$ edges. If we sample \emph{with} replacement until we have $m$ different edges, then whenever we have sampled $i \in [m-1]$ edges, the next new edge is chosen by sampling edges with replacement until a new edge is selected. This new edge is thus chosen uniformly at random among remaining edges.

    For each round $t$, we sample $m$ edges with replacement and call the resulting set of edges $E_t^{(2)}$. This is the set of edges for scheme $2$. To build $E_t{(1)}$, the set of edges for scheme $1$, we start with $E_t{(1)} = E_t^{(2)}$, and add to sampled edges with replacement until $\card{E_t^{(1)}} = m$. The set $E_t{(1)}$ sampled that way follows the distribution of $m$ sampled edges without replacement. We thus have that, with probability $1$, $E_t^{(2)} \subseteq E_t^{(1)}$.

   We now show that in each round $t$, we have that $\Proba(I^{(2)}_t \subseteq I^{(1)}_t) = 1$, where $I^{(i)}_t$ is the set of informed nodes in scheme $i$ after round $t$. Indeed, by induction, this is trivial for $t=0$. Let's assume it is true for some $t$. Then let $v \in I^{(2)}_{t+1}$. Then we either have $ v \in I^{(2)}_{t} \subseteq I^{(1)}_{t} \subseteq I^{(1)}_{t+1}$, or that $ v \notin I^{(2)}_{t}$, and thus there exist a node $u \in I^{(2)}_{t}$ such that edge $(u,v)\in E_{t+1}^{(2)}$. However, $u \in I^{(1)}_{t}$ by induction hypothesis, and $(u,v)\in E_{t+1}^{(2)}\subseteq E_{t+1}^{(1)}$. Therefore $v \in I^{(1)}_{t+1}$.

This proves that with probability $1$, Broadcast completes in scheme $1$ no later than it completes in scheme $2$, and thus the result holds.
\end{proof}

This allows us to prove the following result on scheme 2:

\thmerdos*
\begin{proof}
To see this, we are going to simulate scheme 3 with scheme 2. The main idea is that if a phase (of scheme 3) takes $x$ edges, sampled with replacement, to end with high probability, and in scheme $2$ each round samples $y$ edges with replacement, then in $\ceil{\frac x y}$ rounds, scheme $2$ samples at least $x$ edges, and, thus, we can simulate the phase in scheme $3$ with $\ceil{\frac x y}$ rounds of scheme 2. The only difference is that scheme 2 groups the edges into rounds to make intermediate progress, whereas scheme 3 only forwards the information at the end of the phase, all at once. This implies that in scheme 2 each phase is faster than the corresponding one in scheme 3, and any upper bound we get with this analysis will thus be an upper bound on the number of rounds scheme 2 needs to complete Broadcast.

    We first start with the phases $i \le \ceil{\log \frac n 2}$ such that $\ln n \le 2^{i-2}$ and the phases $i > \ceil{\log \frac n 2}$ with $ i = 2 \ceil{\log \frac n 2} - j$  for $j \ge 1 $ such that $\ln n \le 2^{j-2}$. 
    In that case, phase $i$ needs at most $8\cdot c \cdot n$ %
    sampled edges to end with probability larger than $1-n^{-c}$. We need at most $\ceil{\frac{8 \cdot c\cdot n}{m}} = \ceil{\frac{8 \cdot c}{m/n}}$ rounds to gather that many edges, and thus phase $i$ ends in $\ceil{\frac{8 \cdot c}{m/n}}$ rounds with probability greater than $1-n^{-c}$. There are at most $\ceil{\log n}$ such phases, and thus over all phases we require at most $O\left(\ceil{\frac{c}{m/n}} \log n\right)$.

    Let us now analyze the phases $i\le \ceil{\log \frac n 2}$ where $\ln n > 2^{i-2}$ and symmetrically phases $i> \ceil{\log \frac n 2},$ with $i=2\ceil{\log \frac n 2} -j$  for $j \ge 1 $ such that $\ln n > 2^{j-2}$. In that case, phase $i$ needs at most $8\cdot c \cdot \ln n \cdot \frac n {2^{i-2}}$ sampled edges to end with probability larger than $1-n^{-c}$. We need at most $\ceil{\frac{8\cdot c \cdot \ln n \cdot \frac n {2^{i-2}}}{m}} = \ceil{\frac{8\cdot c \cdot \ln n  }{m/n \cdot {2^{i-2}}}}$ rounds to gather that many edges, and thus phase $i$ ends in $\ceil{\frac{8\cdot c \cdot \ln n  }{m/n \cdot {2^{i-2}}}}$ rounds with probability greater than $1-n^{-c}$. 
   Summing the number of rounds over all such phases we get:

$$
   \sum_i \ceil{\frac{8\cdot c \cdot \ln n  }{m/n \cdot {2^{i-2}}}} \le  \sum_i (\frac{8\cdot c \cdot \ln n  }{m/n \cdot {2^{i-2}}}+1) \le \frac{32\cdot c \cdot \ln n  }{m/n }+\log n = O\left(\ceil{\frac{c}{m/n}} \log n\right)
$$

The probability of success $p\ge 1-n^{-c}\log n$ is simply a union bound on the number of phases.
\end{proof}

This result is particularly interesting if $m \le cn$ . We can also show the following result if $m \ge n\ln n$, which is more interesting in that particular case.

\thmerdostwo*

\begin{proof}
    We show that bound for scheme 2. With Lemma \ref{lem:scheme2&1}
    the bound for scheme 1 immediately follows.    Again, we introduce scheme 3, however, we modify it so that the goal of each phase is not to  multiply (respectively, divide) the number of informed (respectively, uninformed) nodes by 2, but instead, it is to multiply (respectively, divide) it by $(1+m/n)$. As a result, we get a total number of $O\left(\frac{\log n}{\log (1+m/n)}\right)$ phases.

    In this case, as formally discussed below, each phase necessitates $16\cdot c\cdot m$ edges to complete with high probability, but each round provides $m$ edges. Therefore each phase consists of $\ceil{16c}$ rounds.

    Formally, in phase $i \le \frac{\log (n/2)}{\log (1+m/n)}$, we start with $(1+m/n)^{i-1}$ informed nodes, and  at least  $\frac n 2 \cdot (1+m/n)^{i-1}$ edges out of the $n^2$ potential edges can inform an uninformed node. Thus, each sampled edge has probability at least $\frac 1 {2n} \cdot (1+m/n)^{i-1}$ of informing an uninformed node. Hence, by the same argument as in Lemma~\ref{lem:phasesbound}, we need to sample $\frac{2n}{(1+m/n)^{i-1}}$ edges to inform a new node with probability at least $\frac 1 2$. 

    To go from $(1+m/n)^{i-1}$ informed nodes to $(1+m/n)^{i}$ informed nodes, we need to inform $m/n(1+m/n)^{i-1}$ uninformed nodes. If we sample $8c\cdot m/n (1+m/n)^{i-1}$ buckets of $\frac{2n}{(1+m/n)^{i-1}}$ edges  each(for a total of $16c\cdot m$ edges), we get by Hoeffding's inequality that the probability that not enough edges inform a new node is:
\small
    \begin{multline*}        
    \Proba(\card{I_i \union \Out_E(I_i)} \le (1+m/n)^{i}-1) \le \exp\left( - 2 \cdot 8 \cdot c \cdot m/n (1+m/n)^{i-1} \left(\frac 1 2 - \frac{m/n (1+m/n)^{i-1}} {8 \cdot c \cdot m/n (1+m/n)^{i-1}}\right)^2\right) \\ \le \exp\left( -16\cdot c \cdot \ln n \left(\frac 1 2 - \frac 1 4\right)^2\right) = n^{-c}
    \end{multline*}
    \normalsize
\end{proof}

We also show a lower bound:

\erdoslowerbound*

\begin{proof}
    We will show by induction that, in scheme 1, $\expect(\card{I_t}) \le (1+m/n) ^ {t}$ for every $t \in \N$. We will then apply Markov's inequality to conclude.

    Let us first compute $\expect\left(\card{I_t}\big|\card{I_{t-1}} =x\right)$. Let $v \notin I_{t-1}$ be an uninformed node, and let $e$ be an incoming edge to $v$ such that its tail is in $I_{t-1}$. Then the probability that this edge is picked is $\frac{m}{n^2}:=\rho$. By a union bound over the $x$ edges $(u,v)$ such that $u\in I_{t-1}$, we have that $\Proba(v \in I_t \big| \card{I_{t+1}} = x) \le \rho x$. %
    Denoting $X_v$ the variable $v \in I_t$, we then have that $\expect(X_v\big| \card{I_{t+1}} = x) \le \rho x$.%

    Summing that expectation over all $v \in [n]\setminus I_{t-1}$%
    , we have that:

\small
    \begin{equation}\label{eq:expectation}
        \expect\left(\card{I_t}\big|\card{I_{t-1}} =x\right) = x+\sum_{v \in [n]-I_{t-1}} \expect (X_v) = x+\sum_{v \in [n]-I_{t-1}} \rho x \le x+(n-x) \rho x \le x(1+m/n)
    \end{equation}
    \normalsize

    We can now prove our claim by induction. For the induction basis, we clearly have $\expect(I_0)=1=(1+m/n) ^0$. For the induction step, assume that for some $t \ge 1$, we have that $\expect(I_{t-1}) \le (1+m/n) ^ {t-1}$. Then:
    
    $$
    \expect\left(\card{I_t}\right) = \expect\left(\expect(\card{I_t} \big| \card{I_{t-1}})\right) \le  \expect \left(\card{I_{t-1}} (1+m/n)\right) \le (1+m/n)^t
    $$ 
    
    Where the first inequality holds by Equation~\ref{eq:expectation} and the second inequality holds by the induction hypothesis.

    Therefore, we have that $\expect\left(\card{I_{\frac{\log (n) -1}{\log(1+m/n)}}}\right) \le \frac n 2$. Using Markov's inequality, we then have that:

    $$
    \Proba\left( \card{I_{\frac{\log (n) -1}{\log(1+m/n)}}} \ge n \right) \le \frac{\expect\left(\card{I_{\frac{\log (n) -1}{\log(1+m/n)}}}\right)} {n} \le \frac 1 2
    $$%
\end{proof}

Using a union-bound in the same way as for the uniformly random trees model, we have a result on All-to-All Broadcast:

\begin{theorem}
For any $c \ge 1$, $n\ge 5$ $m \in [n^2]$, All-to-All Broadcast on directed Erdős–Rényi graphs completes within $O\left(\ceil{\frac c {m/n}} \log n\right)$ rounds with probability $p>1- n^{c-1}\log n$. Moreover, if $m/n \ge \ln n$, All-to-All Broadcast on directed Erdős–Rényi graphs completes within $O\left(\frac{c\cdot\log n}{\log (1+m/n)}\right)$ rounds with probability $p>1- n^{c-1}\log n$.
\end{theorem}

Finally, using Algorithm~\ref{alg:Consensus} for Consensus, we have:

\begin{theorem}
For any $c \ge 1$, $n\ge 5$ $ m \in [n^2]$, there exists a protocol for Consensus on directed Erdős–Rényi graphs that satisfies Agreement and Validity, terminates within $O\left(\ceil{\frac c {m/n}} \log n\right)$ rounds with probability $p>1-\frac 1 {n^c}$, and only requires messages of 1 bit over each edge in each round.  Moreover, if $m/n \ge \ln n$, then we get the better bound $O\left(\frac{c\cdot\log n}{\log (1+m/n)}\right)$ for the number of rounds, with the same probability of success.
\end{theorem}

\subsection{Byzantine Nodes in directed Erdős–Rényi graphs}

We now analyze what happens if some nodes deviate arbitrarily from the protocol. More specifically, we allow up to $f < \frac {2n} 3$ nodes, the \emph{Byzantine nodes}, to coordinate to delay Broadcast as much as possible.  Moreover, we give every node access to a cryptographic tools, so that nodes can sign messages, and ensure any message they receive, even if forwarded, has been sent ``as is'' from the not who signed the message. As in Section~\ref{sec:Byzantine}, the best strategy Byzantine nodes can thus have is to stop forwarding messages. To analyze the problem, we consider the three schemes as above:

In \emph{scheme 1}, in each round, $m$ edges are chosen uniformly at random \emph{without} replacement among the $n^2$ possible edges. This is equivalent to our model.

In \emph{scheme 2}, in each round, $m$ edges are chosen uniformly at random \emph{with} replacement among the $n^2$ possible edges. This can only result in more rounds than scheme 1  as fewer disjoint edges are chosen in each round compared to scheme 1.

In \emph{scheme 3}, we start with a unique informed node $1$. We then run $\ceil{\log \frac {n-f} 2 }$ phases. In phase $i$, we have $N_i=2^{i-1}$ honest informed nodes $I_i$. We set $E=\varnothing$ and add one edge at a time, sampled with replacement, to $E$  until $\card {I_i \union \Out_E(I_i)} = \min \{2 ^i, \ceil{\frac {n-f} 2}\}$. We then set $I_{i+1}=I_i \union \Out_E(I_i)$.

We then run $\ceil{\log \frac {n-f} 2 }$ other phases. In phase $i=2\ceil{\log \frac {n-f} 2 }-j$, we have $N_i$ honest informed nodes $I_i$, and honest uninformed nodes $S_i$, with $\card{S_i} = \min\{2^{j}, \floor{\frac {n-f} 2}\}$. We set $E=\varnothing$ and add to $E$, one edge at a time, sampled with replacement until $\card {S_i \setminus \Out_E(I_i)} =  2^{j-1}$. We then set $I_{i+1}=I_i \union \Out_E(I_i)$. 

As above, we start by analyzing scheme 3:

\begin{lemma}
    For any $c\ge 1$, any phase $i \le \ceil{\log \frac {n-f} 2}$ needs at most $ 24\cdot c \cdot \max\{\ln n, 2^{i-2}\} \frac{n}{2^{i-2}}$ sampled edges to complete with probability $p \ge 1-n^{-c}$. 
\end{lemma}

\begin{proof}
    We first note that in phase $i$, the probability that an edge being sampled makes $I_i \union \Out_E(I_i)$ larger is at least $\frac {2^{i-2}} {3n}$. Indeed, $\card{I_i \union \Out_E(I_i)} \le \frac {n-f} 2$ (otherwise the phase would have ended) and there are $\card{I_i}\cdot \card{[n]\setminus (I_i \union \Out_E(I_i))} \ge 2^{i-2} ({n-f})$ edges that would make $I_i \union \Out_E(I_i)$ larger. Therefore picking any of those edges, out of $n^2$ possible ones, has probability at least $\frac {2^{i-2}(n-f)} {n^2} \ge \frac {2^{i-2}}{3n}$.

    Next, we remark that if we sample $\frac {3n} {2^{i-2}}$ edges, then the probability that at least one of those edges makes $I_i \union \Out_E(I_i)$ larger is at least $\frac 1 2$. Indeed, the probability that all of those edges do not make $I_i \union \Out_E(I_i)$ larger is $(1-\frac {2^{i-2}} {3n})^{\frac {3n} {2^{i-2}}} \le e^{-1} \le \frac 1 2$. We will, thus, group the edges into ``buckets'' of $\frac {3n} {2^{i-2}}$ edges.

    We then use the fact that we only need $2^{i-1}$ edges that make $I_i \union \Out_E(I_i)$ larger to end phase $i$. If we take $8\cdot c \cdot \max\{\ln n, 2^{i-2}\}$ buckets of $\frac {3n} {2^{i-2}}$ edges, then applying Heoffding's inequality, we have:
    \begin{multline*}
        \Proba(\card{I_i \union \Out_E(I_i)} < 2^{i}) \le \exp\left( - 2 \cdot 8 \cdot c \cdot \max\{\ln n, 2^{i-2}\} \left(\frac 1 2 - \frac{2^{i-1}} {8 \cdot c \cdot \max\{\ln n, 2^{i-2}\}}\right)^2\right) \\ \le \exp\left( -16\cdot c \cdot \ln n \left(\frac 1 2 - \frac 1 4\right)^2\right) = n^{-c}
    \end{multline*}
    Which proves that with probability $p \ge 1-n^{-c}$, phase $i$ ends after sampling at most $24\cdot c \cdot \max\{\ln n, 2^{i-2}\} \frac{n}{2^{i-2}}$ edges.
\end{proof}

We have a symmetric result:

\begin{lemma}
    For any $c\ge 1$, any phase $i > \ceil{\log \frac {n-f} 2}, i= 2\ceil{\log \frac {n-f} 2} -j$ needs at most $ 24\cdot c \cdot \max\{\ln n, 2^{j-2}\} \frac{n}{2^{j-2}}$ samplings to complete with probability $p \ge 1-n^{-c}$. 
\end{lemma}%

\begin{proof}
    We first note that in phase $i$, the probability that an edge being sampled makes $S_i \setminus \Out_E(I_i)$ smaller is at least $\frac {2^{j-2}} {3n}$. Indeed, $\card{S_i \setminus \Out_E(I_i)} \ge 2^{j-1}$ (otherwise the phase would have ended) and there are $\card{I_i}\cdot \card{[n]\setminus (I_i \union \Out_E(I_i))} \ge 2^{j-2} (n-f)$ edges that would make $S_i \setminus \Out_E(I_i)$ smaller. Therefore picking any of those edges, out of $n^2$ possible ones, has probability at least $\frac {2^{j-2}(n-f)} {n^2} \ge \frac {2^{j-2}}{3n}$.

    The rest of the proof is symmetrical to the previous one.
\end{proof}

This allows us to prove the following result on scheme 2:

\begin{theorem}
    For any $c \ge 1$, in scheme 2, and therefore scheme 1, Broadcast completes within $O\left(\ceil{\frac{c}{m/n}} \log n\right)$ rounds with probability $p \ge 1 -  n^{-c}\log n$.
\end{theorem}

\begin{proof}
To see this, we are going to simulate scheme 3 with scheme 2. The main idea is that if a phase (of scheme 3) takes $x$ number of edges to end with high probability, and in scheme $2$ each rounds samples $y$ edges without replacement, then in $\ceil{\frac x y}$ rounds, scheme $2$ samples more than $x$ edge, and thus we can simulate the phase in scheme $3$ with $\ceil{\frac x y}$ rounds of scheme 2. The only difference is that scheme 2 groups the edges per round to make intermediate progress, whereas scheme 3 only forwards the information at the end of the phase, all at once. This is only beneficial to scheme 2, and any upper bound we get with this analysis will thus be an upper bound on the number of rounds scheme 2 needs to complete Broadcast.

    We first start with the phases $i \le \ceil{\log \frac {n-f} 2}$ %
    such that $\ln n \le 2^{i-2}$ and the phases $i > \ceil{\log \frac n 2}, i = 2 \ceil{\log \frac n 2} - j$ such that $\ln n \le 2^{i-1}$. 
    In that case, phase $i$ needs at most $24\cdot c \cdot n$ sampled edges to end with probability larger than $1-n^{-c}$. We need at most $\ceil{\frac{24 \cdot c\cdot n}{m}} = \ceil{\frac{24 \cdot c}{m/n}}$ rounds to gather that many edges, and thus phase $i$ ends in $\ceil{\frac{24 \cdot c}{m/n}}$ rounds with probability greater than $1-n^{-c}$. There are at most $\ceil{\log n}$ such phases, and thus over all phases we require at most $O\left(\ceil{\frac{c}{m/n}} \log n\right)$.

    Let us now analyze the phases $i\le \ceil{\log \frac {n-f} 2}$ %
    where $\ln n > 2^{i-2}$ (And symmetrically phases $i> \ceil{\log \frac n 2}, i=2\ceil{\log \frac n 2} -j$ where $\ln n > 2^{j-2}$). In that case, phase $i$ needs at most $24\cdot c \cdot \ln n \cdot \frac n {2^{i-2}}$ sampled edges to end with probability larger than $1-n^{-c}$. We need at most $\ceil{\frac{24\cdot c \cdot \ln n \cdot \frac n {2^{i-2}}}{m}} = \ceil{\frac{24\cdot c \cdot \ln n  }{m/n \cdot {2^{i-2}}}}$ rounds to gather that many edges, and thus phase $i$ ends in $\ceil{\frac{24\cdot c \cdot \ln n  }{m/n \cdot {2^{i-2}}}}$ rounds with probability greater than $1-n^{-c}$. 
   Summing the number of rounds over all such phases we get:

$$
   \sum_i \ceil{\frac{24\cdot c \cdot \ln n  }{m/n \cdot {2^{i-2}}}} \le  \sum_i \frac{24\cdot c \cdot \ln n  }{m/n \cdot {2^{i-2}}}+1 \le \frac{72\cdot c \cdot \ln n  }{m/n }+\log n = O\left(\ceil{\frac{c}{m/n}} \log n\right)
$$

The probability of success $p\ge 1-n^{-c}\log n$ is simply a union bound on the number of phases.
\end{proof}

We can expand this result to all-to-all Broadcast, using a simple union-bound:

\begin{corollary}\label{sec.8:thm:alltoall}
    For any $c \ge 1$, in scheme 2, and therefore scheme 1, All-to-All Broadcast completes within $O\left(\ceil{\frac{c}{m/n}} \log n\right)$ rounds with probability $p \ge 1 -  n^{-c+1}\log n$.
\end{corollary}

\begin{theorem}\label{sec.8.thm:reduction}
    Let $\Alg$ be a distributed synchronous algorithm that runs on a static clique in $T$ time, where $T \le \alpha n^x$ for some constant $\alpha\in \R_+, x\in\N$, and has a probability of success $p$. Assume $\Alg$ is robust to $f$ Byzantine nodes, and $f < \frac 2 3 n $.
    Then, assuming cryptographic tools that allow nodes to sign and encrypt messages, there exists a distributed algorithm $\Alg'$ that runs on directed Erdős–Rényi graphs with $m$ edges in $O\left(T\ceil{\frac{c}{m/n}} \log n\right)$ time, and has a probability of success $p' \ge p(1-\alpha n^{1+x-c} \log n )$, for any $c\ge 1+x$. Moreover, $\Alg'$ is robust to $f$ Byzantine nodes.
\end{theorem}

\begin{proof}
    The algorithm $\Alg'$ works as follows: Each round of $\Alg$ will be simulated using $O\left(\ceil{\frac{c}{m/n}} \log n\right)$ rounds of directed Erdős–Rényi graphs. 
    Each round in $\Alg$ can be seen as (1) a \emph{computation step}, where each node decides what message to send to every other node, and (2) a \emph{communication step}, where each node sends the messages and receives the messages the other nodes have sent it. 
    In $\Alg'$ the computation step is unchanged, except that each node uses the PKI to sign and encrypt the messages. 
    The communication step on the other hand is extended over $O\left(\ceil{\frac{c}{m/n}} \log n\right)$ rounds, in which every (honest) node simply forwards all received messages to all its out-neighbors. 
    Theorem~\ref{sec.8:thm:alltoall} ensures that with probability at most $n^{1-c}\log n$, at least one honest node fails to send a message to all other honest nodes. 
    By a union bound over the $T$ rounds, the probability that at least one message fails to be delivered is at most $\alpha n^{1+x-c}\log n$. By taking its complement, the probability that all messages are delivered in time before the next round's computation step is larger than $ 1-\alpha n^{x+1-c}\log n$. The probability that $\Alg'$ succeeds is then $p' \ge p (1-\alpha n^{x+1-c}\log n)$
\end{proof}

We now give two applications of this theorem, namely Reliable Broadcast and Byzantine Consensus.

\begin{theorem}
    For any $c\ge 1$, and $f \le \frac 2 3 n -1$, in the directed Erdős–Rényi graphs model, there exists an algorithm for Reliable Broadcast, that is robust to $f$ Byzantine nodes, that runs in $O\left((f+1)\ceil{\frac{c}{m/n}} \log n\right)$ rounds, and succeeds with probability $p \ge 1-n^{2-c} \log n$.
\end{theorem}

\begin{proof}
    Dolev and Strong~\cite{reliable} have given an algorithm that solves reliable Broadcast, is robust to $f$ Byzantine nodes, and runs in $T=f+1$ rounds. Since $T\le n$, we can apply Theorem~\ref{sec.8.thm:reduction} with $x=1, \alpha = 1$, and we get the desired result.
\end{proof}

\begin{theorem}
    For any $c\ge 1$, in the directed Erdős–Rényi graphs model, there exists an algorithm for Byzantine Consensus, that is robust to $f$ Byzantine nodes as long as $f < \frac n 3 $, that runs in $O\left((f+1)\ceil{\frac{c}{m/n}} \log n\right)$ rounds, and succeeds with probability $p \ge 1-2n^{2-c}\log n$.
\end{theorem}

\begin{proof}
   Berman, Garay and Perry~\cite{king} have given an algorithm (known as the King's algorithm) that solves reliable Broadcast, is robust to $f$ Byzantine nodes, and runs in $T=3(f+1)$ rounds. Since $T\le 2n$, we can apply Theorem~\ref{sec.8.thm:reduction} with $x=1, \alpha = 2$, and we get the desired result.
\end{proof}

\subsection{Adversarial Edges in directed Erdős–Rényi graphs}

In this section, we will study the case where in each round, an adversary chooses $k$ edges in each round, then $m-k$ edges are chosen among the remaining edges. We restrict $k$ to be smaller than $\frac 3 4 n^2$, so that the adversary is not forced to choose an edge from an informed node to an uninformed one.
In fact, in that case, the edges chosen by the adversary do not matter (as long as she doesn't choose an edge from an informed node to an uninformed one), as the edges chosen by the adversary cannot ``protect'' uninformed nodes as in the case of the trees. We will, thus, in the rest of this section, simply assume that $k$ edges have been removed from the pool of possible edges, none of them being an edge from an informed node to an uninformed node.

As before, we will analyze three different schemes.

In \emph{scheme 1}, in each round, $m-k$ edges are chosen uniformly at random \emph{without} replacement among the $n^2-k$ possible edges. 

In \emph{scheme 2}, in each round, $m-k$ edges are chosen uniformly at random \emph{with} replacement among the $n^2-k$ possible edges. This can only result in more rounds than scheme 1 as fewer disjoint edges are chosen in each round compared to scheme 1.

In \emph{scheme 3}, we start with a unique informed node, let say node $1$. 
There are two types of phases for a total of $2\ceil{\log \frac n 2 }$ phases. In each phase $i$ with $1 \le i \le \ceil{\log \frac n 2 }$%
, we have $N_i=2^{i-1}$ informed nodes in $I_i$ at the beginning of the phase and the goal is to double that number within the phase. We set $E=\varnothing$ and add to $E$ one edge at a time, sampled with replacement, until $\card {I_i \union \Out_E(I_i)} = \min \{2 ^i, \ceil{\frac n 2}\}$. We then set $I_{i+1}=I_i \union \Out_E(I_i)$.
Note that at the end of phase $i = \ceil{\log \frac n 2}$, $N_{i+1} = \ceil{\frac n 2}$ and
this scheme is independent of $m$.

Then in each phase $i=2\ceil{\log \frac n 2 }-j$ with $0 \le j < \ceil{\log \frac n 2 }$, we initially have $N_i$ informed nodes in $I_i$, and  $\min\{2^{j}, \floor{\frac n 2}\} $ uninformed nodes in $S_i$ and the goal is to half the number of uninformed nodes in each phase. We set $E=\varnothing$ and add to $E$ one edge at a time sampled with replacement, until $\card {S_i \setminus \Out_E(I_i)} =  2^{j-1}$. We then set $I_{i+1}=I_i \union \Out_E(I_i)$. %

We start by analyzing scheme 3:

\begin{lemma}
    For any $c\ge 1$, any phase $i \le \ceil{\log \frac n 2}$ needs at most $ 8\cdot c \cdot \max\{\ln n, 2^{i-2}\} \frac{{(n^2-k)}}{2^{i-1}n}$ sampled edges to complete with probability $p \ge 1-n^{-c}$. 
\end{lemma}

\begin{proof}
    We first note that in phase $i$ with  $i \le \ceil{\log \frac n 2}$, the probability that an edge that is sampled increases $I_i \union \Out_E(I_i)$  is at least $\frac {2^{i-2}n} {(n^2-k)}$. Indeed, $\card{I_i \union \Out_E(I_i)} \le \frac n 2$ (otherwise the phase would have ended) and there are $\card{I_i}\cdot \card{[n]\setminus (I_i \union \Out_E(I_i))} \ge 2^{i-2} n$ edges that can increase $I_i \union \Out_E(I_i)$. Thus, picking any of those edges, out of $(n^2-k)$ possible ones, has probability at least $\frac {2^{i-2}n}  {(n^2-k)}$.

    Next, we remark that if we sample $\frac { {(n^2-k)}} {2^{i-2}n}$ edges, then the probability that at least one of those edges increases $I_i \union \Out_E(I_i)$ is at least $\frac 1 2$. Indeed, the probability that all of those edges do not increase the size of $I_i \union \Out_E(I_i)$  is $(1-\frac {2^{i-2}n}  {(n^2-k)})^{\frac  {(n^2-k)} {2^{i-2}n}} \le e^{-1} \le \frac 1 2$. We will, thus, group the sampled edges into disjoint ``buckets'' of $\frac  {(n^2-k)} {2^{i-2}n}$ consecutively sampled edges.

    We then use the fact that we only need $2^{i-1}$ edges that increase $I_i \union \Out_E(I_i)$  to end phase $i$. If we take $8\cdot c \cdot \max\{\ln n, 2^{i-2}\}$ buckets of $\frac  {(n^2-k)} {2^{i-2}n}$ edges each, then applying Hoeffding's inequality, we have:
    \begin{multline*}
        \Proba(\card{I_i \union \Out_E(I_i)} \le 2^{i}-1) \le \exp\left( - 2 \cdot 8 \cdot c \cdot \max\{\ln n, 2^{i-2}\} \left(\frac 1 2 - \frac{2^{i-1}} {8 \cdot c \cdot \max\{\ln n, 2^{i-2}\}}\right)^2\right) \\ \le \exp\left( -16\cdot c \cdot \ln n \left(\frac 1 2 - \frac 1 4\right)^2\right) = n^{-c}
    \end{multline*}
    which proves that with probability $p \ge 1-n^{-c}$, phase $i$ ends after sampling at most $8\cdot c \cdot \max\{\ln n, 2^{i-2}\} \frac{ {(n^2-k)}}{2^{i-2}n}$ edges.
\end{proof}

We have a symmetric result for the second phase:

\begin{lemma}
    For any $c\ge 1$, any phase $i$ with  $i= 2\ceil{\log \frac n 2} -j$ for $0 \le j < \ceil{\log \frac n 2}$ needs at most $ 8\cdot c \cdot \max\{\ln n, 2^{j-2}\} \frac{ {(n^2-k)}}{2^{j-1}n}$ sampled edges to complete with probability $p \ge 1-n^{-c}$. 
\end{lemma}

\begin{proof}
    We first note that, by the stopping condition for phase $\ceil{\frac n 2 }$,  for any  phase $i$ with
    $i= 2\ceil{\log \frac n 2} -j$ for $0 \le j < \ceil{\log \frac n 2}$, it holds that
    $N_i \ge \ceil{\frac n 2}$. We first show that in phase $i$, the probability that a sampled edge decreases $S_i \setminus \Out_E(I_i)$ is at least $\frac {2^{j-2}n}  {(n^2-k)}$. Indeed, $\card{[n]\setminus (I_i \union \Out_E(I_i))} = \card{S_i \setminus \Out_E(I_i)} > 2^{j-1}$ (otherwise the phase would have ended) and, thus,  there are $N_i\cdot \card{[n]\setminus (I_i \union \Out_E(I_i))} \ge 2^{j-2} n$ edges that would decrease $S_i \setminus \Out_E(I_i)$. Thus, picking any of those edges, out of $ {(n^2-k)}$ possible ones, has probability at least  $\frac {2^{j-2}n}  {(n^2-k)}$.

    The rest of the proof is symmetrical to the previous one.
\end{proof}

This allows us to prove the following result on scheme 2:

\begin{theorem}
    For any $c \ge 1$, in scheme 2, and therefore scheme 1, Broadcast completes within $O\left(\ceil{\frac{c \cdot (n^2-k)}{(m-k)n}} \log n\right)$ rounds with probability $p \ge 1 -  n^{-c}\log n$.
\end{theorem} %

\begin{proof}
To see this, we are going to simulate scheme 3 with scheme 2. The main idea is a follows: If a phase (of scheme 3) requires $x$ edges, sampled with replacement, in order to end, and in scheme $2$ each round samples $y$ edges with replacement, then in $\ceil{\frac x y}$ rounds, scheme $2$ samples at least $x$ edges, and, thus, we can simulate the phase in scheme $3$ with $\ceil{\frac x y}$ rounds of scheme 2.
 The only difference is that scheme 2 groups the edges into rounds to make intermediate progress, whereas scheme 3 only forwards the information at the end of the phase, all at once. This implies that in scheme 2 each phase is faster than the corresponding one in scheme 3, and any upper bound we get with this analysis will thus be an upper bound on the number of rounds scheme 2 needs to complete Broadcast.

    We first start with the phases $i \le \ceil{\log \frac n 2}$ such that $\ln n \le 2^{i-2}$ and the phases $i > \ceil{\log \frac n 2}$ with $ i = 2 \ceil{\log \frac n 2} - j$  for $j \ge 1 $ such that $\ln n \le 2^{j-2}$. 
    In that case, phase $i$ needs at most $8\cdot \frac c n \cdot  {(n^2-k)}$ %
    sampled edges to end with probability larger than $1-n^{-c}$. We need at most $\ceil{\frac{8 \cdot c\cdot  {(n^2-k)}}{n(m-k)}}$ rounds to gather that many edges, and thus phase $i$ ends in $\ceil{\frac{8 \cdot c\cdot  {(n^2-k)}}{n(m-k)}}$ rounds with probability greater than $1-n^{-c}$. There are at most $\ceil{\log n}$ such phases, and thus over all phases we require at most $O\left(\ceil{\frac{c \cdot (n^2-k)}{(m-k) n}} \log n\right)$.

    Let us now analyze the phases $i\le \ceil{\log \frac n 2}$ where $\ln n > 2^{i-2}$ and symmetrically phases $i> \ceil{\log \frac n 2}, i=2\ceil{\log \frac n 2} -j$  for $j \ge 1 $ such that $\ln n > 2^{j-2}$. In that case, phase $i$ needs at most $8\cdot c \cdot \ln n \cdot \frac {n^2-k} {2^{i-2}n}$ sampled edges to end with probability larger than $1-n^{-c}$. We need at most $\ceil{\frac{8\cdot c \cdot \ln n \cdot \frac {(n^2-k)} {2^{i-2}n}}{m-k}} = \ceil{\frac{8\cdot c \cdot \ln n \cdot (n^2-k)}{(m-k){2^{i-2}n}}}$ rounds to gather that many edges, and thus phase $i$ ends in $\ceil{\frac{8\cdot c \cdot \ln n \cdot (n^2-k)}{(m-k){2^{i-2}n}}}$ rounds with probability greater than  $1-n^{-c}$. 
   Summing the number of rounds over all such phases we get:

\begin{multline*}
   \sum_i \ceil{\frac{8\cdot c \cdot \ln n \cdot (n^2-k)}{(m-k){2^{i-2}n}}} \le  \sum_i \left(\frac{8\cdot c \cdot \ln n \cdot (n^2-k)}{(m-k){2^{i-2}n}}+1\right) \\ \le \frac{32\cdot c \cdot \ln n \cdot (n^2-k)}{(m-k){n}}+\log n = O\left(\ceil{\frac{c \cdot (n^2-k)}{(m-k)n}} \log n\right)
\end{multline*}

The probability of success $p\ge 1-n^{-c}\log n$ is simply a union bound on the number of phases.
\end{proof}

\end{document}